\documentclass[aps,pra,twocolumn,superscriptaddress,longbibliography]{revtex4-1}

\usepackage{tikz}
\usepackage{algorithm}
\usepackage{booktabs} 
\usepackage{algorithmic}
\usepackage{amsmath, amssymb}
\usepackage{quantikz}
\usepackage{comment}
\usetikzlibrary{3d,arrows.meta,decorations.pathreplacing}
\usetikzlibrary{angles,quotes}
\usepackage{pgfplots}
\pgfplotsset{compat=1.18}
\usepackage{natbib}
\usepackage[breaklinks]{hyperref}
\usepackage{amssymb}
\usepackage{amsmath,amsthm,bm}
\usepackage{amsfonts,dsfont}
\usepackage{graphicx}
\usepackage{lipsum}
\usepackage{subcaption}
\usepackage{mathtools}
\usepackage{tikz}
\usepackage{cleveref}
\usetikzlibrary{angles, quotes}
\usetikzlibrary{positioning}
\newcommand{\expect}[1]{\langle #1 \rangle}
\newcommand{\Hh}{\hat H}
\usepackage[normalem]{ulem}
\usepackage{color}
\usepackage{xcolor}
\newcommand{\eps}{\varepsilon}

\newtheorem{proposition}{Proposition}

\newtheorem{corollary}{Corollary}
\newtheorem{remark}{Remark}

\newcommand{\Hb}{\bar{H}}

\newcommand{\Tr}{\mbox{Tr}}

\newtheorem{thm}{Theorem}
\newtheorem{lema}[thm]{Lemma}

\makeatletter
\long\def\@makecaption#1#2{%
  \par
  \vskip\abovecaptionskip
  \begingroup
    \small\rmfamily
    \samepage
    \flushing
    \let\footnote\@footnotemark@gobble
    \@make@capt@title{#1}{#2}\par
  \endgroup
  \vskip\belowcaptionskip
}
\makeatother

\begin{document}

\title{Augmenting Imaginary-Time Evolution with Local Geometric Information}

\author{Carlos L. Benavides-Riveros}
\email{carlos.benavides@iqm.tech}
\author{Prachi Sharma}
\author{Pietropaolo Frisoni}
\author{Fedor  Šimkovic IV}
\email{fedor.simkovic@iqm.tech}
\affiliation{IQM Quantum Computers, Georg-Brauchle-Ring 23-25, 80992 Munich, Germany}

\date{\today}
%TC:ignore
\begin{abstract}
Imaginary-time evolution (ITE) underpins a broad family of algorithms for ground-state preparation in quantum simulation and quantum many-body physics. In these methods, convergence is governed by the energy variance of the instantaneous state, causing the flow to approach the ground state only asymptotically. We introduce an augmented imaginary-time evolution (AITE) framework that replaces the standard gradient flow on the energy landscape with a geometrically informed descent along locally optimal directions, which are identified by exploiting the higher-order statistical structure of the instantaneous energy distribution. The resulting flow strictly outperforms standard ITE throughout the entire evolution and exhibits two qualitatively distinct regimes: a superlinear convergence regime, followed by an extinction regime in which the energy error vanishes exactly at a finite value of the imaginary-time flow parameter, in contrast to the asymptotic exponential decay of ITE. Standard ITE is recovered in the zero-skewness limit of AITE, implying that the acceleration extends naturally across the broader ITE algorithmic family.
\end{abstract}
%TC:endignore
\maketitle

\section*{Introduction}

Imaginary-time evolution (ITE) is one of the central ideas behind modern approaches to low-energy many-body phy\-sics. Under the Wick rotation $t \rightarrow -i\tau$, {\color {black} where the real parameter $\tau\ge0$ is the imaginary time}, coherent dynamics are replaced by a non-unitary flow that exponentially suppresses excited-state components and drives the system toward its ground state~\cite{PhysRev.96.1124}. This simple projection mechanism underlies a wide range of methods for ground-state preparation, thermal calculations, and spectral estimation in condensed-matter physics, quantum chemistry, and quantum field theory~\cite{Motta2022, Cao2019,10.1063/1.4916647}, and sits within a broader landscape of relaxation, filtering, and annealing methods that reach well beyond physics and chemistry~\cite{Kirkpatrick1983}.

While employing the same filtering principle, algorithms rooted in ITE take several distinct computational forms. In direct implementations, one approximates the non-unitary propagator $\exp[-\tau \hat H]$, where $\hat{H}$ is the Ha\-mil\-tonian whose ground state is sought, iteratively, so that low-energy structure is progressively revealed by explicit imaginary-time cooling. In many-body settings, this includes product-formula propagation and tensor-network schemes, with block decimation and imaginary-time variational principles as natural matrix-product-state realizations \cite{KOSLOFF1986223, PhysRevLett.93.207204,PhysRevA.109.052430,gtq3-j37b}. Alternatively, stochastic projector methods realize the same spectral filtering statistically, through walker population dynamics or sampled paths whose branching and reweighting drive the dynamics toward low energy, as in diffusion Monte Carlo, auxiliary-field quantum Monte Carlo, and related projector approaches \cite{PhysRevLett.45.566, 10.1063/1.431514,Zhang2018,PhysRevLett.90.136401}. In yet another approach, variational formulations restrict the evolution to a tractable manifold of trial states, replacing exact propagation by projected descent within an ansatz state ma\-ni\-fold~\cite{PRXQuantum.2.010342, 10.21468/SciPostPhys.15.6.229, zima2026,g4ch-5x8m}. Mo\-re\-over, at the classical level, simulated annealing distills the same principle into a cooling schedule for combinatorial optimization \cite{Kirkpatrick1983}, with applications ranging from circuit design to portfolio optimization \cite{CRAMA2003546, PhysRevLett.130.050601, pub.1195650887, 10367741}. 

Since the ITE propagator is non-unitary, it cannot be straightforwardly implemented as a quantum circuit, and this has led to several distinct quantum realizations of the same ITE objective. Some approaches stay as close as possible to the original flow, approximating short imaginary-time steps through implementable unitary updates or related hybrid constructions \cite{Motta2020}. In contrast, others impose the evolution variationally in the form of a parameterized circuit \cite{McArdle2019,AmeeraAnuar_2026}. Most recently, two wider viewpoints have become explicit. In one, the non-unitary propagator is treated as a particular instance of spectral filtering, placing ITE within a larger family of energy-selective transformations for low-energy state preparation \cite{D4FD00039K,Cianci2024,Yuan2019theoryofvariational}. In the other, cooling is reformulated in terms of structured flow equations, with double-bracket dynamics providing an alternative route to monotonic energy descent \cite{alghadeer2025}. 

The geometric structure of ITE has also attracted growing attention. Brockett's double-bracket flow established ITE as an isospectral gradient flow on the manifold of density operators \cite{BROCKETT199179}, and the connection between ITE and gradient flow on the Riemannian manifold of quantum states (encoded in the quantum geometric tensor) has provided a natural language for understanding variational implementations and their convergence properties \cite{Stokes2020quantumnatural}.  This geometric perspective has spurred an active line of research, establishing fidelity bounds for ground-state preparation and energy minimization, among other results \cite{Stokes_2023, hartung2025,mcmahon2025equatingquantumimaginarytime}.

\begin{figure*}[th!]
\centering
{\includegraphics[width=0.9\textwidth]{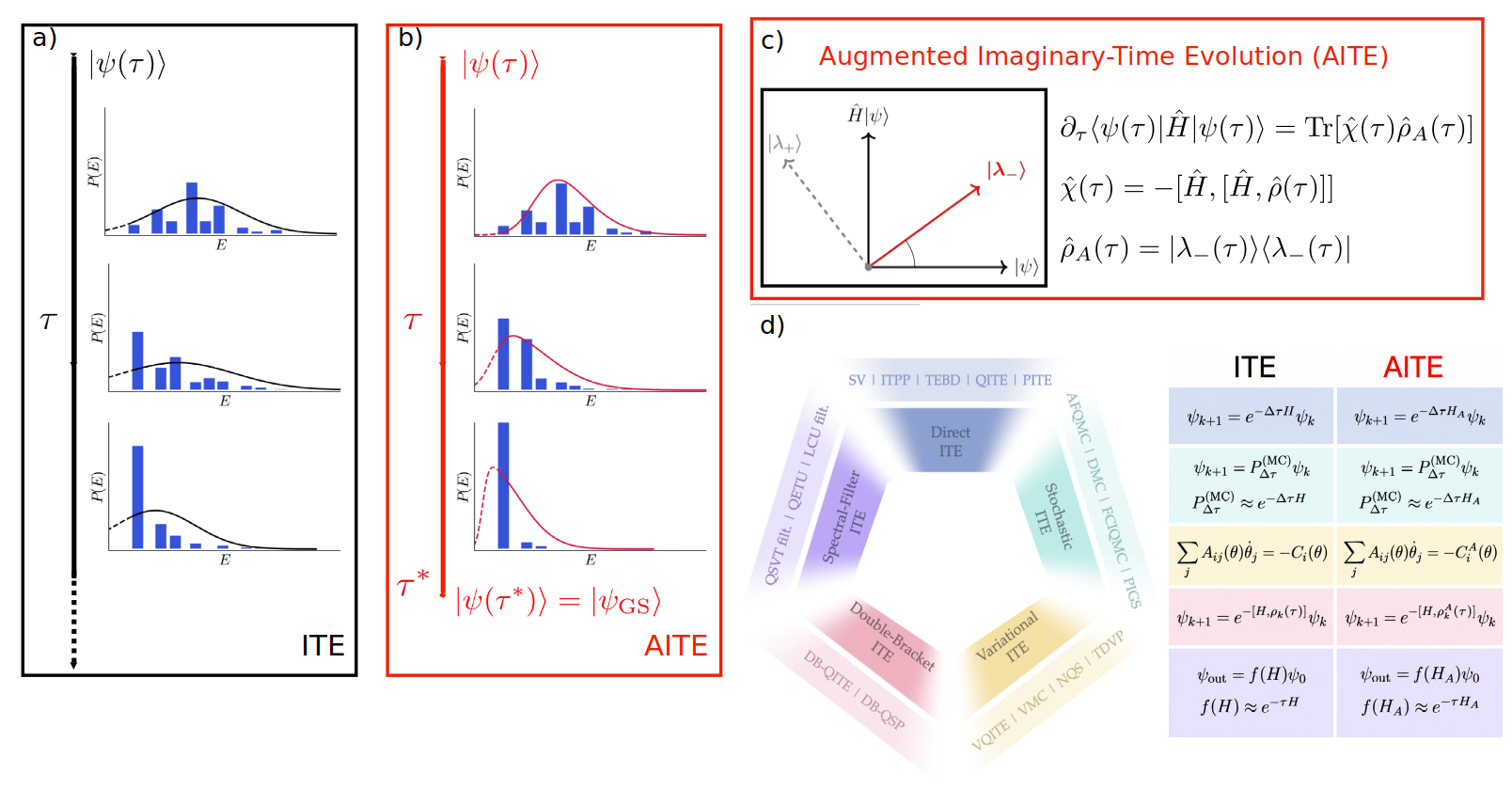}}
\caption{Schematic comparison of standard imaginary-time evolution (ITE) and augmented imaginary-time evolution (AITE). \textbf{a)} In standard ITE, the energy distribution $P(E)$ of the evolving state $|\psi(\tau)\rangle$ shifts monotonically toward lower energies as $\tau$ increases, converging to the ground state asymptotically as $\tau \to \infty$. The convergence rate is governed solely by the energy variance. \textbf{b)} In AITE, the skewness of the instantaneous energy distribution is exploited to identify geometry-informed descent directions, yielding the ground state $|\psi_{\text{GS}}\rangle$ at finite imaginary time $\tau^*$. \textbf{c)} The key equations of AITE: the rate of energy decrease is given by $\mathrm{Tr}[\hat{\chi}(\tau)\hat{\rho}_A(\tau)]$, where $\hat{\chi}(\tau)$ is a double-commutator operator of the type that generates double-bracket flow. The optimal operator $\hat{\rho}_A(\tau) = |\lambda_-\rangle\langle\lambda_-|$ is the projector onto the lowest eigenvector of $\hat{\chi}(\tau)$, projected onto the second-order Krylov subspace spanned by $|\psi\rangle$ and $\bar{H}|\psi\rangle$. \textbf{d)} AITE provides a universal upgrade of the ITE algorithmic family. \textit{Left:} five main branches of ITE algorithms (Direct, Stochastic, Variational, Double-Bracket, and Spectral-Filter ITE) with examples of their representative implementations. \textit{Right:} the corresponding AITE upgrades, obtained by the systematic substitution of $\hat{H} \rightarrow \hat{H}_{\mathrm{A}}(\tau)$ in each branch.}
\label{fig_AITE}
\end{figure*}

Despite this remarkable breadth, the underlying dynamics of ITE have remained essentially unchanged across all these algorithmic branches. The literature has focused almost exclusively on finding efficient ways to implement the ITE flow rather than accelerating it. In almost all cases, this flow drives a monotonic energy decrease at a rate governed by the energy variance of the instantaneous state, and as a result, convergence can be prohibitively slow, depending sensitively on the energy gap and correlation length of the system.  This intrinsic limitation has largely gone unaddressed{\color {black}: ITE is, in principle, guaranteed to reach the ground state, yet in practice the number of steps required can grow rapidly with system size or energy gap.}

In this work, we close this gap by introducing a systematic improvement of the ITE algorithm. Our approach builds on the double-bracket formulation of ITE~\cite{gluza2025}, generalizing it to a significantly broader and more efficient family of descent flows. Rather than relying solely on local gradient information, our \textit{augmented} imaginary-time evolution (AITE) explicitly incorporates the local geometric structure of the energy landscape, encoded in the higher-order statistical structure of the instantaneous energy distribution, to identify locally superlinear descent directions. This yields accelerated convergence and, quite remarkably, finite-time extinction of the energy error {\color {black}in the flow parameter}, in sharp contrast to the asymptotic exponential decay of standard ITE, as illustrated in Figs.~\ref{fig_AITE}a) and \ref{fig_AITE}b). Crucially, since AITE subsumes ITE in its ze\-ro-skewness limit, this improvement propagates simultaneously across the broader family of algorithms inspired by ITE. {\color{black}This acceleration is most pronounced once the energy error has fallen to a certain fraction of the spectral gap; applications whose target precision is coarser than this threshold, such as many
combinatorial optimization instances, may not enter this regime and see comparatively little benefit from the augmentation. We also stress that AITE prescribes a filter, i.e., \emph{what} to implement, in the same sense that ITE prescribes an exponential filter and quantum signal processing a polynomial one, rather than a circuit compilation. In that sense, finite extinction finite extinction in $\tau$ is not a statement about circuit depth, and the query complexity of implementing the AITE filter is analyzed separately.}

% APPARENTLY NAT COM DOES NOT NEED THIS KIND OF PARAGRAPHS
%This paper is organized as follows. We first review the double-bracket formulation of quantum ITE and establish the geometric framework underlying our approach. We then introduce the main ingredients of AITE, derive the optimal descent direction, and discuss its implementation in both unitary and non-unitary forms, together with an analysis of its convergence properties. Next, we present numerical results for weakly and strongly correlated systems from condensed matter and quantum chemistry. After studying the query complexity of a specific implementation of AITE, we conclude with an outlook on extensions and potential applications of the AITE framework.

\section*{Results}

ITE prepares ground states by ``cooling" an initial trial state $\ket{\psi_0}$ through the continuous application of a non-unitary operator: $\ket{\psi(\tau)} = e^{-\tau \hat{H}} \ket{\psi_0}$ \cite{PhysRev.106.364}. Provided the initial trial state has a nonzero overlap with the ground state $\ket{\psi_{\rm GS}}$ (i.e., $|\langle \psi_{\rm GS} | \psi_0 \rangle| \neq 0$), this evolution converges to the ground state in the limit of infinite imaginary time, $\tau \to \infty$ \cite{PhysRev.84.350,10.1063/1.431514}. For practical purposes, it is convenient to work with the normalized state
\begin{align}
\label{normalizedstate}
\ket{\psi(\tau)} = \frac{e^{-\tau \hat H}\ket{\psi_0}}{||e^{-\tau \hat H}\ket{\psi_0}||}.
\end{align}
This flow drives monotonic energy decrease at a rate set by the energy variance {\color{black} ${\color{black} {\color{black} \bra{\psi(\tau)}(\hat H - E(\tau))^2\ket{\psi(\tau)}}}$ (where  $  E(\tau) = \bra{\psi(\tau)}\hat H \ket{\psi(\tau)}$ is the instantaneous energy)}, notably, without critical slowing down (the rate never plateaus while $E(\tau)>E_{\rm GS}$, though it can
still be small when the gap is small)~\cite{hartung2025}. As mentioned in the introduction, %and illustrated in Fig.~\ref{fig_AITE}d), 
these ima\-gi\-na\-ry-time dynamics underlie a remarkably rich family of algorithms for quantum simulation and optimization.

The state of Eq.~\eqref{normalizedstate} satisfies the norm-preserving differential equation $\partial_\tau \ket{\psi(\tau)} = -\big(\hat{H} - E(\tau)\big)\ket{\psi(\tau)}$. Recently, Ref.~\cite{gluza2025} observed that this equation can be written as
\begin{align}
\label{DB}
\partial_\tau \ket{\psi(\tau)} = -[\hat{H}, \hat{\rho}(\tau)]\ket{\psi(\tau)},
\end{align}
where $\hat{\rho}(\tau) = \ket{\psi(\tau)}\bra{\psi(\tau)}$. This reformulation is significant for two reasons. First, the generator $[\hat{H}, \hat{\rho}(\tau)]$ is anti-Hermitian, so the flow admits a natural unitary realization. Second, this unitary structure {\color{black} makes the flow implementable on quantum hardware with the same primitives as real-time evolution, namely queries to $\exp[-i\hat H t]$ and reflections about the current state, without the ancilla-based block encodings required for the non-unitary of Eq.\eqref{normalizedstate}~\cite{melendez2025,shrikhande2025rapidgroundstateenergy}}.
 
We generalize the ITE flow by {\color{black}modifying Eq.~\eqref{DB} so that the generator is built from a general rank-one projection operator
$\hat{\mathcal{O}}(\tau)\in\mathcal{P}_1(\mathcal{H})$, where
$\mathcal{P}_1(\mathcal{H})=\{\hat{\mathcal{O}}:\hat{\mathcal{O}}^\dagger=\hat{\mathcal{O}},\ \hat{\mathcal{O}}^2=\hat{\mathcal{O}},\ \mathrm{Tr}\,\hat{\mathcal{O}}=1\}$
denotes the set of pure-state density operators on the Hilbert space $\mathcal{H}$, of which the density operator $\hat{\rho}(\tau)$
is the special case recovering ITE:
\begin{align}
\label{improvedDB}
    \partial_\tau \ket{\psi(\tau)} = -[\hat{H},\, \hat{\mathcal{O}}(\tau)]\ket{\psi(\tau)}.
\end{align}}%
Because $[\hat{H}, \hat{\mathcal{O}}(\tau)]$ is anti-Hermitian, the flow is
norm-preserving for any such choice. The requirement of monotonic energy
decrease, $\partial_\tau \bra{\psi(\tau)}\hat{H}\ket{\psi(\tau)} \leq 0$,
is what singles out the admissible $\hat{\mathcal{O}}(\tau)$, and we use it
below to select the augmented projector. A direct calculation gives the
energy rate as a linear functional of $\hat{\mathcal{O}}(\tau)$,
\begin{align}
\label{functional}
    \partial_\tau \bra{\psi(\tau)}\hat{H}\ket{\psi(\tau)} =
    \mathrm{Tr}\!\big[\hat{\chi}(\tau)\,\hat{\mathcal{O}}(\tau)\big],
\end{align}
where 
\begin{align}
    \hat{\chi}(\tau) = -[\hat{H},[\hat{H},\hat{\rho}(\tau)]]
\end{align} 
is a state-dependent dou\-ble-commutator operator of the type that generates
dou\-ble-bracket flows~\cite{BROCKETT199179, Bach2010}. For the ITE choice
$\hat{\mathcal{O}}=\hat{\rho}$ this functional equals minus twice the
energy variance of $\ket{\psi(\tau)}$, recovering the familiar ITE rate.
Finding a faster ener\-gy-decreasing update thus reduces to minimizing the
linear functional in Eq.~\eqref{functional} over the space of rank-one
projection operators, which defines the augmented projector
{\color{black}\begin{align}
\label{rhoA}
\hat{\rho}^{\mathrm{A}}(\tau) = \operatorname*{argmin}_{\hat{\mathcal{O}}\in\mathcal{P}_1(\mathcal{H})}\, \mathrm{Tr}\![\hat{\chi}(\tau)\,\hat{\mathcal{O}}],
\end{align}}%
with the superscript A standing for AITE. Since
$\mathrm{Tr}[\hat{\chi}\hat{\mathcal{O}}]=\bra{\phi}\hat{\chi}\ket{\phi}$
for $\hat{\mathcal{O}}=\ket{\phi}\!\bra{\phi}$, the minimizer is the
projector onto the lowest eigenvector of $\hat{\chi}(\tau)$, and by the
variational principle the corresponding rate is bounded above by the ITE
rate, $\lambda_{\min}(\hat{\chi})\le\bra{\psi}\hat{\chi}\ket{\psi}$.
Inserted in Eq.~\eqref{improvedDB}, $\hat{\rho}^{\mathrm{A}}(\tau)$
therefore steers the evolution along descent directions that are provably
at least as steep as those of ITE, with equality only when
$\ket{\psi(\tau)}$ is itself an eigenvector of $\hat{\chi}(\tau)$.
{\color{black}We stress that ``faster'' here means a larger energy decrease
per unit $\tau$ \cite{note:tau}}.
The explicit construction of $\hat{\rho}^{\mathrm{A}}(\tau)$ is given in
the Methods section.

\subsection*{Implementation of AITE} 

We now discuss a practical choice of $\hat{\rho}^{\mathrm{A}}(\tau)$ in order to implement AITE. While optimizing Eq.~\eqref{functional} over the full space of Hermitian operators is generally intractable, a natural and implementable solution emerges by restricting the search to a Krylov subspace---a well-established numerical framework that is, by no means, the only possible choice. Here, our choice is to work within the second-order Krylov 
subspace,
\begin{align}
\mathcal{S}_{\hat{H}}(\tau) = \bigg\{\ket{\psi(\tau)}, \ket{v_H(\tau)} \equiv   \frac{\bar{H}  \ket{\psi(\tau)}}{\sqrt{\mu_2(\tau)}}\bigg\},
\end{align}
where $\bar{H}=\hat H - \bra{\psi(\tau)}\hat H \ket{\psi(\tau)}$ is the centered Hamiltonian operator. We also define the $n$th energy central mo\-ment as $\mu_n(\tau) = 
\bra{\psi(\tau)}\bar{H}^n\ket{\psi(\tau)}$.
Within this subspace, the projection of the double-bracket operator $\hat \chi (\tau)$ captures variations of the energy variance along directions orthogonal to $\ket{\psi(\tau)}$. See the Supplementary In\-for\-mation for a geometric interpretation of the resulting projection. 

 Projected in this subspace, the double-bracket operator takes the explicit form:
\begin{align}
\hat{\chi}(\tau)\big|_{\mathcal{S}_H(\tau)} = -2\mu_2(\tau)
\begin{pmatrix}
1 &  \frac{\kappa(\tau)}{2} \\
 \frac{\kappa(\tau)}{2} & -1
\end{pmatrix},
\label{matrix}
\end{align}
where 
\begin{align}
 \kappa(\tau) = \frac{\mu_3(\tau)}{\mu_2^{3/2}(\tau)} .
\end{align}
is the instantaneous Fisher--Pearson skewness coefficient of the energy distribution~\cite{doi:10.1098/rsta.1895.0010}. This matrix admits a  transparent statistical interpretation: the diagonal entries encode the 
energy variance $\mu_2(\tau)$, while the off-diagonal entries are controlled  by $\kappa(\tau)$, which captures the asymmetry of the instantaneous  energy distribution. While skewness is a well-established diagnostic  of non-Gaussianity in quantitative finance and statistical learning~\cite{Groeneveld1984, Cirillo2020, dominguez2025, Bouchaud_Potters_2003,Joanes1998}, it has received comparatively little attention in the quantum  simulation and ITE literature~\cite{PRXQuantum.5.040339, holevo2011probabilistic}. Within $\mathcal{S}_H(\tau)$ subspace, AITE thus naturally steers the descent using statistical information beyond the variance.

The eigenvalues of Eq.~\eqref{matrix} are  
\begin{align}
\lambda_\pm(\tau) = \pm 2\mu_2(\tau)\sqrt{1 + \frac{\kappa^2(\tau)}{4}},
\label{eq:flow_velocity}
\end{align}
with $\lambda_-(\tau) \leq 0 \leq \lambda_+(\tau)$ for all $\tau$. Restricting $\hat{\mathcal{O}}(\tau)$ to rank-one projectors, the functional in Eq.~\eqref{functional} is minimized by the projector onto the eigenvector corresponding to $\lambda_-(\tau)$,
\begin{align}
\hat{\rho}^{\mathrm{A}}(\tau) \equiv \ket{\lambda_-(\tau)}\bra{\lambda_-(\tau)},
\label{eq:rho_G}
\end{align}
where 
\begin{align}
\ket{\lambda_-(\tau)} = \cos\phi(\tau)\ket{\psi(\tau)} +  \sin\phi(\tau)\ket{v_H(\tau)}.
\end{align} 
The mixing angle $\phi(\tau)$ is a central quantity in our framework:  it measures the deviation of the energy distribution from Gaussianity via the relation $\tan(2\phi(\tau)) = \kappa(\tau)/2$. Notice that $\kappa(\tau) = 0$ implies $\phi(\tau) = 0$, and $\hat{\rho}^{\mathrm{A}}(\tau)$ reduces to $\hat{\rho}(\tau)$, recovering standard ITE. The structure of this optimized descent operator is illustrated schematically in Fig.~\ref{fig_AITE}c). Notably, extending the optimization to higher-order Krylov subspaces provides a systematic and principled route to capturing higher-order energy fluctuations (such as the kurtosis $\mu_4$, and beyond), offering a natural hierarchy of improvements over standard ITE.

The descent rate corresponds to the lowest eigenvalue, and it satisfies $\lambda_-(\tau) \leq -2\mu_2(\tau) = \mathrm{Tr}[\hat{\chi}(\tau)\,\hat{\rho}(\tau)]$, thereby demonstrating that ITE is provably suboptimal compared to AITE. The inequality is strict whenever $\kappa(\tau) \neq 0$, that is, whenever the instantaneous energy distribution is asymmetric. This reveals a structural limitation of ITE algorithms that, to the best of our knowledge, has not been previously identified: conventional ITE implicitly assumes a symmetric energy distribution (i.e., $\kappa(\tau) = 0$), discarding non-Gaussian features and, most prominently, the skewness. 

\subsection*{Non-unitary AITE}

A natural question is whether the generalized double-bracket flow of Eq.~\eqref{improvedDB} with the minimizer obtained in Eq.~\eqref{eq:rho_G}, admits a formulation as an \emph{augmented} non-uni\-tary ITE. The answer is affirmative: our framework offers considerable flexibility in designing augmented Hamiltonians $\hat{H}_{\mathrm{A}}(\tau)$ such that the 
update rule
\begin{align}
\ket{\psi(\tau+\delta\tau)} 
= \frac{e^{-\delta\tau\hat{H}_{\mathrm{A}}(\tau)}\ket{\psi(\tau)}}
       {\|e^{-\delta\tau\hat{H}_{\mathrm{A}}(\tau)}\ket{\psi(\tau)}\|}
\label{eq:ite_update}
\end{align}
reproduces, to first order in $\delta\tau$, the same energy descent rate as the optimal double-bracket flow $\lambda_-(\tau)$. A concrete example of such an augmented Hamiltonian is the following:
\begin{align}
\hat{H}_{\mathrm{A}}(\tau) 
= \cos(2\phi)\,\bar{H} 
+ \frac{\sin(2\phi)}{2\sqrt{\mu_2(\tau)}}\,\bar{H}^2,
\label{eq:effectiveH}
\end{align}
{\color{black}which, being a fixed function of $\hat H$ alone, shares $\hat H$'s eigenbasis at every $\tau$. Consequently, integrated along the flow, the update of Eq.~\eqref{eq:ite_update} amounts to a single spectral filter, a Gaussian in $\hat H$, whose query complexity we analyze in the Methods section.} 

It is possible to construct augmented Hamiltonians by downfolding $\bar{H}^2$ and retaining terms up to three-body interactions for electronic systems. In this case, as discussed in the Methods section, the prefactor of $\bar{H}^2$ is mo\-di\-fied accordingly, while the overall structure of the descent flow is preserved.  

\subsection*{Convergence with power-law extinction} 

We now investigate the late-time convergence structure of AITE. Standard ITE converges to the ground state with an energy error decaying as $\varepsilon_{\rm ITE}(\tau) \sim e^{-2\Delta\tau}$, where $\Delta = E_1 - E_{\rm GS}$ is the spectral gap and $\varepsilon_{\rm ITE}(\tau)$ denotes the energy error $E_{\rm ITE}(\tau)-E_{\rm GS}$. We show that AITE belongs to a qualitatively distinct convergence class. In the near-convergence regime $\varepsilon_{\rm AITE}(\tau) \ll \Delta$, the quantum state is dominated by its ground-state component, so that $\mu_2(\tau) \approx \Delta\varepsilon(\tau)$ and $\mu_3(\tau) \approx \Delta^2\varepsilon(\tau)$. Using the AITE energy descent rate from  Eq.~\eqref{eq:flow_velocity} yields $\dot{\varepsilon}_{\rm AITE} = -2\Delta\varepsilon_{\rm AITE}\sqrt{1 + \Delta/4\varepsilon_{\rm AITE}}$, which makes explicit that AITE is strictly faster than ITE at every finite energy error $\varepsilon$. The \textit{superlinear} speedup factor $\sqrt{1 + \Delta/4\varepsilon}$ diverges as $\varepsilon \to 0$, reflecting increasingly aggressive acceleration near the ground state. Standard ITE, $\dot{\varepsilon}_{\rm ITE} = -2\Delta\varepsilon_{\rm ITE}$, is recovered in the large-error limit $\varepsilon \gg \Delta/4$, identifying it as the large-error asymptote of the  augmented dynamics. Interestingly, Hartung and Jansen have recently shown that the threshold $\varepsilon \gg \Delta/4$ is efficiently reachable for certain systems \cite{hartung2025}.

\begin{figure}[t]
\centering
{\includegraphics[width=8.5cm]{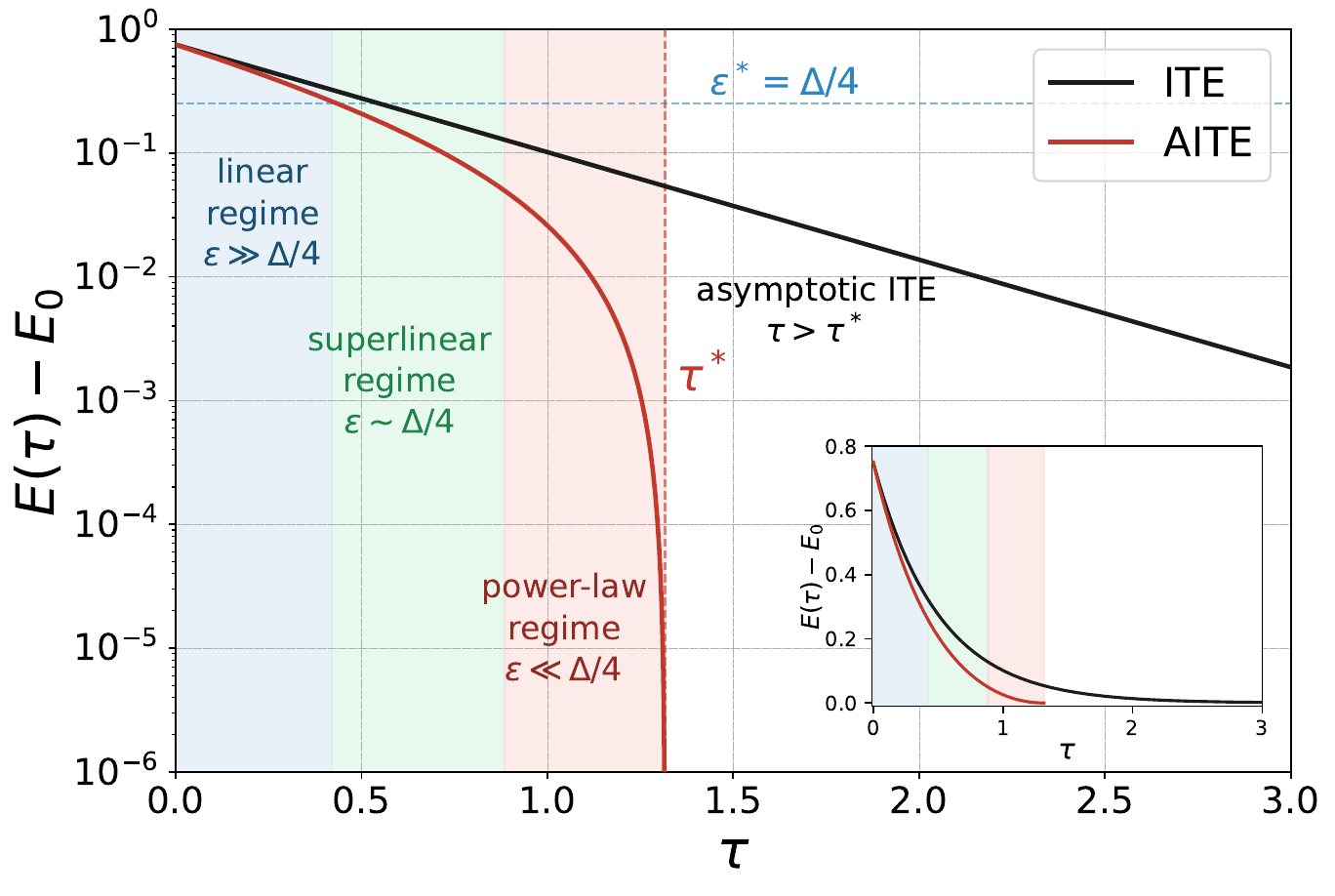}}
\caption{Energy error $E(\tau)-E_{\rm GS}$ as a function of imaginary time $\tau$ for ITE (black) and AITE (red) as predicted in Eq.~\eqref{eq:exact_sol} (with $\Delta=1$ and $\varepsilon_0=0.75$). The main panel shows the error evolution on a logarithmic scale, with shaded regions identifying four dynamical regimes: the linear regime ($\varepsilon\gg\Delta/4$, blue), the superlinear regime ($\varepsilon\sim\Delta/4$, green), the power-law extinction regime ($\varepsilon\ll\Delta/4$, red), and the asymptotic ITE regime. The dashed vertical line marks the finite extinction time $\tau^*$ of AITE. The inset displays the same dynamics on a linear scale, highlighting the finite-time extinction of AITE in contrast to the asymptotic exponential decay of ITE.}
\label{figstwolevel}
\end{figure}

The AITE flow admits the exact closed-form solution
\begin{align}
\varepsilon_{\rm AITE}(\tau) = \frac{\Delta}{4}
\sinh^2\!\left(
\sinh^{-1}\!\left(2\sqrt{\frac{\varepsilon_0}{\Delta}}\right)
- \Delta\tau
\right),
\label{eq:exact_sol}
\end{align}
where $\varepsilon_0 = \varepsilon(0)$ is the initial energy error. Surprisingly, unlike the ITE solution, which decays exponentially and reaches zero only as $\tau \to \infty$, Eq.~\eqref{eq:exact_sol} vanishes exactly (i.e., $\varepsilon_{\rm AITE}(\tau^*) = 0$) at the finite extinction time:
\begin{align}
\tau^* = \frac{1}{\Delta}
\sinh^{-1}\!\left(2\sqrt{\frac{\varepsilon_0}{\Delta}}\right).
\label{eq:tau_star}
\end{align}
Expanding the energy error around $\tau^*$ and using $\sinh(x) \approx x$ for small $x$, Eq.~\eqref{eq:exact_sol} gives
\begin{align}
\label{eq:power_law_extinction}
    \varepsilon_{\rm AITE}(\tau) \approx \frac{\Delta^3}{4}(\tau^* - \tau)^2,
\end{align}
a \emph{power-law extinction} with exponent $2$. This stands in sharp contrast to the exponential tail of ITE and establishes that the two methods belong to provably distinct convergence classes. The finite-time power-law extinction of AITE is reminiscent of finite-time extinction in nonlinear diffusion equations, such as the fast diffusion equation in porous media~\cite{Vasquez2006}, where solutions to $\dot{\varepsilon} = -c\varepsilon^\alpha$ with $\alpha < 1$ are known to reach zero in finite time. The near-ground-state AITE dynamics fall precisely into this class with $\alpha = \tfrac{1}{2}$.

Finally, from Eq.~\eqref{eq:tau_star}, the extinction time $\tau^*$ depends on the initial error only through $\sinh^{-1}(2\sqrt{\varepsilon_0/\Delta})$.
In the large-error regime $\varepsilon_0 \gg \Delta$, using $\sinh^{-1}(x) \approx \ln(2x)$ for $x \gg 1$, this gives $\tau^* \approx ({1}/{2\Delta})\ln({4\varepsilon_0}/{\Delta})$ so $\tau^*$ grows only logarithmically with $\varepsilon_0$. This is to be compared with standard ITE, where the time to reach a fixed target precision $\varepsilon_*$ is $\tau_{\mathrm{ITE}}(\varepsilon_*) = ({1}/{2\Delta})\ln({\varepsilon_0}/{\varepsilon_*})$, which diverges as $\varepsilon_* \to 0$ for any fixed $\varepsilon_0$. In AITE, by contrast, the time to reach \emph{any} target precision is bounded from above by $\tau^*$, which is independent of $\varepsilon_*$ and grows only as $\ln(\varepsilon_0/\Delta)$. {\color{black}We stress that this bound is a statement about the flow parameter $\tau$, not about computational cost: $\tau^*$ carries the full dependence on the spectral gap through the prefactor $1/\Delta$ in Eq.~\eqref{eq:tau_star}, and for Hamiltonians with an exponentially small gap, $\Delta\sim e^{-\Omega(n)}$, $\tau^*$ is itself exponential in system size. The improvement over ITE established here is a reduction in the flow parameter at fixed $\Delta$, not a change in the scaling with system size. Likewise, the computational cost of reaching a target precision $\varepsilon_*$ remains logarithmic in $1/\varepsilon_*$ in every implementation, as shown in the Methods section.} Although the analysis above focuses on the near-convergence regime, where the state is dominated by its ground-state component, the result remains valid for arbitrary state populations. The proof of the general case is deferred to the Methods section.

\begin{figure*}[!t]
    \centering
    % First row
    \begin{subfigure}{0.32\textwidth}
        \centering
        \includegraphics[width=\linewidth]{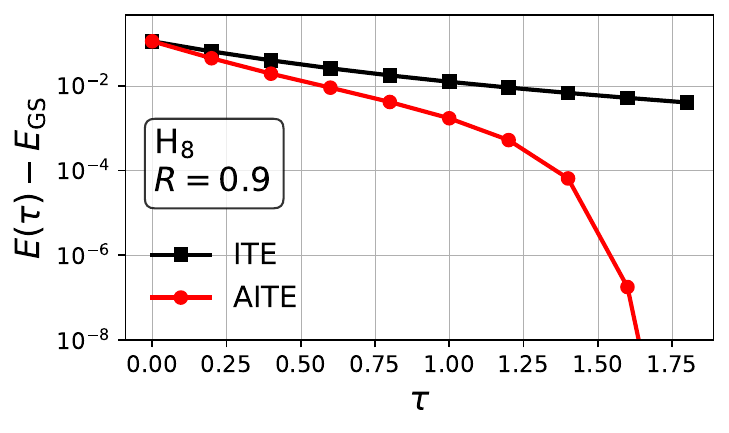}
    \end{subfigure}
    \hfill
    \begin{subfigure}{0.32\textwidth}
        \centering
        \includegraphics[width=\linewidth]{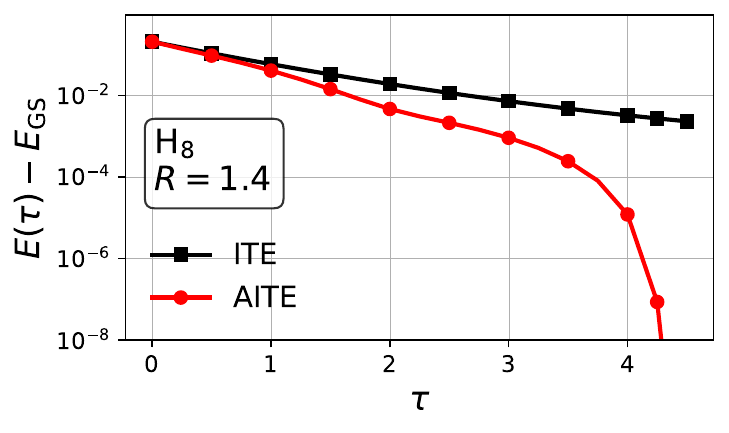}
    \end{subfigure}
    \hfill
    \begin{subfigure}{0.32\textwidth}
        \centering
        \includegraphics[width=\linewidth]{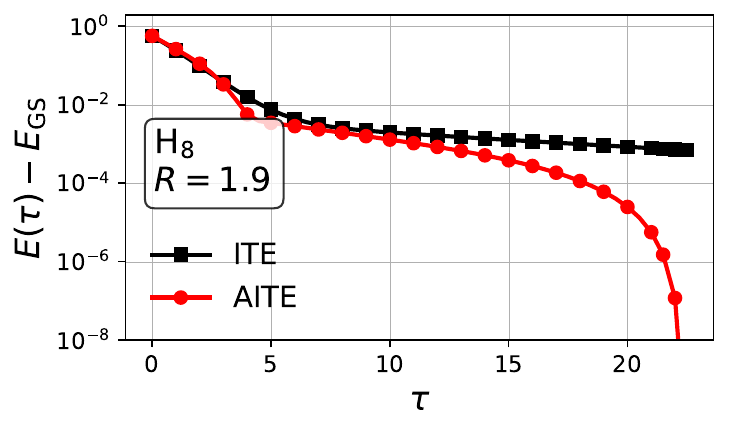}
    \end{subfigure}

    \medskip

    % Second row
    \begin{subfigure}{0.32\textwidth}
        \centering
        \includegraphics[width=\linewidth]{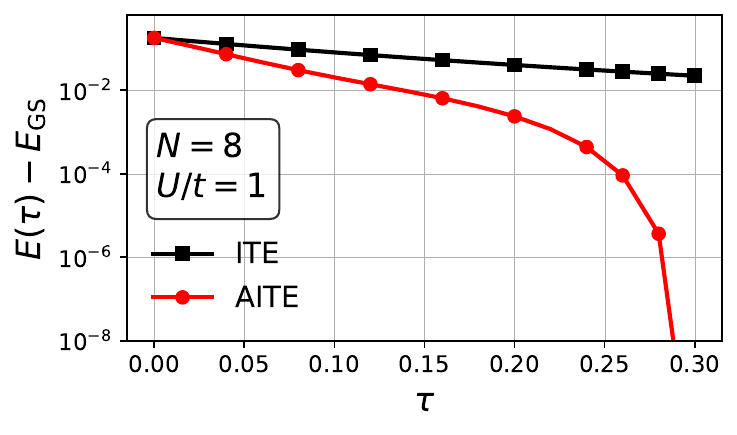}
    \end{subfigure}
    \hfill
    \begin{subfigure}{0.32\textwidth}
        \centering
        \includegraphics[width=\linewidth]{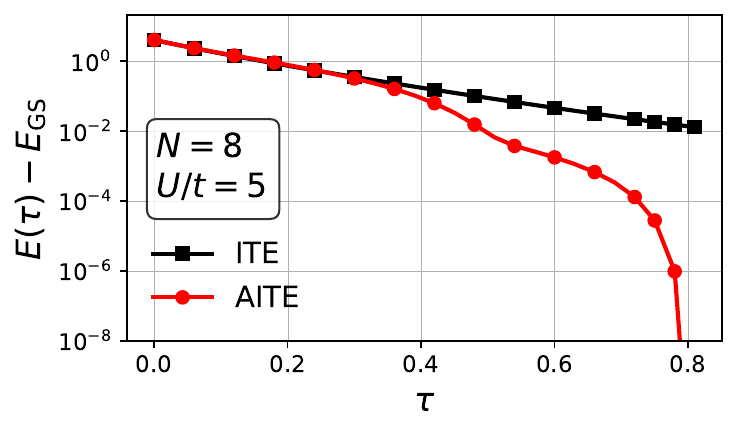}
    \end{subfigure}
    \hfill
    \begin{subfigure}{0.32\textwidth}
        \centering
        \includegraphics[width=\linewidth]{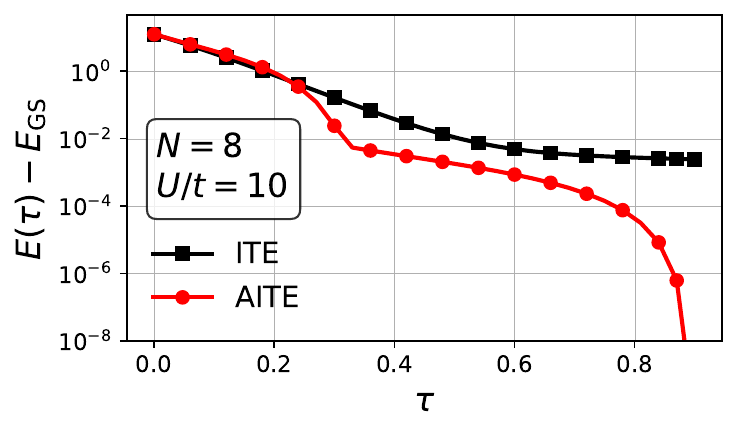}
    \end{subfigure}
\caption{Energy error $E(\tau)-E_{\rm GS}$ as a function of the imaginary-time flow parameter $\tau$ for ITE (black) and AITE (red), starting from the Hartree--Fock state. Upper panels: equidistant $\mathrm{H}_8$ chain (STO-3G basis), with interatomic spacing $R=0.9$, $1.4$ and $1.9$~\AA{} from left to right; energies in Hartree and $\tau$ in $\mathrm{Ha}^{-1}$. Lower panels: one-dimensional Fermi--Hubbard model with $N=8$ sites at half filling and open boundary conditions, with $U/t=1$, $5$ and $10$ from left to right; energies in units of $t$ and $\tau$ in units of $1/t$. Both methods use the same step size: $\delta\tau=0.2$, $0.25$ and $0.5$
for $R=0.9$, $1.4$ and $1.9$~\AA{} (9, 18 and 45 steps), and $\delta\tau= 0.03$ for $U/t=1$, $5$ and $10$ (15, 27 and 30 steps). AITE uses the non-unitary
update of Eq.~(\ref{eq:ite_update}) with the augmented Hamiltonian of
Eq.~(\ref{eq:effectiveH}). }
\label{fig1}
\end{figure*}

In summary, the fundamental asymmetry between the two algorithms is this: in ITE, the convergence time grows without bound as the target precision is tightened, whereas in AITE it is determined solely by the initial energy error and the spectral gap, independent of the target precision altogether.
As illustrated in Fig.~\ref{figstwolevel}, AITE exhibits three distinct dynamical regimes: a \textit{linear regime}, in which the descent closely tracks standard ITE; a \textit{superlinear regime}, in which the non-Gaussian skewness begins to dominate and drive the AITE dynamics; and a \textit{power-law extinction regime}, in which the energy error vanishes exactly at the finite time $\tau^*$, while ITE continues its asymptotic exponential decay. 

\subsection*{Numerical experiments} 

We benchmark AITE across both weakly and strongly correlated regimes of two model systems in condensed matter and quantum chemistry. Fig.~\ref{fig1} compares AITE to ITE for the equidistant hydrogen chain  $\mathrm{H}_8$ at three interatomic distances ($R = 0.9, 1.4, 1.9$~\AA) and the eight-site one-dimensional Fermi-Hubbard model with open boundary conditions at half filling at three interaction strengths ($U/t = 1, 5, 10$). In both cases, we use the Hartree--Fock state as the initial state. We find that across all geometries and interaction strengths, the energy error, as a function of the parameter $\tau$, exhibits the three dynamical regimes predicted analytically above: the linear, superlinear, and power-law extinction regimes. The consistent acceleration across different correlation regimes corroborates the central role of non-Gaussian features of the instantaneous energy distribution in driving convergence beyond the exponential barrier of standard ITE.

\textit{Heuristics for higher momenta.---}
The skewness appearing in Eq.~\eqref{matrix} probes non-Gaussian structure in the energy distribution that is costly to access in practice, both on near-term quantum hardware~\cite{PRXQuantum.5.040339, PRXQuantum.2.020321} and in classical implementations, as it requires estimating three-fold correlations of the Hamiltonian. Below, we address this by proposing a practical strategy to reduce computational overhead without sacrificing the superlinear convergence of AITE.

\begin{figure}[b]
\centering
{\includegraphics[width=8.8cm]{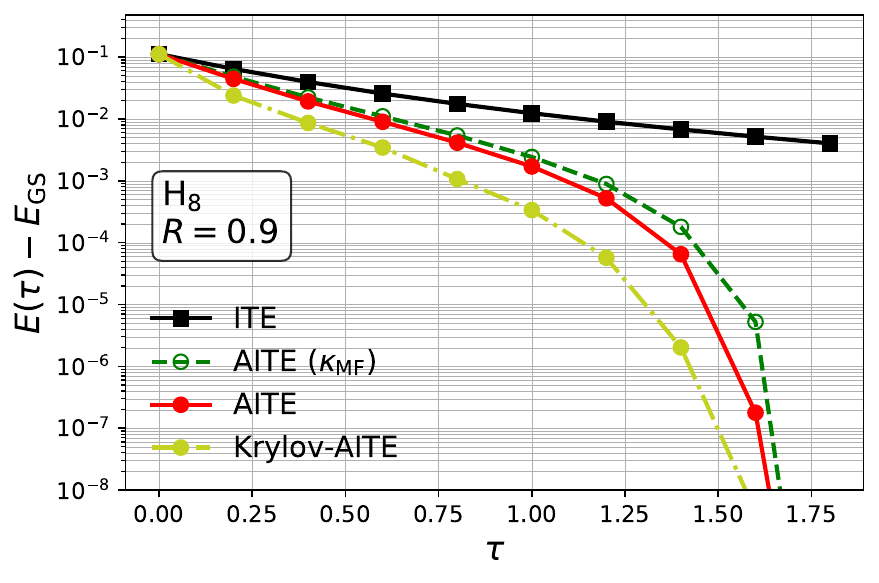}}
\caption{Energy error $E(\tau)-E_{\rm GS}$ versus $\tau$ for H$_8$ at $R=0.9$~\AA{}
(STO-3G, Hartree--Fock initial state, $\delta\tau=0.2$): ITE (black), AITE with exact skewness $\kappa$ (red), AITE with mean-field skewness $\kappa_{\rm MF}=-2E/\sqrt{\mu_2}$ (green), and the Krylov energy estimate, i.e.\ the lowest eigenvalue of $\hat H$ in $\mathcal S_H(\tau)$ along the AITE trajectory (yellow).}
\label{fig4}
\end{figure}

We introduce a mean-field approximation
$\langle \hat{H}^3\rangle \approx \langle \hat{H}^2\rangle\langle \hat{H}\rangle$, which decouples three-point correlations into products of lower-order expectation values. {\color{black}Inserting it into the third central moment, $\mu_3=\langle\hat H^3\rangle-3E\langle\hat H^2\rangle+2E^3$, gives $\mu_3^{\rm MF}(\tau)=-2E(\tau)\,\mu_2(\tau)$ and hence the mean-field skewness $\kappa_{\rm MF}(\tau) = -2E(\tau)/\sqrt{\mu_2(\tau)}$. Unlike the exact $\kappa(\tau)$, which is built from central moments and is invariant under $\hat H\to\hat H+c$, the factorization refers to the uncentered operator and therefore depends on the energy origin: $\kappa_{\rm MF}$ is defined with respect to the physical energy origin of $\hat H$ as written, and is meaningful only when $|E(\tau)|\gg\sqrt{\mu_2(\tau)}$,
(i.e., when the mean energy dominates the width in that frame). The sign
of $\kappa_{\rm MF}$ is that of $-E(\tau)$; for the electronic and
lattice Hamiltonians considered below, $E(\tau)<0$ throughout the
descent and $\kappa_{\rm MF}>0$, consistent with the right-skewed energy
distribution of a state concentrated near the ground state with a small
excited-state tail.} Crucially, this expresses the third moment entirely in terms of first- and second-order expectation values of $\hat{H}$,
which are directly accessible from energy and variance measurements
{\color{black}without measuring three-fold correlations}. {\color{black}In
the near-convergence regime (where $\mu_2\simeq\Delta\varepsilon$ and
$\mu_3\simeq\Delta^2\varepsilon$), the mean-field skewness reproduces the sign and the $\varepsilon^{-1/2}$ divergence of the exact one,
$\kappa_{\rm MF}/\kappa\to2|E_{\rm GS}|/\Delta$, up to a system-dependent
constant. Since the mixing angle depends on $\kappa$ only through $\arctan(\kappa/2)$ and saturates for $\kappa\gg1$, this constant has a vanishing effect on the update direction as the algorithm converges, which is why the heuristic retains the superlinear convergence of AITEin the numerical experiments below.
} 

Fig.~\ref{fig4} compares the performance of our heuristic approximation to the skewness against standard ITE and AITE, and also includes a Krylov-subspace energy estimate in which, at each imaginary-time step, the energy is optimized within the second-order Krylov subspace $\mathcal{S}_H(\tau)$ generated by the instantaneous state. We note that this Krylov energy estimate is upper-bounded by the AITE energy and comes essentially for free, as it requires only the expectation values $\langle \hat{H}^k \rangle$ with $k \leq 3$, all of which are already evaluated as part of the AITE flow.

At early times, when $|\kappa(\tau)/2|\lesssim1$, all four curves nearly coincide. Once the skewness grows, AITE enters the superlinear regime standard ITE cannot reach. Remarkably, $\kappa_{\rm MF}$ captures this skewness well enough that the approximate AITE matches the exact method's convergence---a practical alternative avoiding third-order measurements---with especially strong performance at equilibrium geometries. Extending this heuristic beyond such settings is left for future work. As expected, the Krylov estimate yields even lower energies than AITE, sometimes by orders of magnitude.

\begin{figure*}[t!]
\centering
\includegraphics[width=\textwidth]{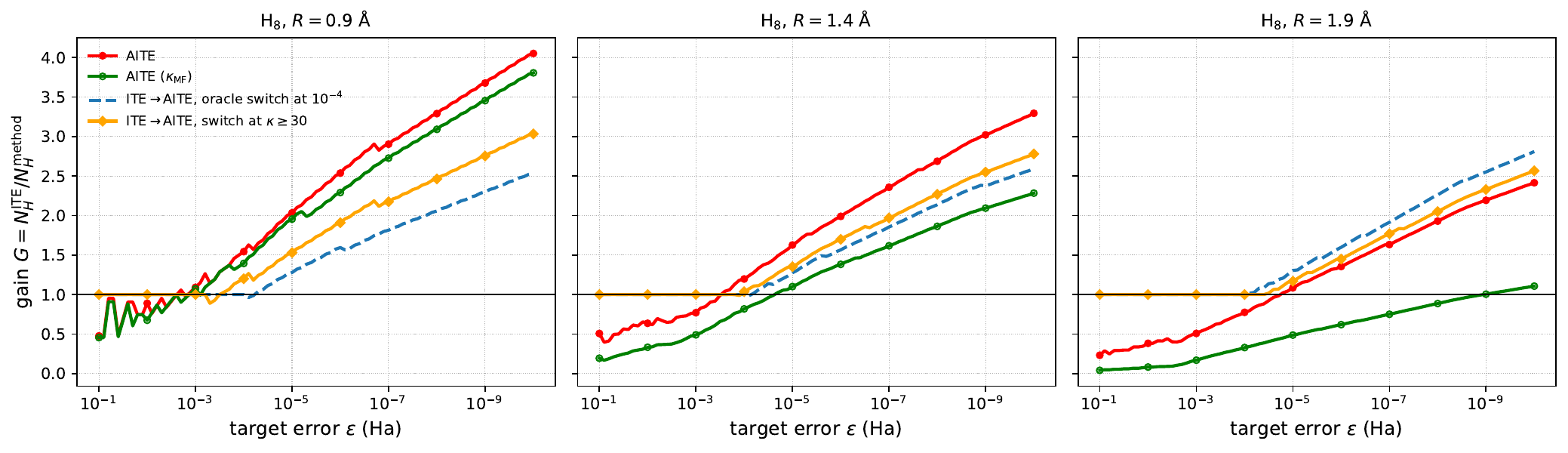}
\caption{Gain $G(\varepsilon)=N_H^{\rm ITE}(\varepsilon)/N_H(\varepsilon)$ in the number of
Hamiltonian applications needed to reach an energy error $\varepsilon$, for H$_8$ at
$R=0.9$, $1.4$ and $1.9$~\AA{} (left to right), for AITE, AITE with $\kappa_{\rm MF}$,
and ITE followed by AITE (switch at $10^{-4}$, oracle; switch at $\kappa\geq30$).
Values $G>1$ indicate fewer applications than ITE. Both methods use the step sizes of
Fig.~2 and the same Krylov integrator on $\hat H$   .}
\label{fig:gain}
\end{figure*}

{\color{black}\subsection*{Computational effort} 

Since the imaginary time $\tau$ is not a
computational resource, we compare the methods by the number of Hamiltonian
applications (matrix--vector products) needed to apply the propagators, the standard
cost measure for propagation methods and the classical analogue of query complexity.
Both methods use the same step sizes as in Fig.~\ref{fig1} and the same adaptive
Krylov (Lanczos) integrator, converged to a relative accuracy of $10^{-10}$ at every
step. Because the augmented Hamiltonian $\hat H_{\rm A}$ of Eq.~\eqref{eq:effectiveH}
is a function of $\hat H$, each AITE propagator $\exp[-\delta\tau\hat H_{\rm A}]$ is evaluated as a function of $\hat H$ on the Krylov space of $\hat H$ itself, exactly as
for ITE. A Krylov expansion in $\hat H_{\rm A}$ would inherit the slow convergence of folded-spectrum methods, since the $\bar H^2$ term folds the spectrum around the instantaneous energy, and would cost about twice as many Hamiltonian applications for the same trajectory. %The estimation of $\mu_2$ and $\mu_3$ is a measurement cost common to all AITE variants and can be counted separately.

Figure~\ref{fig:gain} shows the following ratio:
\begin{align}
G(\varepsilon)=N_H^{\rm ITE}(\varepsilon)
/N^{\rm AITE}_H(\varepsilon),
\end{align} 
where $N_H(\varepsilon)$ is the
number of Hamiltonian applications needed to reach an energy error $\varepsilon$. For H$_8$ each AITE step becomes more expensive as $\mu_2\to0$, since the coefficient of $\bar H^2$ grows. While far from convergence ITE is therefore the cheaper method, close to the ground state the balance reverses: ITE pays a fixed, gap-limited cost for every additional decade of precision, whereas AITE covers many decades in a few steps as it approaches extinction. The gain exceeds 1 below $\varepsilon\sim10^{-3}$--$10^{-4}$~Ha and grows steadily with the target precision. At
$\varepsilon=10^{-10}$~Ha, AITE requires 155, 384 and 918 Hamiltonian applications for $R=0.9$, $1.4$ and $1.9$~\AA{}, against 628, 1\,265 and 2\,217 for ITE, i.e.\
$2.4$--$4.1$ times fewer, and it reaches $10^{-8}$~Ha in 9, 24 and 45 steps against 89, 201 and 369. The mean-field variant performs almost as well as exact AITE at equilibrium, but, as discussed above, its advantage decreases with correlation.

These results suggest combining the two methods: ITE where it is cheaper, AITE where it is faster. A practical switching criterion follows from the skewness itself. In the
two-level regime $\kappa\simeq\sqrt{\Delta/\varepsilon}$, so a large measured skewness signals proximity to the ground state relative to the gap, without any knowledge of $E_{\rm GS}$. Switching from ITE to AITE once $\kappa\geq30$ occurs at
$\varepsilon\simeq0.003\,\Delta$ for all three geometries and yields gains of $2.6$--$3.0$ at $10^{-10}$~Ha. It is the most efficient practical choice at strong correlation, within 10\% of a switch placed with exact knowledge of the error.}

\subsection*{Discussion}

We have introduced an augmented ima\-gi\-nary-time evolution (AITE) framework that replaces the standard gradient flow on the energy landscape with a geometrically informed descent along locally optimal directions. The resulting flow strictly outperforms standard ITE and exhibits three qualitatively distinct regimes: a linear and a superlinear convergence regimes, followed by a finite-time extinction regime. 

Since AITE recovers standard ITE as its zero-skewness limit, the acceleration is not tied to any particular implementation. As sketched in Fig.~\ref{fig_AITE}d), the substitution $\hat{H} \rightarrow \hat{H}_{\mathrm{A}}$ offers a natural upgrade path for every branch of the ITE algorithmic family, requiring targeted modifications to the underlying algorithmic structure. In fact, beyond requiring statistical information of the instantaneous state, AITE amounts to replacing the standard propagator with an augmented one that{\color{black}, integrated along the flow, is exactly a Gaussian filter $e^{-B(\hat H-c)^2}$ (see Methods)}~\cite{Irmejs2024efficientquantum}, 
making it directly accessible within existing implementations. In \textit{direct ITE}, the imaginary-time propagator $\exp[-\tau\hat{H}]$ can be replaced by $\exp[-\tau\hat{H}_{\mathrm{A}}]$, yielding a geometrically informed filter that reaches the ground state at finite time {\color{black}in the flow parameter}. In \textit{stochastic ITE} (auxiliary-field quantum Monte Carlo~\cite{PhysRevLett.90.136401}, diffusion Monte Carlo \cite{10.1063/5.0202800}, and full configuration-interaction quantum Monte Carlo~\cite{Booth2009}), the effective Hamiltonian governing the importance-sampling weights can be replaced by its augmented counterpart, biasing the walker dynamics toward steeper descent directions.
In \textit{variational ITE} (variational quantum imaginary-time evolution~\cite{McArdle2019}, variational Monte Carlo~\cite{PhysRevB.16.3081}, neural quantum states \cite{Lange_2024}, and contracted variational eigensolvers \cite{doi:10.1021/acs.jctc.2c00446,Warren_2025}), the gradient vector $C_i(\theta)$ can be replaced by its augmented counterpart $C_i^{\mathrm{A}}(\theta)$, steering the parameter flow along provably steeper descent directions without changing the ansatz or the circuit structure. In \textit{double-bracket ITE}~\cite{gluza2025}, the bracket generator $[\hat{H}, \hat{\rho}]$ can be promoted to $[\hat{H}, \hat{\rho}^{\mathrm{A}}]$ {\color{black}(Eq.~\eqref{rhoA})} or, alternatively, $[\hat{H}_{\mathrm{A}}, \hat{\rho}]$ {\color{black}(Eq.~\eqref{eq:effectiveH})}, accelerating the flow while preserving its unitary structure. Finally, in \textit{spectral-filter ITE} (quantum singular value transformation \cite{10.1145/3313276.3316366} and quantum eigenvalue transformation of unitaries~\cite{Dong2022}), {\color{black}the integrated AITE flow is itself a Gaussian filter, Eq.~\eqref{eq:gaussfilter}, which can be implemented as a single polynomial of $\hat H$. Its query complexity is of the same order as that of near-optimal filtering methods and, unlike the shifted exponential filter of ITE, it does not require a tight lower bound on the ground-state energy (Methods).}

This apparent accessibility comes with method-de\-pen\-dent overheads, set by whether a given formulation only requires moments of the instantaneous state or must explicitly realize the augmented propagator. The moment-es\-ti\-mation cost is common to all variants of AITE: $\langle \hat H^2\rangle$ is essentially free in repeated-action methods, since $\langle \hat H^2\rangle=|\hat H|\psi\rangle|^2$, but becomes representation-limited in local-compilation approaches such as ITPP, QITE, MITE and PITE, where $\hat H^2$ generates pairwise products of elementary Hamiltonian terms; the moment $\langle \hat H^3\rangle$, which can be avoided by using the mean-field skewness heuristic, produces the corresponding triple products \cite{Saad1992,HochbruckLubich1997,GomezLurbe2026,Motta2020,Mao2023,Kosugi2021,Xie2024}. In tensor-network formulations, the same proliferation appears as an increase of the relevant MPO objects from $D$ to $D^2$ or $D^3$ before compression \cite{Vidal2004,Haegeman2011}, while LCU constructions inherit enlarged decompositions and normalization factors \cite{ChildsWiebe2012,Berry2015}. Spectral-transform methods are exceptional in that these powers remain low-degree functions of the spectrum; in QSVT, $\langle \hat H^2\rangle$ can even be extracted directly from a block-encoding ancilla population, although the analogous shortcut is absent for the odd moment $\langle \hat H^3\rangle$ \cite{Gilyen2019,Dong2022}. In stochastic and variational projector methods the limitation is mainly statistical rather than algebraic: VMC and neural quantum states can estimate $\langle \hat H^2\rangle$ from squared local energies, whereas $\langle \hat H^3\rangle$ requires access to $(\hat H^2\psi(x))/\psi(x)$; AFQMC, DMC, FCIQMC and PIGS instead incur purer estimators, longer back-propagation or forward-walking procedures, midpoint insertions, or higher-hop determinant connectivity \cite{Foulkes2001,Carleo2017,Lee2022,Zhang2018,Booth2009,Yan2018,nappi2026}. This broadly shared moment overhead should be separated from the stronger requirement of implementing $\exp[-\tau \hat{H}_A]$, which arises only in formulations whose primitive is an explicit projector or spectral filter. Such propagation is benign for state-vector or Krylov evolution, natural in QSVT/QETU as the scalar filter $\exp[-\tau(a_1x+a_2x^2)]$, but substantially more costly for local, LCU, and tensor-network representations, where the quadratic products become part of the generator itself \cite{Saad1992,HochbruckLubich1997,Gilyen2019,Dong2022,ChildsWiebe2012,Berry2015,Vidal2004}. It is least na\-tu\-ral for stochastic projector algorithms, whose native Hubbard--Stratonovich, drift--diffusion--branching or short-time-action structures are not generically preserved by replacing $\hat H$ with $\hat{H}_A$ \cite{Lee2022,Zhang2018,Foulkes2001,Yan2018}. In contrast, variational and double-bracket formulations need not apply $\exp[-\tau\hat{H}_A]$ as a primitive; the augmented Hamiltonian enters instead through projected gradients, local estimators, or commutator generators \cite{McArdle2019,Carleo2017,gluza2025,Suzuki2025}. {\color{black}A query-complexity analysis of the AITE filter is given in the Methods section; detailed implementations of the individual AITE variants and compiled resource estimates for specific Hamiltonians are left to future work.}

%TC:ignore

\section*{Methods}
Here we present a detailed account of the construction of the double-bracket flow in experimentally accessible subspaces, the optimization of the descent direction, and the construction of the augmented Hamiltonian $\hat{H}_{\mathrm{A}}(\tau)$. We also generalize our finite-extinction proof. 

\subsection*{Projected double-bracket flow}

Let $\hat{M}$ be a Hermitian operator and consider the associated two-dimensional 
Krylov subspace 
\begin{align}
 \mathcal{S}_M(\tau) =  \mathrm{span} \left\{\ket{\psi(\tau)},
    \ket{v_M(\tau)} \equiv \frac{\bar{M}\ket{\psi(\tau)}}{\sqrt{\mu_2^M(\tau)}}\right\},
\end{align}
where $\bar{M} \equiv \hat{M} - M(\tau)$ is the centered operator with 
$M(\tau) = \bra{\psi(\tau)}\hat{M}\ket{\psi(\tau)}$, and
$\mu_n^M(\tau) \equiv \bra{\psi(\tau)}\bar{M}^n\ket{\psi(\tau)}$ is the $n$th central 
moment of $\hat{M}$ in the state $\ket{\psi(\tau)}$.
The operator of core interest is the double commutator
$\hat{\chi}(\tau) \equiv -\bigl[\hat{H},\bigl[\hat{H},\hat{\rho}(\tau)\bigr]\bigr]$, 
$\hat{\rho}(\tau) \equiv \ket{\psi(\tau)}\bra{\psi(\tau)}$. Restricting $\hat{\chi}(\tau)$ to $\mathcal{S}_M(\tau)$ yields the matrix representation 
$\chi(\tau) \equiv \hat{\chi}(\tau)\big|_{\mathcal{S}_M(\tau)}$,
\begin{align}
\label{matrixchi}
    \chi(\tau)
    =
    \begin{pmatrix}
    \bra{\psi(\tau)}\hat{\chi}\ket{\psi(\tau)} & 
    \bra{\psi(\tau)}\hat{\chi}\ket{v_M(\tau)} \\[6pt]
    \bra{v_M(\tau)}\hat{\chi}\ket{\psi(\tau)} & 
    \bra{v_M(\tau)}\hat{\chi}\ket{v_M(\tau)}
    \end{pmatrix}.
\end{align}
We now specialize to $\hat{M} = \hat{H}$ and, for brevity, write 
$\mu_n(\tau) \equiv \mu_n^H(\tau)$.
In this case, the restriction of the double commutator to $\mathcal{S}_H(\tau)$ takes the form of the matrix presented in Eq.~\eqref{matrix}.

We aim to minimize the linear functional
\begin{align}
\mathcal{F}_{\chi(\tau)}[\hat{\mathcal{O}}] \equiv \Tr\bigl[\chi(\tau)\,\hat{\mathcal{O}}\bigr]
\end{align}
over the set of (Hermitian) projectors acting on $\mathcal{S}_H(\tau)$. Since $\mathcal{S}_H(\tau)$ is two-dimensional, the minimum is attained within the set of rank-$1$ projectors $\hat{\mathcal{O}} = \ket{\phi}\bra{\phi}$, for which the 
functional reduces to the Rayleigh quotient
\begin{align}
\mathcal{F}_{\chi(\tau)}[\ket{\phi}\bra{\phi}] = \Tr\bigl[\chi(\tau)\ket{\phi}\bra{\phi}\bigr] 
= \bra{\phi}\chi(\tau)\ket{\phi}.
\end{align}
By the min-max theorem \cite{HornJohnson2012}, the minimum of the Rayleigh quotient over all unit vectors 
$\ket{\phi} \in \mathcal{S}_H(\tau)$ is the smallest eigenvalue of $\chi(\tau)$, 
attained at the corresponding eigenvector. 
The global minimum over all Hermitian projectors on $\mathcal{S}_H(\tau)$ is 
therefore
\begin{align}
\min_{\hat{\mathcal{O}}}\,\mathcal{F}_{\chi(\tau)}[\hat{\mathcal{O}}]
 = -\sqrt{4\mu_2^2(\tau)\big[1 + \tfrac14\kappa^2(\tau)\big]},
\label{eq:bound}
\end{align}
and the minimum is attained uniquely at the rank-one projector
\begin{align}
\hat{\rho}^{\mathrm{A}}(\tau) = \ket{\lambda_-(\tau)}\bra{\lambda_-(\tau)},
\end{align}
where 
$\ket{\lambda_-(\tau)} = \cos(\phi(\tau))\ket{\psi(\tau)} 
+ \sin(\phi(\tau))\ket{v_H(\tau)}$. The mixing angle $\phi(\tau)$ is determined by the eigenvalue equa\-tion for $\chi(\tau)$. A direct application of the double-an\-gle formula yields $\tan\bigl(2\phi(\tau)\bigr) = \kappa(\tau) /2$.

\subsection*{Augmented Hamiltonians} 

We now seek a generalized Hamiltonian of the form
\begin{equation}
\label{eq:K-ansatzMethod}
  \hat{H}_{\mathrm{A}}(\tau) = a_1\bar{H} + a_2\bar{H}^2,
  \qquad a_1,a_2\in\mathbb{R},
\end{equation}
such that normalized imaginary-time evolution under $\hat{H}_{\mathrm{A}}(\tau)$ 
reproduces, to leading order in imaginary time, the target double-bracket flow:
\begin{equation}
\label{eq:goal}
  \frac{e^{-\beta \hat{H}_{\mathrm{A}}(\tau)}|\psi\rangle}
  {\|e^{-\beta \hat{H}_{\mathrm{A}}(\tau)}|\psi\rangle\|}
  \approx
  e^{-s[\hat{H},\hat{\rho}^{\mathrm{A}}(\tau)]}|\psi\rangle.
\end{equation}
A direct computation yields the exact decomposition
\begin{equation}
\label{eq:HO-decompMethod}
  [\hat{H},\hat{\rho}^{\mathrm{A}}(\tau)]|\psi\rangle
  = \Omega_{\rm eff}(\phi)\,|v_H\rangle
    + \frac{\sin(2\phi)}{2\sqrt{\mu_2}}\ket{\perp},
\end{equation}
where
$\Omega_{\rm eff}(\phi)
  = \sqrt{\mu_2}\cos(2\phi) + \frac{\mu_3}{2\mu_2}\sin(2\phi)$,
and
\begin{align*}
\ket{\perp} = \bar{H}^2|\psi\rangle - \mu_2|\psi\rangle
  - \frac{\mu_3}{\sqrt{\mu_2}}|v_H\rangle.
\end{align*}
Note that $\ket{\perp}$ is not normalized:
$\|\ket{\perp}\|^2 = \mu_4 - \mu_2^2 - \mu_3^2/\mu_2$.

The target state to first order in $s$ is therefore
\begin{equation}
\label{eq:targetMethod}
  e^{-s[\hat{H},\hat{\rho}^{\mathrm{A}}(\tau)]}|\psi\rangle
  = |\psi\rangle
    - s\,\Omega_{\rm eff}(\phi)\,|v_H\rangle
    - \frac{s\sin(2\phi)}{2\sqrt{\mu_2}}\ket{\perp} + \mathcal{O}(s^2).
\end{equation}
Expanding $e^{-\beta \hat{H}_{\mathrm{A}}}$ and normalizing:
\begin{equation}
\label{eq:ITE-expandMethod}
  \frac{e^{-\beta \hat{H}_{\mathrm{A}}}|\psi\rangle}
  {\|e^{-\beta \hat{H}_{\mathrm{A}}}|\psi\rangle\|}
  \approx |\psi\rangle
    - \beta\,\bigl(\hat{H}_{\mathrm{A}} - 
    \langle \hat{H}_{\mathrm{A}}\rangle_\psi\bigr)|\psi\rangle + \mathcal{O}(\beta^2).
\end{equation}
Matching Eq.~\eqref{eq:ITE-expandMethod} to Eq.~\eqref{eq:targetMethod} requires 
equating components along $|v_H\rangle$ and $\ket{\perp}$ separately.
Projecting onto $|v_H\rangle$ and using the ansatz~\eqref{eq:K-ansatzMethod} gives
\begin{equation}
\label{eq:match-e1-explicitMethod}
  \frac{\beta}{\sqrt{\mu_2}}\bigl(a_1\mu_2 + a_2\mu_3\bigr)
  = s\,\Omega_{\rm eff}(\phi),
\end{equation}
where we used $\bra{v_H}\bar{H}|\psi\rangle = \sqrt{\mu_2}$ and 
$\bra{v_H}\bar{H}^2|\psi\rangle = \mu_3/\sqrt{\mu_2}$.
Projecting onto $\ket{\perp}$ and using 
$\bra{\perp}\psi\rangle = 0$ gives
\begin{equation}
\label{eq:match-perp-explicitMethod}
  \beta a_2\,\|\ket{\perp}\|^2
  = \frac{s\sin(2\phi)}{2\sqrt{\mu_2}}\,\|\ket{\perp}\|^2,
\end{equation}
from which
$a_2 = (s/\beta)\sin(2\phi)/(2\sqrt{\mu_2})$. 
Finally, substituting into Eq.~\eqref{eq:match-e1-explicitMethod} yields $a_1 = \frac{s}{\beta}\,\cos(2\phi)$. 

Alternatively, one can construct a different augmented Hamiltonian of the form
\begin{align}
\hat{H}'_{\mathrm{A}}(\tau) = a'_1\bar{H} + a'_2\bar{K},
\label{eq:altH}
\end{align}
where $\hat{K}$ is an operator obtained by downfolding $\bar{H}^2$ and retaining 
terms up to a prescribed many-body rank, and $a'_1$, $a'_2$ are coefficients determined by the same matching condition as before. As explained in the Supplemental Material, this provides a computationally cheaper alternative to the full $\bar{H}^2$.

\begin{figure}[t]
\begin{tikzpicture}[>=Stealth, line cap=round, line join=round]
  \def\xzero{0}
  \def\xbasin{2.1}
  \def\xfull{7.2}

  % axis
  \draw[-, thick] (0,0) -- (7.7,0) node[right, font=\footnotesize] {\Large $\eps$};

  \draw[very thick, blue!55!black] (\xbasin,0.05) -- (\xfull,0.05);
  \fill[blue!55!black] (\xfull,0.05) circle (1.6pt);
  \fill[blue!55!black] (\xbasin,0.05) circle (1.6pt);

  % extinction regime: (0, Delta/4]
  \draw[very thick, red!70!black] (\xzero,0.05) -- (\xbasin,0.05);
  \fill[red!70!black] (\xzero,0.05) circle (1.6pt);

  % tick marks and labels
  \draw (0,0.08) -- (0,-0.08) node[below, font=\footnotesize] {$0$};
  \draw (\xbasin,0.08) -- (\xbasin,-0.08) node[below, font=\footnotesize] {$\Delta/4$};
  \draw (\xfull,0.08) -- (\xfull,-0.08) node[below, font=\footnotesize] {$\Delta$};

  % region labels above the axis
  \node[font=\footnotesize, align=center, red!70!black] at ({(\xzero+\xbasin)/2},0.65)
        {extinction regime\\[-1pt]Theorem~\ref{thm:extinction}, Eq.~\eqref{eq:sqrtbound}};
  \node[font=\footnotesize, align=center, blue!55!black] at ({(\xbasin+\xfull)/2},0.65)
        {basin-entry regime\\[-1pt]$|\dot\eps|\ge2\eps(\Delta-\eps)$};

  % flow direction
  \draw[->, thick, gray] (8.0,-0.7) -- (0,-0.7);
  \node[gray, font=\scriptsize] at (4.8,-0.95)
        {AITE flow: $\eps(\tau)$ decreases monotonically};
\end{tikzpicture}
\caption{Anatomy of the proof of Theorem~\ref{thm:extinction}. A trajectory starting at any $\eps_0\in(0,\Delta)$ first crosses the blue \emph{basin-entry regime}, where a crude but universally valid bound,
$|\dot\eps|\ge2\eps(\Delta-\eps)$, already guarantees arrival at $\eps\le\Delta/4$. From there, the sharper bound of Eq.~\eqref{eq:sqrtbound} takes over in the red \emph{extinction regime}, forcing the trajectory all the way to $\eps=0$ at a finite extinction time $\tau^*$.}
    \label{fig:placeholder}
\end{figure}

\subsection*{Finite-time extinction beyond the two-level regime} 

The closed-form solution of Eq.~\eqref{eq:exact_sol} and the extinction time of Eq.~\eqref{eq:tau_star} were derived in the near-convergence regime, where the state is dominated by its ground-state component and the moments reduce to $\mu_2\simeq\Delta\eps$, $\mu_3\simeq\Delta^2\eps$. We now show that finite-time extinction is not an artifact of this reduction.

\begin{lema}[Moment bounds]\label{lem:moments}
Let $\hat H = \sum_k E_k \ket{E_k}\bra{E_k}$ have a nondegenerate ground state, spectral gap $\Delta = E_1-E_{\rm GS}>0$, and
let $\ket{\psi}$ be any normalized state with energy error
$\eps=\langle\hat H\rangle-E_{\rm GS}$, $0<\eps<\Delta$. Write
$\Delta_k=E_k-E_{\rm GS}$ and $p_k=|\braket{k}{\psi}|^2$. Then
\begin{equation}
  \mu_2 \,\ge\, \eps\,(\Delta-\eps),
  \qquad
  \mu_3 \,\ge\, (\Delta-\eps)\,\mu_2 \,-\, p_0\,\eps^2\Delta .
  \label{eq:momentbounds}
\end{equation}
\end{lema}
\begin{proof}
Since $\Delta_k\ge\Delta$, $\sum_k p_k\Delta_k^2\ge\Delta\sum_{k\ge1}p_k\Delta_k=\Delta\eps$, and $\mu_2=\sum_k p_k\Delta_k^2-\eps^2\ge\eps(\Delta-\eps)$. For the third moment, split off the ground-state term: $\mu_3=-p_0\eps^3+\sum_{k\ge1}p_k(\Delta_k-\eps)^3$. For $k\ge1$ and $\eps<\Delta$ one has $\Delta_k-\eps\ge\Delta-\eps>0$, hence $(\Delta_k-\eps)^3\ge(\Delta-\eps)(\Delta_k-\eps)^2$, and $\sum_{k\ge1}p_k(\Delta_k-\eps)^2=\mu_2-p_0\eps^2$. Therefore, $\mu_3\ge(\Delta-\eps)(\mu_2-p_0\eps^2)-p_0\eps^3
=(\Delta-\eps)\mu_2-p_0\eps^2\Delta$.
\end{proof}

\begin{thm}[Finite-time extinction]
\label{thm:extinction}
Along the exact AITE flow,
$\dot\eps=-\sqrt{4\mu_2^2+\mu_3^2/\mu_2}$, the energy error obeys, for all
$\eps\in(0,\Delta/4]$,
\begin{equation}
  \dot\eps \;\le\; -\,\frac{5}{16}\,\Delta^{3/2}\sqrt{\eps}\,.
  \label{eq:sqrtbound}
\end{equation}
Consequently, once a trajectory satisfies $\eps(\tau_1)=\eps_1\le\Delta/4$,
it reaches $\eps=0$ exactly, at a time
\begin{equation}
  \tau^* \;\le\; \tau_1+\frac{32}{5}\,
  \frac{1}{\Delta}\sqrt{\frac{\eps_1}{\Delta}}
  \;\le\; \tau_1+\frac{16}{5\Delta}\,.
  \label{eq:tstarbound}
\end{equation}
Moreover, any trajectory with $\eps(0)=\eps_0<\Delta$ enters this basin in
the finite time
$\tau_1\le\frac{1}{2\Delta}\ln\!\big(\tfrac{3\eps_0}{\Delta-\eps_0}\big)$.
\end{thm}
\begin{proof}
Fix $\eps\le\Delta/4$ and abbreviate $x=\eps/\Delta\in(0,\tfrac14]$. The
rate obeys
$|\dot\eps|=\sqrt{4\mu_2^2+\mu_3^2/\mu_2}\ge|\mu_3|/\sqrt{\mu_2}
\ge\mu_3/\sqrt{\mu_2}$ unconditionally. By Lemma~\ref{lem:moments} and $p_0\le1$,
\begin{align}
  \frac{\mu_3}{\sqrt{\mu_2}}
  \;\ge\; (\Delta-\eps)\sqrt{\mu_2}\,-\,\frac{\eps^2\Delta}{\sqrt{\mu_2}} .
\end{align}
%The last inequality is non-trivial  only when $\mu_3>0$; to confirm this, apply Lemma~\ref{lem:moments} with  $p_0\le1$ and $\mu_2\ge\eps(\Delta-\eps)$: for $\eps\le\Delta/4$ one has
%\begin{align*}
%\mu_3 &\;\ge\; (\Delta-\eps)\mu_2 - \eps^2\Delta 
%      \\&\;\ge\; \eps(\Delta-\eps)^2 - \eps^2\Delta 
%      \\& \;=\;   \eps\Delta(1 - 3x + x^2)\big|_{x=\eps/\Delta}
%      \;\ge\; \tfrac{5}{16}\,\eps\Delta \;>\; 0,
%\end{align*}
%where the last step uses $\min_{x\in(0,1/4]}(1-3x+x^2)=\tfrac{5}{16}$, attained at $x=\tfrac14$.
The right-hand side is increasing in $\mu_2$ (i.e., its derivative, $(\Delta-\eps)/2\sqrt{\mu_2}+\eps^2\Delta/2\mu_2^{3/2}$, is positive) so it is minimized at the smallest variance, $\mu_2=\eps(\Delta-\eps)$ from Lemma~\ref{lem:moments}, giving
\begin{align}
  \frac{\mu_3}{\sqrt{\mu_2}}
  &\ge\ (\Delta-\eps)^{3/2}\sqrt{\eps}
  -\frac{\Delta\,\eps^{3/2}}{\sqrt{\Delta-\eps}}
  \nonumber \\ &=\Delta^{3/2}\sqrt{\eps}\;\frac{1-3x+x^2}{\sqrt{1-x}} .
\end{align}
On $(0,\tfrac14]$, $1-3x+x^2$ is decreasing and $\sqrt{1-x}\le1$, hence the prefactor is bounded below by $1-3\cdot\tfrac14+\tfrac1{16}=\tfrac{5}{16}$. 
%{\color{blue} Question: $|\dot{\eps}| \geq \mu_3/\sqrt{\mu_2} \geq 5/16 \Delta^{3/2} \sqrt{\eps}$ gives both conditions $\dot{\eps} \leq - 5/16 \Delta^{3/2} \sqrt{\eps}$, which is Eq. 29, and $\dot{\eps} \geq  5/16 \Delta^{3/2} \sqrt{\eps}$ since $\Delta^{3/2} \sqrt{\eps}>0$. Why are we not considering the other?  } 
%{\color{black} Answer: Since $\dot\eps\le0$ along the flow by construction (Eq.~\eqref{eq:flow_velocity}),  we have $|\dot\eps|=-\dot\eps$, so the bound $|\dot\eps|\ge\mu_3/\sqrt{\mu_2}\ge\tfrac{5}{16}\Delta^{3/2}\sqrt{\eps}$ translates unambiguously into $\dot\eps\le-\tfrac{5}{16}\Delta^{3/2}\sqrt{\eps}$, which is Eq.~\eqref{eq:sqrtbound}.}
This proves Eq.~\eqref{eq:sqrtbound}. Setting now $u=\sqrt{\eps}$, inequality \eqref{eq:sqrtbound} reads $\dot u\le-\tfrac{5}{32}\Delta^{3/2}$, so $u$ reaches zero no later than $\tau^* = \tau_1+\tfrac{32}{5}\Delta^{-3/2}\sqrt{\eps_1}$, which is Eq.~\eqref{eq:tstarbound}. The trajectory is well defined up to that time: the flow's generator is smooth in $\ket{\psi}$ wherever $\mu_2>0$, so the solution exists and is unique until extinction, and $\eps\equiv0$ (the ground state, which is a stationary point of the flow) continues it. For basin entry, use the complementary bound $|\dot\eps|\ge2\mu_2\ge2\eps(\Delta-\eps)$, valid for all $\eps<\Delta$, and integrate the separable inequality from $\eps_0$ down to $\Delta/4$: $\tau_1\le\int_{\Delta/4}^{\eps_0}\frac{d\eps}{2\eps(\Delta-\eps)} =\frac{1}{2\Delta}\ln\!\frac{3\,\eps_0}{\Delta-\eps_0}$.
\end{proof}

Notice that the proof uses only $\Delta_k\ge\Delta$, so arbitrary energy-level structure above the gap is allowed. Ground state degeneracy is also allowed after reading $p_0$ as the total ground-space population and $\Delta$ as the gap above it. For $\eps_0\ge\Delta$, however, the entry estimate does not apply; entry into the basin then follows from the strict descent $\dot\eps\le-2\mu_2<0$ away from eigenstates together with the standard nonzero ground-overlap assumption, and is observed in our numerical experiments. Finally, we notice that the constant $\tfrac{5}{16}$ is not optimal: in the near-convergence regime the sharp rate is $\sqrt{\Delta^3\eps}\,(1+\mathcal{O}(\eps/\Delta))$, recovering the extinction time in Eq.~\eqref{eq:tau_star} with unit constant.

{\color{black}\subsection*{Computational cost of the AITE implementation}

The augmentation of Eq.~\eqref{improvedDB} can enter the flow in at least two distinct ways, which we refer to as the operator and projector pictures. In the \emph{operator picture}, the generator acting on the state is modified, $\bar H\to\hat H_{\mathrm A}(\tau)$, and the flow is realized as the non-unitary update of Eq.~\eqref{eq:ite_update}, with $\hat H_{\mathrm A}(\tau)$ given by Eq.~\eqref{eq:effectiveH}. In the \emph{projector picture}, the Hamiltonian is left unchanged and instead the projector inside the double bracket is modified, $\hat\rho\to\hat\rho^{\mathrm A}(\tau)$, so that the flow is realized as the unitary double-bracket rotation generated by $[\hat H,\hat\rho^{\mathrm A}(\tau)]$. The two pictures generate the same continuous flow within $\mathcal S_H(\tau)$, but they
lead to different implementations with different costs.

In the operator picture the per-step primitive is $\hat H_{\mathrm{A}}$, a degree-two polynomial in $\bar H$. The overhead relative to the host ITE method depends on the circuit model. Writing $\hat H=\sum_{k=1}^L h_kP_k$ for its Pauli decomposition, a \emph{query} to a black-box $\exp[-i\hat Ht]$ or to a block encoding of $\hat H$ costs $\hat H^2$ exactly two calls against $\hat H$'s one, with the $1/\sqrt{\mu_2}$ prefactor entering the subnormalization. If $\hat H$ is instead compiled term-by-term, as in Trotter or other explicit-circuit implementations, $\bar H^2$ generically produces up to $L(L+1)/2$ Pauli strings before cancellation, so the overhead scales as $\mathcal{O}(L)$. Near extinction, $\mu_2\to0$ and the $1/\sqrt{\mu_2}$ prefactor grows in either model. On the support of the state, $\|\hat H_{\mathrm A}\|\simeq\Delta^2/(2\sqrt{\mu_2})\simeq(\tau^*-\tau)^{-1}$, so a step of fixed accuracy $c$ obeys $\delta\tau\lesssim c\,(\tau^*-\tau)$ and, by Eq.~\eqref{eq:power_law_extinction}, the number of steps needed to go from $\varepsilon=\Delta/4$ to a target error $\varepsilon_*$ is $N\simeq(2c)^{-1}\ln(\Delta/4\varepsilon_*)$. Although the flow parameter reaches its target in a bounded interval, the discretized cost therefore remains logarithmic in $1/\varepsilon_*$, as for ITE.

In the projector picture the per-step primitive is instead the double-bracket rotation $\exp[-s[\hat H,\hat\rho^{\mathrm A}]]$, whose
circuit is that of double-bracket ITE with $\hat\rho$ replaced by $\hat\rho^{\mathrm A}$ in the reflection. The generator remains bounded:
$\|[\hat H,\hat\rho^{\mathrm A}]\ket\psi\|=[\mathrm{Var}_\psi(\hat H_{\mathrm A})]^{1/2}$
exactly, and this tends to $\Delta/2$ in the convergence domain, whereas
the ITE generator $\sqrt{\mu_2}$ vanishes. The $1/\sqrt{\mu_2}$ prefactor
is, however, not removed but relocated to state preparation. Since
$2\phi\to\pi/2$ near convergence, $\ket{\lambda_-}$ retains weight
$\sin^2\phi\to1/2$ on the normalized vector
$\ket{v_H}=\bar H\ket\psi/\sqrt{\mu_2}$. Preparing it from a block encoding
of $\bar H$ with subnormalization $\alpha$ succeeds with amplitude
$\sqrt{\mu_2}/\alpha$ and requires $\mathcal O(\alpha/\sqrt{\mu_2})$ rounds
of amplitude amplification per reflection; with real-time queries alone,
the Mandelstam--Tamm bound~\cite{MandelstamTamm1945} likewise implies an
evolution time of order $1/\sqrt{\mu_2}$ to rotate $\ket\psi$ by an
$\mathcal O(1)$ angle toward $\ket{v_H}$. In addition, reflections about
the current state make the circuit depth grow geometrically with the
number of steps, as in double-bracket ITE~\cite{gluza2025}.

\subsection*{AITE as a spectral filter} 

A cost statement that does not depend on discretization follows from the operator picture. Each update of Eq.~\eqref{eq:ite_update} multiplies the amplitude of every eigenstate $\ket{E_k}$ of $\hat H$ by $e^{-\delta\tau\,h_\tau(E_k)}$, with $h_\tau(x)=a(\tau)[x-E(\tau)]+b(\tau)[x-E(\tau)]^2$, $a=\cos 2\phi$ and $b=\sin2\phi/(2\sqrt{\mu_2})$. Along a given trajectory $a$, $b$ and $E$ are scalar functions of $\tau$, and all factors commute, so that
\begin{align}
\ket{\psi(\tau)}\propto e^{-B(\tau)\,[\hat H-c(\tau)]^2}\ket{\psi_0},
\label{eq:gaussfilter}
\end{align}
with
\begin{align}
B(\tau)&=\int_0^\tau b\,dt, \nonumber\\
c(\tau)&=\frac{1}{B(\tau)}\left[\int_0^\tau b\,E\,dt-\frac12\int_0^\tau a\,dt\right].
\label{eq:Bc}
\end{align}
The same form holds exactly for the discretized flow, with sums replacing integrals. AITE therefore applies a Gaussian filter of width $B^{-1/2}$ centered at $c$. Near convergence $B$ diverges logarithmically as $\tau\to\tau^*$, so that finite extinction corresponds to a vanishing filter width, and $c\to E_{\rm GS}$ with an offset of order $1/B$. The whole trajectory can thus be realized as a single polynomial of $\hat H$ by QSVT/QETU~\cite{10.1145/3313276.3316366,Dong2022}, rather than by stepping.

To quantify its cost, normalize $\hat H$ to spectrum in $[-1,1]$, let $\gamma_0=|\langle\psi_{\rm GS}|\psi_0\rangle|$ be the initial overlap and $\eta$ the target infidelity, and write $L_\eta=\ln[1/(\gamma_0\sqrt\eta)]$. A filter $F$ with $|F|\le1$, implemented as a polynomial of degree $d$ followed by amplitude amplification, costs $Q=\mathcal O(d/[\gamma_0F(E_{\rm GS})])$ queries and reaches infidelity $\eta$ if $\max_{x\ge E_1}|F(x)|/F(E_{\rm GS})\le\gamma_0\sqrt\eta$. Using standard polynomial approximations of exponentials and Gaussians~\cite{SachdevaVishnoi2014}, the shifted ITE filter $e^{-\beta(x-L)}$, which requires a lower bound $L=E_{\rm GS}-\delta_L$ on the spectrum, and the AITE filter of Eq.~\eqref{eq:gaussfilter} centered at $c=E_{\rm GS}+\delta$, with $\delta<\Delta/2$, give
\begin{align}
Q_{\rm ITE}&=\mathcal O\!\left(\frac{L_\eta}{\gamma_0}\sqrt{\frac{1+\delta_L/\Delta}{\Delta}}\;e^{L_\eta\delta_L/\Delta}\right),
\label{eq:QITE}\\
Q_{\rm AITE}&=\mathcal O\!\left(\frac{L_\eta}{\gamma_0\Delta}\;e^{L_\eta\chi_\delta}\right),
\qquad \chi_\delta=\frac{\delta^2}{\Delta(\Delta-2\delta)},
\label{eq:QG}
\end{align}
the latter up to the corresponding change in polynomial degree. For a centered filter, $Q_{\rm AITE}=\mathcal O(L_\eta/\gamma_0\Delta)$, the scaling of near-optimal filtering-based ground-state preparation~\cite{LinTong2020}: finite extinction in $\tau$ does not yield a cost independent of the target precision. The comparison with ITE is controlled by the penalty exponents. For ITE it is linear in $\delta_L/\Delta$, whereas for the Gaussian filter it is quadratic in $\delta/\Delta$ for $|\delta|\ll\Delta$, and the filter also tolerates centers above $E_{\rm GS}$. The AITE filter is cheaper whenever $\delta_L\gtrsim\Delta\ln(1/\Delta)/(2L_\eta)$, i.e., unless $E_{\rm GS}$ is known in advance to a precision well below the gap; conversely, with $E_{\rm GS}$ known exactly and placed at the spectral edge, ITE benefits from Chebyshev edge acceleration and scales as $\Delta^{-1/2}$. The advantage of AITE in this setting is therefore robustness to the prior knowledge of $E_{\rm GS}$, together with self-centering: $c\to E_{\rm GS}$ is obtained from the measured moments alone.

Two caveats apply. First, $b\ge0$ requires $\kappa>0$. For small initial overlaps the energy distribution can be left-skewed at early times, $\kappa<0$, and the intermediate filters are then inverted Gaussians that cannot be block-encoded without a normalization penalty; this can be easily avoided by taking ITE steps ($\phi=0$) until $\kappa$ becomes positive. Second, the coefficients must be learned from sampled moments. With $M$ energy samples per step, the excited population $p\simeq\varepsilon/\Delta$ becomes unobservable once $p\lesssim1/M$, so a single flow that learns its own filter down to error $\varepsilon$ requires $M\gtrsim\Delta/\varepsilon$. This is avoided by a two-stage procedure in which the flow is used only to locate $E_{\rm GS}$ to a fraction of the gap (an offset $\delta=\Delta/10$ gives $\chi_\delta\simeq10^{-2}$), after which a fixed centered Gaussian with $B=L_\eta/\Delta^2$ is applied once.

None of these implementations changes the complexity class of ground-state preparation: the factors $1/\Delta$ and $1/\gamma_0$ in Eq.~\eqref{eq:QG} make the worst-case exponential scaling with system size explicit for Hamiltonians with an exponentially small gap or initial overlap, in agreement with Theorem~\ref{thm:extinction}. Compiled circuit depths for specific Hamiltonians and the shot budget of the learning stage are left to future work.}

\section*{Data and code availability}

The code and data used to produce the figures of this work are openly available on Zenodo~\cite{BenavidesRiveros2026AITEcode}.

\bibliography{Refs}

@article{Yuan2019theoryofvariational,
  doi = {10.22331/q-2019-10-07-191},
  url = {https://doi.org/10.22331/q-2019-10-07-191},
  title = {Theory of variational quantum simulation},
  author = {Yuan, X. and Endo, S. and Zhao, Q. and Li, Y. and Benjamin, S.},
  journal = {{Quantum}},
  issn = {2521-327X},
  publisher = {{Verein zur F{\"{o}}rderung des Open Access Publizierens in den Quantenwissenschaften}},
  volume = {3},
  pages = {191},
  month = oct,
  year = {2019}
}

@book{Bouchaud_Potters_2003, 
place={Cambridge}, edition={2}, 
title={{Theory of Financial Risk and Derivative Pricing: From Statistical Physics to Risk Management}}, publisher={Cambridge University Press}, 
author={Bouchaud, J.-P. and Potters, M.}, year={2003}}

@Article{10.21468/SciPostPhys.15.6.229,
	title={{Scalable imaginary time evolution with neural network quantum states}},
	author={E. Ledinauskas and E. Anisimovas},
	journal={SciPost Phys.},
	volume={15},
	pages={229},
	year={2023},
	publisher={SciPost},
	doi={10.21468/SciPostPhys.15.6.229},
	url={https://scipost.org/10.21468/SciPostPhys.15.6.229},
}

@article{PRXQuantum.2.010342,
  title = {{Real- and Imaginary-Time Evolution with Compressed Quantum Circuits}},
  author = {Lin, S. and Dilip, R. and Green, A. and Smith, A. and Pollmann, F.},
  journal = {PRX Quantum},
  volume = {2},
  issue = {1},
  pages = {010342},
  numpages = {15},
  year = {2021},
  month = {Mar},
  publisher = {American Physical Society},
  doi = {10.1103/PRXQuantum.2.010342},
  url = {https://link.aps.org/doi/10.1103/PRXQuantum.2.010342}
}

@article{Joanes1998,
 ISSN = {00390526, 14679884},
 URL = {http://www.jstor.org/stable/2988433},
 author = {D. Joanes and C. Gill},
 journal = {J. R. Stat. Soc. Ser. D Stat.},
 number = {1},
 pages = {183},
 publisher = {[Royal Statistical Society, Wiley]},
 title = {{Comparing Measures of Sample Skewness and Kurtosis}},
 urldate = {2025-12-30},
 volume = {47},
 year = {1998}
}

@INPROCEEDINGS{10367741,
  author={Gacon, J. and Zoufal, C. and Carleo, G. and Woerner, S.},
  booktitle={2023 IEEE International Conference on Quantum Computing and Engineering (QCE)}, 
  title={Stochastic Approximation of Variational Quantum Imaginary Time Evolution}, 
  year={2023},
  volume={03},
  number={},
  pages={129-139},
  doi={10.1109/QCE57702.2023.10367741}}

@inproceedings{pub.1195650887,
   author={Lee, X. and Lau, H.},
  booktitle={2025 IEEE International Conference on Quantum Computing and Engineering (QCE)}, 
  title={{Solving Constrained Combinatorial Optimization Problems with Variational Quantum Imaginary Time Evolution}}, 
  year={2025},
  volume={01},
  number={},
  pages={1955-1964},
  doi={10.1109/QCE65121.2025.00213}}

@article{PhysRev.96.1124,
  title = {{Properties of Bethe-Salpeter Wave Functions}},
  author = {Wick, G.},
  journal = {Phys. Rev.},
  volume = {96},
  issue = {4},
  pages = {1124},
  numpages = {0},
  year = {1954},
  month = {Nov},
  publisher = {American Physical Society},
  doi = {10.1103/PhysRev.96.1124},
  url = {https://link.aps.org/doi/10.1103/PhysRev.96.1124}
}

@article{gtq3-j37b,
  title = {Classical optimization with imaginary-time block encoding on quantum computers: The MaxCut problem},
  author = {Zhong, D. and Francis, A. and Rrapaj, E.},
  journal = {Phys. Rev. A},
  volume = {112},
  issue = {4},
  pages = {042420},
  numpages = {10},
  year = {2025},
  month = {Oct},
  publisher = {American Physical Society},
  doi = {10.1103/gtq3-j37b},
  url = {https://link.aps.org/doi/10.1103/gtq3-j37b}
}

@article{PhysRevA.109.052430,
  title = {Combinatorial optimization with quantum imaginary time evolution},
  author = {Bauer, N. and Alam, R. and Siopsis, G. and Ostrowski, J.},
  journal = {Phys. Rev. A},
  volume = {109},
  issue = {5},
  pages = {052430},
  numpages = {8},
  year = {2024},
  month = {May},
  publisher = {American Physical Society},
  doi = {10.1103/PhysRevA.109.052430},
  url = {https://link.aps.org/doi/10.1103/PhysRevA.109.052430}
}

@misc{dominguez2025,
      title={Causal Portfolio Optimization: Principles and Sensitivity-Based Solutions}, 
      author={A. Rodriguez Dominguez},
      year={2025},
      eprint={2504.05743},
      archivePrefix={arXiv},
      url={https://arxiv.org/abs/2504.05743}, 
}

@Inbook{AshtekarSchilling,
author="Ashtekar, A. and Schilling, T.",
editor="Harvey, A.",
title="Geometrical Formulation of Quantum Mechanics",
bookTitle="On Einstein's Path: Essays in Honor of Engelbert Schucking",
year="1999",
publisher="Springer New York",
address="New York, NY",
pages="23--65",
isbn="978-1-4612-1422-9",
doi="10.1007/978-1-4612-1422-9_3",
url="https://doi.org/10.1007/978-1-4612-1422-9_3"
}

@article{Motta2022,
author = {Motta, M. and Rice, J.},
title = {Emerging quantum computing algorithms for quantum chemistry},
journal = {WIREs Comput. Mol. Sci.},
volume = {12},
number = {3},
pages = {e1580},
doi = {https://doi.org/10.1002/wcms.1580},
url = {https://wires.onlinelibrary.wiley.com/doi/abs/10.1002/wcms.1580},
year = {2022}
}

@article{Cao2019,
author = {Cao, Y. and others},
title = {{Quantum Chemistry in the Age of Quantum Computing}},
journal = {Chem. Rev.},
volume = {119},
number = {19},
pages = {10856},
year = {2019},
doi = {10.1021/acs.chemrev.8b00803},
URL = {https://doi.org/10.1021/acs.chemrev.8b00803}
}

@article{AnandanAharonov,
  title = {Geometry of quantum evolution},
  author = {Anandan, J. and Aharonov, Y.},
  journal = {Phys. Rev. Lett.},
  volume = {65},
  issue = {14},
  pages = {1697},
  numpages = {0},
  year = {1990},
  month = {Oct},
  publisher = {American Physical Society},
  doi = {10.1103/PhysRevLett.65.1697},
  url = {https://link.aps.org/doi/10.1103/PhysRevLett.65.1697}
}

@article{BraunsteinCaves,
  title = {Statistical distance and the geometry of quantum states},
  author = {Braunstein, S. and Caves, C.},
  journal = {Phys. Rev. Lett.},
  volume = {72},
  issue = {22},
  pages = {3439},
  numpages = {0},
  year = {1994},
  month = {May},
  publisher = {American Physical Society},
  doi = {10.1103/PhysRevLett.72.3439},
  url = {https://link.aps.org/doi/10.1103/PhysRevLett.72.3439}
}

@Article{Bach2010,
author={Bach, V.
and Bru, J.},
title={{Rigorous foundations of the Brockett--Wegner flow for operators}},
journal={J. Evol. Equ.},
year={2010},
month={May},
day={01},
volume={10},
number={2},
pages={425},
issn={1424-3202},
doi={10.1007/s00028-010-0055-1},
url={https://doi.org/10.1007/s00028-010-0055-1}
}

@Article{Cirillo2020,
author={Cirillo, P.
and Taleb, N.},
title={Tail risk of contagious diseases},
journal={Nat. Phys.},
year={2020},
month={Jun},
day={01},
volume={16},
number={6},
pages={606},
issn={1745-2481},
doi={10.1038/s41567-020-0921-x},
url={https://doi.org/10.1038/s41567-020-0921-x}
}

@article{Groeneveld1984,
 ISSN = {00390526, 14679884},
 URL = {http://www.jstor.org/stable/2987742},
 author = {R. Groeneveld and G. Meeden},
 journal = {J. R. Stat. Soc. Ser. D Stat.},
 number = {4},
 pages = {391},
 publisher = {[Royal Statistical Society, Wiley]},
 title = {Measuring Skewness and Kurtosis},
 urldate = {2025-12-09},
 volume = {33},
 year = {1984}
}

@misc{shrikhande2025rapidgroundstateenergy,
      title={{Rapid ground state energy estimation with a Sparse Pauli Dynamics-enabled Variational Double Bracket Flow}}, 
      author={C. Shrikhande and A. Bachhar and A. Rodriguez Jimenez and N. Mayhall},
      year={2025},
      eprint={2511.21651},
      archivePrefix={arXiv},
      url={https://arxiv.org/abs/2511.21651}, 
}

@article{BROCKETT199179,
title = {Dynamical systems that sort lists, diagonalize matrices, and solve linear programming problems},
journal = {Linear Algebra Appl.},
volume = {146},
pages = {79},
year = {1991},
issn = {0024-3795},
doi = {https://doi.org/10.1016/0024-3795(91)90021-N},
url = {https://www.sciencedirect.com/science/article/pii/002437959190021N},
author = {R. Brockett}
}

@misc{melendez2025,
      title={{Adaptive time Compressed QITE (ACQ) and its geometrical interpretation}}, 
      author={A. Acevedo Meléndez and C. Almudéver and M. Garcia-March and R. Gómez-Lurbe and L. Ion and M. Bera and R. Sanz and S. Mehrabankar and T. Pandit and A. Pérez and A. Anglés-Castillo},
      year={2025},
      eprint={2510.15781},
      archivePrefix={arXiv},
      url={https://arxiv.org/abs/2510.15781}, 
}

@article{Warren_2025,
doi = {10.1088/2058-9565/ad9fa5},
url = {https://doi.org/10.1088/2058-9565/ad9fa5},
year = {2025},
month = {jan},
publisher = {IOP Publishing},
volume = {10},
number = {2},
pages = {02LT02},
author = {Warren, S. and Wang, Y. and Benavides-Riveros, C. L. and Mazziotti, D.},
title = {Quantum algorithm for polaritonic chemistry based on an exact ansatz},
journal = {Quantum Sci. Technol.}
}

@article{doi:10.1098/rsta.1895.0010,
author = {Pearson, K.},
title = {{Contributions to the mathematical theory of evolution.—II. Skew variation in homogeneous material}},
journal = {Philos. Trans. R. Soc. A.},
volume = {186},
number = {},
pages = {343},
year = {1895},
doi = {10.1098/rsta.1895.0010},
URL = {https://royalsocietypublishing.org/doi/abs/10.1098/rsta.1895.0010}
}

@Article{McArdle2019,
author={McArdle, S.
and Jones, T.
and Endo, S.
and Li, Y.
and Benjamin, S.
and Yuan, X.},
title={Variational ansatz-based quantum simulation of imaginary time evolution},
journal={npj Quantum Inf.},
year={2019},
month={Sep},
day={06},
volume={5},
number={1},
pages={75},
issn={2056-6387},
doi={10.1038/s41534-019-0187-2},
url={https://doi.org/10.1038/s41534-019-0187-2}
}

@misc{alghadeer2025,
      title={Double-Bracket Algorithmic Cooling}, 
      author={M. Alghadeer and others},
      year={2025},
      eprint={2510.00302},
      archivePrefix={arXiv},
      primaryClass={quant-ph},
      url={https://arxiv.org/abs/2510.00302}, 
}

@article{Kempe2006,
author = {Kempe, J. and Kitaev, A. and Regev, O.},
title = {{The Complexity of the Local Hamiltonian Problem}},
journal = {SIAM J. Comput.},
volume = {35},
number = {5},
pages = {1070},
year = {2006},
doi = {10.1137/S0097539704445226},
URL = {https://doi.org/10.1137/S0097539704445226}
}

@Article{D4FD00039K,
author ="Magnusson, E. and Fitzpatrick, A. and Knecht, S. and Rahm, M. and Dobrautz, W.",
title  ="Towards efficient quantum computing for quantum chemistry: reducing circuit complexity with transcorrelated and adaptive ansatz techniques",
journal  ="Faraday Discuss.",
year  ="2024",
volume  ="254",
issue  ="0",
pages  ="402",
publisher  ="The Royal Society of Chemistry",
doi  ="10.1039/D4FD00039K",
url  ="http://dx.doi.org/10.1039/D4FD00039K"}

@article{Booth2009,
    author = {Booth, G. and Thom, A. and Alavi, A.},
    title = {{Fermion Monte Carlo without fixed nodes: A game of life, death, and annihilation in Slater determinant space}},
    journal = {J. Chem. Phys.},
    volume = {131},
    number = {5},
    pages = {054106},
    year = {2009},
    month = {08},
    issn = {0021-9606},
    doi = {10.1063/1.3193710},
    url = {https://doi.org/10.1063/1.3193710}
}

@Article{Cianci2024,
author={Cianci, C.
and Santos, L.
and Batista, V.},
title={{Subspace-Search Quantum Imaginary Time Evolution for Excited State Computations}},
journal={J. Chem. Theory Comput.},
year={2024},
month={Oct},
day={22},
publisher={American Chemical Society},
volume={20},
number={20},
pages={8940-8947},
issn={1549-9618},
doi={10.1021/acs.jctc.4c00915},
url={https://doi.org/10.1021/acs.jctc.4c00915}
}

@article{PhysRevLett.93.207204,
  title = {{Matrix Product Density Operators: Simulation of Finite-Tem\-pe\-ra\-ture and Dissipative Systems}},
  author = {Verstraete, F. and Garc\'{\i}a-Ripoll, J. and Cirac, J.},
  journal = {Phys. Rev. Lett.},
  volume = {93},
  issue = {20},
  pages = {207204},
  numpages = {4},
  year = {2004},
  month = {Nov},
  publisher = {American Physical Society},
  doi = {10.1103/PhysRevLett.93.207204},
  url = {https://link.aps.org/doi/10.1103/PhysRevLett.93.207204}
}

@misc{zima2026,
      title={{Fast Tensor Network Imaginary Time Evolution by Implicit Stepping on Logarithmic Grids}}, 
      author={J. Zima and E. Stoudenmire and S. White and O. Parcollet and J. Kaye},
      year={2026},
      eprint={2606.02930},
      archivePrefix={arXiv},
      url={https://arxiv.org/abs/2606.02930}, 
}

@Article{Motta2020,
author={Motta, M. and others},
title={Determining eigenstates and thermal states on a quantum computer using quantum imaginary time evolution},
journal={Nat. Phys.},
year={2020},
month={Feb},
day={01},
volume={16},
number={2},
pages={205},
issn={1745-2481},
doi={10.1038/s41567-019-0704-4},
url={https://doi.org/10.1038/s41567-019-0704-4}
}

@article{PhysRevLett.45.566,
  title = {{Ground State of the Electron Gas by a Stochastic Method}},
  author = {Ceperley, D. and Alder, B.},
  journal = {Phys. Rev. Lett.},
  volume = {45},
  issue = {7},
  pages = {566},
  numpages = {0},
  year = {1980},
  month = {Aug},
  publisher = {American Physical Society},
  doi = {10.1103/PhysRevLett.45.566},
  url = {https://link.aps.org/doi/10.1103/PhysRevLett.45.566}
}

@article{10.1063/5.0202800,
    author = {Caffarel, M. and Del Moral, P. and de Montella, L.},
    title = {{On the mathematical foundations of diffusion Monte Carlo}},
    journal = {J. Math. Phys.},
    volume = {66},
    number = {1},
    pages = {013301},
    year = {2025},
    month = {01},
    issn = {0022-2488},
    doi = {10.1063/5.0202800},
    url = {https://doi.org/10.1063/5.0202800},
}

@inproceedings{10.1145/3313276.3316366,
author = {Gily\'{e}n, A. and Su, Y. and Low, G. and Wiebe, N.},
title = {Quantum singular value transformation and beyond: exponential improvements for quantum matrix arithmetics},
year = {2019},
isbn = {9781450367059},
publisher = {Association for Computing Machinery},
address = {New York, NY, USA},
url = {https://doi.org/10.1145/3313276.3316366},
doi = {10.1145/3313276.3316366},
booktitle = {Proceedings of the 51st Annual ACM SIGACT Symposium on Theory of Computing},
pages = {193–204},
numpages = {12},
location = {Phoenix, AZ, USA},
series = {STOC 2019}
}

@article{Irmejs2024efficientquantum,
  doi = {10.22331/q-2024-06-27-1389},
  url = {https://doi.org/10.22331/q-2024-06-27-1389},
  title = {Efficient {Q}uantum {A}lgorithm for {F}iltering {P}roduct {S}tates},
  author = {Irmejs, R. and Ba{\~{n}}uls, M. and Cirac, J.},
  journal = {{Quantum}},
  issn = {2521-327X},
  publisher = {{Verein zur F{\"{o}}rderung des Open Access Publizierens in den Quantenwissenschaften}},
  volume = {8},
  pages = {1389},
  month = jun,
  year = {2024}
}

@article{10.1063/1.4916647,
    author = {Tanimura, Y.},
    title = {{Real-time and imaginary-time quantum hierarchal Fokker-Planck equations}},
    journal = {J. Chem. Phys.},
    volume = {142},
    number = {14},
    pages = {144110},
    year = {2015},
    month = {04},
    issn = {0021-9606},
    doi = {10.1063/1.4916647},
    url = {https://doi.org/10.1063/1.4916647}
}

@article{doi:10.1021/acs.jctc.2c00446,
author = {Smart, S. and Mazziotti, D.},
title = {{Accelerated Convergence of Contracted Quantum Eigensolvers through a Quasi-Second-Order, Locally Parameterized Optimization}},
journal = {J. Chem. Theory Comput.},
volume = {18},
number = {9},
pages = {5286},
year = {2022},
doi = {10.1021/acs.jctc.2c00446},
URL = {https://doi.org/10.1021/acs.jctc.2c00446}
}

@book{holevo2011probabilistic,
  title     = {{Probabilistic and Statistical Aspects of Quantum Theory}},
  author    = {Holevo, A.},
  year      = {2011},
  publisher = {Edizioni della Normale},
  address   = {Pisa},
  edition   = {2nd},
  series    = {Publications of the Scuola Normale Superiore},
  isbn      = {978-8876423789}
}

@article{Lange_2024,
doi = {10.1088/2058-9565/ad7168},
url = {https://doi.org/10.1088/2058-9565/ad7168},
year = {2024},
month = {sep},
publisher = {IOP Publishing},
volume = {9},
number = {4},
pages = {040501},
author = {Lange, H. and Van de Walle, A. and Abedinnia, A. and Bohrdt, A.},
title = {From architectures to applications: a review of neural quantum states},
journal = {Quantum Sci. Technol.}
}

@article{PhysRevB.16.3081,
  title = {{Monte Carlo simulation of a many-fermion study}},
  author = {Ceperley, D. and Chester, G. and Kalos, M.},
  journal = {Phys. Rev. B},
  volume = {16},
  issue = {7},
  pages = {3081},
  numpages = {0},
  year = {1977},
  month = {Oct},
  publisher = {American Physical Society},
  doi = {10.1103/PhysRevB.16.3081},
  url = {https://link.aps.org/doi/10.1103/PhysRevB.16.3081}
}

@article{PhysRev.84.350,
  title = {{Bound States in Quantum Field Theory}},
  author = {Gell-Mann, M. and Low, F.},
  journal = {Phys. Rev.},
  volume = {84},
  issue = {2},
  pages = {350},
  numpages = {0},
  year = {1951},
  month = {Oct},
  publisher = {American Physical Society},
  doi = {10.1103/PhysRev.84.350},
  url = {https://link.aps.org/doi/10.1103/PhysRev.84.350}
}

@article{gluza2025,
  title = {{Double-Bracket Quantum Algorithms for Quantum Imaginary-Time Evolution}},
  author = {Gluza, M. and Son, J. and Tiang, B. and Zander, R. and Seidel, R. and Suzuki, Y. and Holmes, Z. and Ng, N. Y.},
  journal = {Phys. Rev. Lett.},
  volume = {136},
  issue = {2},
  pages = {020601},
  numpages = {12},
  year = {2026},
  month = {Jan},
  publisher = {American Physical Society},
  doi = {10.1103/rw81-k8vk},
  url = {https://link.aps.org/doi/10.1103/rw81-k8vk}
}

@article{PhysRevLett.130.050601,
  title = {{Quantum Approximate Optimization Algorithm Pseudo-Boltzmann States}},
  author = {D\'{\i}ez-Valle, P. and Porras, D. and Garc\'{\i}a-Ripoll, J.},
  journal = {Phys. Rev. Lett.},
  volume = {130},
  issue = {5},
  pages = {050601},
  numpages = {6},
  year = {2023},
  month = {Feb},
  publisher = {American Physical Society},
  doi = {10.1103/PhysRevLett.130.050601},
  url = {https://link.aps.org/doi/10.1103/PhysRevLett.130.050601}
}

@article{Stokes_2023,
doi = {10.1088/2632-2153/acc8b9},
url = {https://doi.org/10.1088/2632-2153/acc8b9},
year = {2023},
month = {may},
publisher = {IOP Publishing},
volume = {4},
number = {2},
pages = {021001},
author = {Stokes, J. and Chen, B. and Veerapaneni, S.},
title = {{Numerical and geometrical aspects of flow-based variational quantum Monte Carlo}},
journal = {Mach. Learn.: Sci. Technol.},
}

@article{Dong2022,
  title = {{Ground-State Preparation and Energy Estimation on Early Fault-Tolerant Quantum Computers via Quantum Eigenvalue Transformation of Unitary Matrices}},
  author = {Dong, Y. and Lin, L. and Tong, Y.},
  journal = {PRX Quantum},
  volume = {3},
  issue = {4},
  pages = {040305},
  numpages = {25},
  year = {2022},
  month = {Oct},
  publisher = {American Physical Society},
  doi = {10.1103/PRXQuantum.3.040305},
  url = {https://link.aps.org/doi/10.1103/PRXQuantum.3.040305}
}

@article{10.1063/1.431514,
    author = {Anderson, J.},
    title = {{A random‐walk simulation of the Schrö\-din\-ger equation: H$^+_3$}},
    journal = {J. Chem. Phys.},
    volume = {63},
    number = {4},
    pages = {1499},
    year = {1975},
    month = {08},
    issn = {0021-9606},
    doi = {10.1063/1.431514},
    url = {https://doi.org/10.1063/1.431514}
}

@article{CRAMA2003546,
title = {Simulated annealing for complex portfolio selection problems},
journal = {Eur. J. Oper. Res.},
volume = {150},
number = {3},
pages = {546},
year = {2003},
issn = {0377-2217},
doi = {https://doi.org/10.1016/S0377-2217(02)00784-1},
url = {https://www.sciencedirect.com/science/article/pii/S0377221702007841},
author = {Y. Crama and M. Schyns}
}

@book{Vasquez2006,
    author = {Vazquez, J.},
    title = {{The Porous Medium Equation: Mathematical Theory}},
    publisher = {Oxford University Press},
    year = {2006},
    month = {10},
    isbn = {9780198569039},
    doi = {10.1093/acprof:oso/9780198569039.001.0001},
    url = {https://doi.org/10.1093/acprof:oso/9780198569039.001.0001},
}

@article{g4ch-5x8m,
  title = {Constrained search in imaginary time},
  author = {Penz, M. and van Leeuwen, R.},
  journal = {Phys. Rev. A},
  volume = {112},
  issue = {3},
  pages = {032815},
  numpages = {17},
  year = {2025},
  month = {Sep},
  publisher = {American Physical Society},
  doi = {10.1103/g4ch-5x8m},
  url = {https://link.aps.org/doi/10.1103/g4ch-5x8m}
}

@article{PhysRevLett.90.136401,
  title = {{Quantum Monte Carlo Method using Phase-Free Random Walks with Slater Determinants}},
  author = {Zhang, S. and Krakauer, H.},
  journal = {Phys. Rev. Lett.},
  volume = {90},
  issue = {13},
  pages = {136401},
  numpages = {4},
  year = {2003},
  month = {Apr},
  publisher = {American Physical Society},
  doi = {10.1103/PhysRevLett.90.136401},
  url = {https://link.aps.org/doi/10.1103/PhysRevLett.90.136401}
}

@article{Stokes2020quantumnatural,
  doi = {10.22331/q-2020-05-25-269},
  url = {https://doi.org/10.22331/q-2020-05-25-269},
  title = {Quantum {N}atural {G}radient},
  author = {Stokes, J. and Izaac, J. and Killoran, N. and Carleo, G.},
  journal = {{Quantum}},
  issn = {2521-327X},
  publisher = {{Verein zur F{\"{o}}rderung des Open Access Publizierens in den Quantenwissenschaften}},
  volume = {4},
  pages = {269},
  month = may,
  year = {2020}
}

@article{KOSLOFF1986223,
title = {A direct relaxation method for calculating eigenfunctions and eigenvalues of the schrödinger equation on a grid},
journal = {Chem.  Phys. Lett.},
volume = {127},
number = {3},
pages = {223},
year = {1986},
issn = {0009-2614},
doi = {https://doi.org/10.1016/0009-2614(86)80262-7},
url = {https://www.sciencedirect.com/science/article/pii/0009261486802627},
author = {R. Kosloff and H. Tal-Ezer}
}

@article{Kirkpatrick1983,
author = {S. Kirkpatrick  and C. D. Gelatt  and M. P. Vecchi },
title = {{Optimization by Simulated Annealing}},
journal = {Science},
volume = {220},
number = {4598},
pages = {671},
year = {1983},
doi = {10.1126/science.220.4598.671},
URL = {https://www.science.org/doi/abs/10.1126/science.220.4598.671}}

@article{PRXQuantum.2.020321,
  title = {{Algorithms for Quantum Simulation at Finite Energies}},
  author = {Lu, S. and Ba\~nuls, M. and Cirac, J.},
  journal = {PRX Quantum},
  volume = {2},
  issue = {2},
  pages = {020321},
  numpages = {22},
  year = {2021},
  month = {May},
  publisher = {American Physical Society},
  doi = {10.1103/PRXQuantum.2.020321},
  url = {https://link.aps.org/doi/10.1103/PRXQuantum.2.020321}
}

@article{PRXQuantum.5.040339,
  title = {{Initial State Preparation for Quantum Chemistry on Quantum Computers}},
  author = {Fomichev, S. and others},
  journal = {PRX Quantum},
  volume = {5},
  issue = {4},
  pages = {040339},
  numpages = {37},
  year = {2024},
  month = {Dec},
  publisher = {American Physical Society},
  doi = {10.1103/PRXQuantum.5.040339},
  url = {https://link.aps.org/doi/10.1103/PRXQuantum.5.040339}
}

@Article{hartung2025,
author={Hartung, T. and Jansen, K.},
title={Convergence and efficiency proof of quantum imaginary time evolution for bounded order systems},
journal={Commun. Phys.},
year={2026},
month={Sep},
day={08},
issn={2399-3650},
url={https://doi.org/10.1038/s42005-026-02836-9}
}

@article{mcmahon2025equatingquantumimaginarytime,
  title = {Equating quantum imaginary time evolution, Riemannian gradient flows, and stochastic implementations},
  author = {McMahon, N. and Pervez, M. and Arenz, C.},
  journal = {Phys. Rev. Res.},
  volume = {8},
  issue = {2},
  pages = {023024},
  numpages = {6},
  year = {2026},
  month = {Apr},
  publisher = {American Physical Society},
  doi = {10.1103/ht2m-1j91},
  url = {https://link.aps.org/doi/10.1103/ht2m-1j91}
}

@article{PhysRev.106.364,
  title = {{Correlation Energy of an Electron Gas at High Density}},
  author = {Gell-Mann, M. and Brueckner, K. A.},
  journal = {Phys. Rev.},
  volume = {106},
  issue = {2},
  pages = {364},
  numpages = {0},
  year = {1957},
  month = {Apr},
  publisher = {American Physical Society},
  doi = {10.1103/PhysRev.106.364},
  url = {https://link.aps.org/doi/10.1103/PhysRev.106.364}
}

@article{Saad1992,
author  = {Saad, Y.},
title   = {{Analysis of Some Krylov Subspace Approximations to the Matrix Exponential Operator}},
journal = {SIAM J. Numer. Anal.},
volume  = {29},
number  = {1},
pages   = {209},
year    = {1992},
doi     = {10.1137/0729014}
}

@article{HochbruckLubich1997,
author  = {Hochbruck, M. and Lubich, C.},
title   = {{On Krylov Subspace Approximations to the Matrix Exponential Operator}},
journal = {SIAM J. Numer. Anal.},
volume  = {34},
number  = {5},
pages   = {1911},
year    = {1997},
doi     = {10.1137/S0036142995280572}
}

@misc{GomezLurbe2026,
author        = {G{\'o}mez-Lurbe, R. and P{\'e}rez, A.},
title         = {{Pauli Propagation for Imaginary Time Evolution}},
year          = {2026},
eprint        = {2601.14400},
archivePrefix = {arXiv},
doi           = {10.48550/arXiv.2601.14400}
}

@misc{note:tau,
  note = {Per unit displacement of the state on the unit sphere, $\hat{\rho}(\tau)$ remains optimal: since $[\hat H,\hat\rho]\ket\psi=(\hat H-\langle\hat H\rangle)\ket\psi$ is the gradient of the energy on the sphere, ITE is its steepest-descent flow, and $\hat{\rho}^{\mathrm{A}}(\tau)$ achieves its larger rate by moving the state faster, at the price of a reduced alignment with the gradient.}
}

@article{AmeeraAnuar_2026,
doi = {10.1088/1751-8121/ae767d},
url = {https://doi.org/10.1088/1751-8121/ae767d},
year = {2026},
month = {jun},
publisher = {IOP Publishing},
volume = {59},
number = {24},
pages = {245301},
author = {Ameera Anuar, A. and Jamet, F. and Gironella, F. and Šimkovic IV, F. and Rossi, R.},
title = {Operator-projected variational quantum imaginary time evolution},
journal = {J. Phys. A: Math. Theor.}
}

@article{SachdevaVishnoi2014,
author = {Sachdeva, S. and Vishnoi, N.},
title = {{Faster Algorithms via Approximation Theory}},
    year = {2014},
    issue_date = {Mar 2014},
    publisher = {Now Publishers Inc.},
    address = {Hanover, MA, USA},
    volume = {9},
    number = {2},
    issn = {1551-305X},
    url = {https://doi.org/10.1561/0400000065},
    doi = {10.1561/0400000065},
    journal = {Found. Trends Theor. Comput. Sci.},
    month = mar,
    pages = {125},
    numpages = {85}
    }

@article{LinTong2020,
  author  = {Lin, L. and Tong, Y.},
  title   = {Near-optimal ground state preparation},
  journal = {Quantum},
  volume  = {4},
  pages   = {372},
  year    = {2020},
  doi     = {10.22331/q-2020-12-14-372}
}

@article{MandelstamTamm1945,
  author  = {Mandelstam, L. and Tamm, I.},
  title   = {The uncertainty relation between energy and time in non-relativistic quantum mechanics},
  journal = {J. Phys. (USSR)},
  volume  = {9},
  pages   = {249--254},
  year    = {1945}
}

@software{BenavidesRiveros2026AITEcode,
  author    = {Benavides-Riveros, C. L. and Sharma, P. and Frisoni, P. and
               {\v{S}}imkovic IV, F. },
  title     = {Code and data for ``Augmenting Imaginary-Time Evolution
               with Local Geometric Information''},
  year      = {2026},
  publisher = {Zenodo},
  version   = {v1.0},
  doi       = {10.5281/zenodo.23085330},
  url       = {https://doi.org/10.5281/zenodo.23085330}
}

@inproceedings{GharibianLeGall2022,
  author    = {Gharibian, S. and Le Gall, F.},
  title     = {Dequantizing the Quantum Singular Value Transformation: Hardness
               and Applications to Quantum Chemistry and the Quantum {PCP} Conjecture},
  booktitle = {Proceedings of the 54th Annual ACM SIGACT Symposium on Theory of Computing (STOC 2022)},
  pages     = {19--32},
  year      = {2022},
  doi       = {10.1145/3519935.3519991},
  eprint    = {2111.09079},
  archivePrefix = {arXiv}
}

@article{GeTuraCirac2019,
    author = {Ge, Y. and Tura, J. and Cirac, J.},
    title = {Faster ground state preparation and high-precision ground energy estimation with fewer qubits},
    journal = {J. Math. Phys.},
    volume = {60},
    number = {2},
    pages = {022202},
    year = {2019},
    month = {02},
    issn = {0022-2488},
    doi = {10.1063/1.5027484},
    url = {https://doi.org/10.1063/1.5027484},
}

@misc{nappi2026,
      title={Quantum-Classical Auxiliary-Field Quantum Monte Carlo at the Edge of Practicability}, 
      author={F. Nappi and M. Kiser and F. Šimkovic IV},
      year={2026},
      eprint={2606.19239},
      archivePrefix={arXiv},
      primaryClass={quant-ph},
      url={https://arxiv.org/abs/2606.19239}, 
}

@book{HornJohnson2012,
  author    = {Horn, R. and Johnson, C.},
  title     = {Matrix Analysis},
  edition   = {2nd},
  publisher = {Cambridge University Press},
  year      = {2012}
}

@article{Mao2023,
author  = {Mao, Y. and Chaudhary, M. and Kondappan, M. and Shi, J. and Ilo-Okeke, E. and Ivannikov, V. and Byrnes, T.},
title   = {{Measurement-Based Deterministic Imaginary Time Evolution}},
journal = {Phys. Rev. Lett.},
volume  = {131},
pages   = {110602},
year    = {2023},
doi     = {10.1103/PhysRevLett.131.110602}
}

@article{Kosugi2021,
author        = {Kosugi, T. and Nishiya, Y. and Nishi, H. and Matsushita, Y.},
title         = {{Ima\-gi\-nary-Time Evolution Using Forward and Backward Real-Time Evolution with a Single Ancilla: First-Quantized Eigensolver Algorithm for Quantum Chemistry}},
journal       = {Phys. Rev. Res.},
volume        = {4},
pages         = {033121},
year          = {2022},
doi           = {10.1103/PhysRevResearch.4.033121}
}

@article{Xie2024,
author  = {Xie, H. and Wei, S. and Yang, F. and Wang, Z. and Chen, C. and Fan, H. and Long, G.},
title   = {{A Probabilistic Imaginary Time Evolution Algorithm Based on Nonunitary Quantum Circuit}},
journal = {Phys. Rev. A},
volume  = {109},
pages   = {052414},
year    = {2024},
doi     = {10.1103/PhysRevA.109.052414}
}

@article{Vidal2004,
author  = {Vidal, G.},
title   = {{Efficient Simulation of One-Dimensional Quan\-tum Many-Body Systems}},
journal = {Phys. Rev. Lett.},
volume  = {93},
pages   = {040502},
year    = {2004},
doi     = {10.1103/PhysRevLett.93.040502}
}

@inproceedings{Gilyen2019,
author    = {Gily{'e}n, A. and Su, Y. and Low, G. and Wiebe, N.},
title     = {{Quantum Singular Value Transformation and Beyond: Exponential Improvements for Quantum Matrix Arithmetics}},
booktitle = {Proceedings of the 51st Annual ACM SIGACT Symposium on Theory of Computing},
pages     = {193--204},
publisher = {Association for Computing Machinery},
address   = {New York, NY},
year      = {2019},
doi       = {10.1145/3313276.3316366}
}

@article{ChildsWiebe2012,
author        = {Childs, A. and Wiebe, N.},
title         = {{Hamiltonian Simulation Using Linear Combinations of Unitary Operations}},
journal       = {Quantum Inf. Comput.},
volume        = {12},
number        = {11},
pages         = {901},
year          = {2012},
doi           = {10.26421/QIC12.11-12-1},
eprint        = {1202.5822},
}

@article{Berry2015,
author  = {Berry, D. and Childs, A. and Cleve, R. and Kothari, R. and Somma, R.},
title   = {{Simulating Hamiltonian Dynamics with a Truncated Taylor Series}},
journal = {Phys. Rev. Lett.},
volume  = {114},
pages   = {090502},
year    = {2015},
doi     = {10.1103/PhysRevLett.114.090502}
}

@article{Lee2022,
author  = {Lee, J. and Pham, H. and Reichman, D.},
title   = {{Twenty Years of Auxiliary-Field Quantum Monte Carlo in Quantum Chemistry: An Overview and Assessment on Main-Group Chemistry and Bond-Breaking}},
journal = {J. Chem. Theory Comput.},
volume  = {18},
number  = {12},
pages   = {7024},
year    = {2022},
doi     = {10.1021/acs.jctc.2c00802}
}

@incollection{Zhang2018,
author    = {Zhang, S.},
title     = {{Ab Initio Electronic Structure Calculations by Auxiliary-Field Quantum Monte Carlo}},
booktitle = {Handbook of Materials Modeling: Methods: Theory and Modeling},
editor    = {Andreoni, Wanda and Yip, Sidney},
publisher = {Springer International Publishing},
address   = {Cham},
pages     = {123--149},
year      = {2020},
doi       = {10.1007/978-3-319-44677-6_47}
}

@article{Foulkes2001,
author  = {Foulkes, W. and Mitas, L. and Needs, R. J. and Rajagopal, G.},
title   = {{Quantum Monte Carlo Simulations of Solids}},
journal = {Rev. Mod. Phys.},
volume  = {73},
pages   = {33},
year    = {2001},
doi     = {10.1103/RevModPhys.73.33}
}

@article{Yan2018,
author  = {Yan, Y. and Blume, D.},
title   = {{Path Integral Monte Carlo Ground State Approach: Formalism, Implementation, and Applications}},
journal = {J. Phys. B: At. Mol. Opt. Phys.},
volume  = {50},
number  = {22},
pages   = {223001},
year    = {2017},
doi     = {10.1088/1361-6455/aa8d7f}
}

@article{Carleo2017,
author  = {Carleo, G. and Troyer, M.},
title   = {{Solving the Quantum Many-Body Problem with Artificial Neural Networks}},
journal = {Science},
volume  = {355},
number  = {6325},
pages   = {602},
year    = {2017},
doi     = {10.1126/science.aag2302}
}

@article{Haegeman2011,
author  = {Haegeman, J. and Cirac, J. and Osborne, T. and Pi{\v{z}}orn, I. and Verschelde, H. and Verstraete, F.},
title   = {{Time-Dependent Variational Principle for Quantum Lattices}},
journal = {Phys. Rev. Lett.},
volume  = {107},
pages   = {070601},
year    = {2011},
doi     = {10.1103/PhysRevLett.107.070601}
}

@article{Suzuki2025,
author        = {Suzuki, Y. and Tiang, B. and Son, J. and Ng, N. and Holmes, Z. and Gluza, M.},
title         = {{Double-Bracket Algorithm for Quantum Signal Processing Without Post-Selection}},
journal       = {Quantum},
volume        = {9},
pages         = {1954},
year          = {2025},
doi           = {10.22331/q-2025-12-23-1954},
}

\section*{Acknowledgments}

We thank Aeishah Ameera Anuar, Dimitri Pimenov, Jernej Rudi Finzgar, Raphael Seidel, and Manuel Algaba for insightful discussions.

\onecolumngrid

\newpage 

 \section*{ Supplemental Material for\\
``Augmenting Imaginary-Time Evolution with Local Geometric Information''}

\section{Foundation of augmented imaginary-time evolution}

\subsection{Expressing double-bracket in experimentally accessible subspaces}

Let $\hat{M}$ be a Hermitian operator and consider the associated two-dimensional Krylov subspace
\begin{align}
\mathcal{S}_M(\tau) = 
\left\{\ket{\psi(\tau)},\ket{v_M(\tau)} \equiv \frac{\bar{M}\ket{\psi(\tau)}}{\sqrt{\mu_2^M(\tau)}}\right\},
\end{align}
where $\bar{M} \equiv \hat{M} - M(\tau)$ denotes the centered operator with $M(\tau) = \bra{\psi(\tau)}\hat{M}\ket{\psi(\tau)}$, and $\mu_n^M(\tau) \equiv \bra{\psi(\tau)}\bar{M}^n\ket{\psi(\tau)}$ is the $n$th central moment of $\hat{M}$ in the state $\ket{\psi(\tau)}$. The operator of central interest is the double commutator
\begin{align}
\hat{\chi}(\tau) \equiv -\bigl[\hat{H},\bigl[\hat{H},\hat{\rho}(\tau)\bigr]\bigr], \qquad \hat{\rho}(\tau) \equiv \ket{\psi(\tau)}\bra{\psi(\tau)}.
\end{align}
Restricting $\hat{\chi}(\tau)$ to the subspace $\mathcal{S}_M(\tau)$, we obtain the matrix representation
\begin{align}
\chi(\tau) \equiv \hat{\chi}(\tau)\big|_{\mathcal{S}_M(\tau)}
=
\begin{pmatrix}
\bra{\psi(\tau)}\hat{\chi}\ket{\psi(\tau)} & \bra{\psi(\tau)}\hat{\chi}\ket{v_M(\tau)} \\[6pt]
\bra{v_M(\tau)}\hat{\chi}\ket{\psi(\tau)} & \bra{v_M(\tau)}\hat{\chi}\ket{v_M(\tau)}
\end{pmatrix}
\end{align}
whose physical meaning we now study term by term. 

\medskip
\noindent\textbf{1. Diagonal term: $\bra{\psi(\tau)}\hat{\chi}(\tau)\ket{\psi(\tau)}$.}
This term evaluates straightforwardly,
\begin{align}
\bra{\psi(\tau)}\hat{\chi}(\tau)\ket{\psi(\tau)}
= -2\bra{\psi(\tau)}\bar{H}^2\ket{\psi(\tau)}
= -2\mu_2^H(\tau),
\end{align}
and is proportional to the energy variance. It coincides with the instantaneous rate of change of the energy under standard ITE, and quantifies the intrinsic curvature of the energy functional at $\ket{\psi(\tau)}$.

\medskip
\noindent \noindent\textbf{2. Off-diagonal term: $\bra{\psi(\tau)}\hat{\chi}(\tau)\ket{v_M(\tau)}$.}
Evaluating explicitly,
\begin{align}
\bra{\psi(\tau)}\hat{\chi}(\tau)\ket{v_M(\tau)}
&= \frac{-1}{\sqrt{\mu_2^M(\tau)}}\Bigl[
\bra{\psi(\tau)}\hat{H}^2\bar{M}\ket{\psi(\tau)}
- 2E(\tau)\bra{\psi(\tau)}\hat{H}
\bar{M}\ket{\psi(\tau)}
\Bigr] \nonumber\\
&= \frac{-1}{\sqrt{\mu_2^M(\tau)}}\bra{\psi(\tau)}\bar{H}^2\bar{M}\ket{\psi(\tau)},
\end{align}
or, equivalently, $\bra{\psi(\tau)}\hat{\chi}(\tau)\ket{v_M(\tau)} = -\bra{\psi(\tau)}\bar{H}^2\ket{v_M(\tau)}$. We have used $E(\tau) = \bra{\psi(\tau)}\hat{H}\ket{\psi(\tau)}$. This matrix element admits a clear geometric interpretation: it corresponds to the directional derivative of the energy variance along the tangent direction $-\bar{M}\ket{\psi(\tau)}$, which lies 
orthogonal to $\ket{\psi(\tau)}$. Since the energy variance defines a scalar field on projective Hilbert space, this term represents the contraction of its Fubini-Study gradient, $\bar{H}^2\ket{\psi(\tau)}$, with the tangent vector generated by $\hat{M}$.

Decomposing into symmetric and antisymmetric parts,
\begin{align}
\bra{\psi(\tau)}\bar{H}^2\bar{M}\ket{\psi(\tau)}
= \frac{1}{2}\bra{\psi(\tau)}\bigl\{\bar{H}^2,\bar{M}\bigr\}\ket{\psi(\tau)}
+ \frac{1}{2}\bra{\psi(\tau)}\bigl[\bar{H}^2,\bar{M}\bigr]\ket{\psi(\tau)}.
\end{align}
The anticommutator term captures correlations between energy fluctuations and fluctuations of 
$\hat{M}$, while the commutator term governs the dynamical response of the energy variance under 
the unitary flow generated by $\hat{M}$. To make this explicit, consider the one-parameter family 
of states
\begin{align}
\ket{\psi_\lambda(\tau)} = e^{-i\lambda\hat{M}}\ket{\psi(\tau)}.
\end{align}
A direct calculation gives
\begin{align}
\frac{d}{d\lambda}\bra{\psi_\lambda(\tau)}\bar{H}^2\ket{\psi_\lambda(\tau)}
\bigg|_{\lambda=0}
= -i\bra{\psi(\tau)}\bigl[\bar{H}^2,\hat{A}\bigr]\ket{\psi(\tau)},
\end{align}
demonstrating that the commutator term in the off-diagonal matrix element equals the 
infinitesimal rate of change of the energy variance along the unitary orbit of $\hat{M}$.

\medskip
\noindent\textbf{3. Diagonal term: $\bra{v_M(\tau)}\hat{\chi}(\tau)\ket{v_M(\tau)}$.}
A direct calculation yields
\begin{align}
\bra{v_M(\tau)}\hat{\chi}(\tau)\ket{v_M(\tau)}
= \frac{2}{\mu_2^M(\tau)}\bra{\psi(\tau)}\bar{M}\hat{H}\ket{\psi(\tau)}\bra{\psi(\tau)}\hat{H}\bar{M}\ket{\psi(\tau)}
= \frac{2}{\mu_2^M(\tau)}\bigl|\bra{\psi(\tau)}\bar{M}\,\bar{H}\ket{\psi(\tau)}\bigr|^2.
\end{align}
This positive-definite contribution reflects a genuine geometric curvature effect. The overlap
$\bra{\psi(\tau)}\bar{M}\,\bar{H}\ket{\psi(\tau)}$ is the inner product between the variance 
gradient $\bar{H}\ket{\psi(\tau)}$ and the tangent direction $\bar{M}\ket{\psi(\tau)}$, 
quantifying how efficiently the unitary flow generated by $\hat{M}$ probes the energy landscape.

\vspace{0.5em}
\noindent\textbf{Final Result}.
Collecting all terms, the restriction of the double commutator to $\mathcal{S}(\tau)$ reads
\begin{align}
\chi(\tau) \equiv \hat{\chi}(\tau)\big|_{\mathcal{S}(\tau)} 
=
\begin{pmatrix}
-2 \bra{\psi(\tau)} \bar H^2 \ket{\psi(\tau)} & -\bra{\psi(\tau)} \bar H^2 \ket{v_M(\tau)} \\
-\bra{v_M(\tau)} \bar H^2 \ket{\psi(\tau)} & 2|\bra{\psi(\tau)} \bar H \ket{v_M(\tau)}|^2
\end{pmatrix}.
\label{xmatrix}
\end{align}

\subsection{Geometric interpretation}

We now present a geometric interpretation of the preceding results in terms of the K\"ahler geometry of projective Hilbert space. Let $\mathcal{P}(\mathcal{H})$ denote the projective Hilbert space. At a point $[\psi] \in \mathcal{P}(\mathcal{H})$, the tangent space $T_{[\psi]}\mathcal{P}(\mathcal{H})$ is identified with the orthogonal complement of $\ket{\psi}$ in $\mathcal{H}$, i.e., with horizontal vectors $P_\perp\ket{\phi}$, where $P_\perp \equiv 1 - \ket{\psi}\bra{\psi}$. On this tangent space, $\mathcal{P}(\mathcal{H})$ carries a K\"ahler structure $(g_{\mathrm{FS}}, \omega_{\mathrm{FS}}, J)$, where $J$ is multiplication by $i$ and~\cite{AshtekarSchilling,AnandanAharonov}
\begin{align}
g_{\mathrm{FS}}(u,v) = 2\,\mathrm{Re}\braket{u}{v},
\qquad
\omega_{\mathrm{FS}}(u,v) = 2\,\mathrm{Im}\braket{u}{v},
\qquad
g_{\mathrm{FS}}(u,v) = \omega_{\mathrm{FS}}(u,Jv).
\label{eq:kahler}
\end{align}

Every Hermitian operator $\hat{A}$ defines both a real-valued expectation function $f_A([\psi]) \equiv \bra{\psi}\hat{A}\ket{\psi}$ on $\mathcal{P}(\mathcal{H})$ and a tangent vector
\begin{align}
|\dot{\psi}_A\rangle \equiv -i\bar{A}\ket{\psi}
\in T_{[\psi]}\mathcal{P}(\mathcal{H}),
\end{align}
the infinitesimal generator of the unitary flow
$\ket{\psi_\lambda} = e^{-i\lambda\hat{A}}\ket{\psi}$. With the normalization in Eq.\eqref{eq:kahler} and the convention $\iota_{X_f}\omega_{\mathrm{FS}} = df$, the Hamiltonian vector field of $f_A$ is precisely $X_{f_A} = |\dot{\psi}_A\rangle$, so that unitary orbits coincide with the Hamiltonian flows of expectation functions. The Fubini--Study gradient and the Hamiltonian vector field of any function $f$ are related by the complex structure, $\mathrm{grad}\,f = J X_f$; explicitly, $\mathrm{grad}\,f_A = \bar{A}\ket{\psi}$. For a pure state, the quantum Fisher information associated with the generator $\hat{A}$ is~\cite{BraunsteinCaves}:
\begin{align}
\mathcal{F}_A = 4\mu_2^A
= 2\,g_{\mathrm{FS}}\!\big(\dot{\psi}_A,\dot{\psi}_A\big),
\end{align}
so the norm of the tangent vector $|\dot{\psi}_A\rangle$ is fixed by $\mathcal{F}_A$.\footnote{The normalization \eqref{eq:kahler} follows Refs.~\cite{AshtekarSchilling,AnandanAharonov} (in units $\hbar=1$). If one instead chooses $g_{\mathrm{FS}} = \mathrm{Re}\braket{\cdot}{\cdot}$ and $\omega_{\mathrm{FS}} = \mathrm{Im}\braket{\cdot}{\cdot}$, then $\mathcal{F}_A = 4 g_{\mathrm{FS}}(\dot\psi_A,\dot\psi_A)$ and all gradients below acquire a factor of $2$; the convention-independent statements, such as the value of $d\mathcal{V}/d\lambda$, are unaffected.} The two structures encode the elementary quantum brackets: for any two observables,
\begin{align}
g_{\mathrm{FS}}\big(\mathrm{grad}\,f_A,\mathrm{grad}\,f_B\big) = \bra{\psi}\{\bar{A},\bar{B}\}\ket{\psi} = 2\,\mathrm{Cov}_\psi(A,B),
\qquad
\omega_{\mathrm{FS}}\big(X_{f_A},X_{f_B}\big)
= -i\bra{\psi}[\hat{A},\hat{B}]\ket{\psi}
= \{f_A,f_B\},
\label{eq:brackets}
\end{align}
i.e., the Riemannian structure measures symmetrized covariances while the symplectic structure reproduces commutators as Poisson brackets.

The energy variance defines a scalar field
$\mathcal{V}([\psi]) \equiv \bra{\psi}\bar{H}^2\ket{\psi}$ on $\mathcal{P}(\mathcal{H})$. For a horizontal variation $\ket{\delta\psi}$ (so that $\braket{\delta\psi}{\psi}=0$), one finds $\delta\mathcal{V} = 2\,\mathrm{Re}\bra{\delta\psi}(\hat{H}^2 - 2E\hat{H})\ket{\psi} = 2\,\mathrm{Re}\bra{\delta\psi}\bar{H}^2\ket{\psi}$;
note that the contribution from the variation of $E(\psi)$ is exactly accounted for by the replacement $\hat{H}^2 \to \bar{H}^2$. Hence
\begin{align}
\mathrm{grad}\,\mathcal{V}\big|_{[\psi]} = P_\perp \bar{H}^2\ket{\psi},
\qquad
X_{\mathcal{V}} = -i P_\perp \bar{H}^2\ket{\psi}.
\end{align}
The rate of change of $\mathcal{V}$ along the unitary orbit of $\hat{A}$ follows from either structure,
\begin{align}
\frac{d}{d\lambda}\mathcal{V}(\psi_\lambda)\bigg|_{\lambda=0} = g_{\mathrm{FS}}\!\big(\mathrm{grad}\,\mathcal{V},\,\dot{\psi}_A\big) = \omega_{\mathrm{FS}}\!\big(X_{\mathcal{V}},\,X_{f_A}\big)
= \{\mathcal{V}, f_A\}
= -i\bra{\psi}\bigl[\bar{H}^2,\bar{A}\bigr]\ket{\psi}
= 2\,\mathrm{Im}\bra{\psi}\bar{H}^2\bar{A}\ket{\psi},
\label{eq:dVdlambda}
\end{align}
in agreement with the direct calculation above.

The real and imaginary parts of the off-diagonal matrix element therefore admit distinct geometric meanings:
\begin{align}
\mathrm{Re}\bra{\psi}\bar{H}^2\ket{v_A}
&= \frac{1}{2\sqrt{\mu_2^A}}\bra{\psi}\{\bar{H}^2,\bar{A}\}\ket{\psi}
= \frac{1}{2\sqrt{\mu_2^A}}\,
g_{\mathrm{FS}}\!\big(\mathrm{grad}\,\mathcal{V},\,\mathrm{grad}\,f_A\big),
\label{eq:Repart}\\[4pt]
\mathrm{Im}\bra{\psi}\bar{H}^2\ket{v_A}
&= \frac{1}{2i\sqrt{\mu_2^A}}\bra{\psi}[\bar{H}^2,\bar{A}]\ket{\psi}
= \frac{1}{2\sqrt{\mu_2^A}}\,\{\mathcal{V},f_A\}
= \frac{1}{2\sqrt{\mu_2^A}}
\frac{d\mathcal{V}}{d\lambda}\bigg|_{\lambda=0}.
\label{eq:Impart}
\end{align}
The real (anticommutator) part is the Riemannian overlap between the gradients of the variance field $\mathcal{V}$ and of the expectation function $f_A$: it measures the static correlation between energy-variance fluctuations and fluctuations of $\hat{A}$, in the sense of Eq.~\eqref{eq:brackets}, and is invisible to the unitary flow. The imaginary (commutator) part is the Poisson bracket $\{\mathcal{V},f_A\}$, i.e., the dynamical response of the energy variance under the unitary flow generated by $\hat{A}$, consistent with the discussion of the off-diagonal term above. Note that the two pairings are related by the complex structure, $g_{\mathrm{FS}}(\mathrm{grad}\,\mathcal{V},\dot{\psi}_A)
= \omega_{\mathrm{FS}}(X_{\mathcal{V}},\dot{\psi}_A)$, which is why the directional derivative \eqref{eq:dVdlambda} can be read equally as a metric or a symplectic statement.

%% Figure for the Kahler-geometry subsection of the Supplemental Material.
%% Requires: tikz with libraries arrows.meta, decorations.pathreplacing, calc
%% (add \usetikzlibrary{arrows.meta,decorations.pathreplacing,calc} to the preamble;
%%  your preamble already loads arrows.meta and decorations.pathreplacing).
%% Uses \bra and \ket as in the rest of the SM.

\begin{figure}[t]
  \centering
\begin{tikzpicture}[>=Stealth, line cap=round, line join=round]

%================ Panel (a): projective Hilbert space ================
\begin{scope}
  % sphere
  \shade[ball color=gray!8] (0,0) circle (2);
  \draw[gray!60] (0,0) circle (2);

  % level sets of V (latitudes), front solid-dashed, back faint
  \foreach \h/\r in {0/2, 0.9/1.786, 1.45/1.377}{
    \draw[blue!55, dashed] ($(-\r,\h)$) arc (180:360:{\r} and {0.27*\r});
    \draw[blue!25, dashed] ($(\r,\h)$) arc (0:180:{\r} and {0.27*\r});
  }
  \node[blue!60, font=\footnotesize, align=center] at (-2.55,0.45)
        {level sets\\ of $\mathcal{V}$};

  % poles
  \fill (0,2) circle (1.3pt) node[above, font=\footnotesize] {$[E_1]$};
  \fill (0,-2) circle (1.3pt) node[below, font=\footnotesize] {$[E_0]$};

  % point [psi] on the h=0.9 latitude (front)
  \coordinate (psi) at (0.95,1.28);
  \fill[black] (psi) circle (1.6pt);
  \node[font=\footnotesize, anchor=south west] at ($(psi)+(-0.2,0.02)$) {\boldsymbol{$[\psi]$}};

  % unitary orbit of A: tilted red curve through psi
  \draw[red!70!black, thick]
    (-0.1,2.02) .. controls (0.55,1.75) and (0.95,1.45) .. (psi)
    .. controls (1.45,0.45) and (1.5,-0.4) .. (1.05,-1.18);
  \draw[red!70!black, thick, ->]
    (1.05,-1.18) .. controls (1.0,-1.28) .. (0.92,-1.4);
  \node[red!70!black, font=\footnotesize, align=center]
        at (2.55,-1.1) {orbit of\\ $e^{-i\lambda \hat A}$};

  % tangent vector psidot_A along the orbit at psi
  \draw[->, thick] (psi) -- ($(psi)+(0.62,-0.78)$)
        node[right=-1pt, font=\footnotesize] {$|\dot\psi_A\rangle$};

  % grad V at psi, perpendicular to the level set (pointing down, toward larger V)
  \draw[->, thick, red] (psi) -- ($(psi)+(-0.28,-0.78)$)
        node[below left=-2pt, font=\footnotesize] {$\mathrm{grad}\,\mathcal{V}$};

  \node[font=\small] at (0,-2.75) {(a) state manifold $\mathcal{P}(\mathcal{H})$};
\end{scope}

%================ Panel (b): tangent plane at [psi] ================
\begin{scope}[xshift=7.2cm, yshift=-0.7cm]
  % level set through psi (tangent to X_V)
  \draw[blue!55, dashed] (-1.6,0) -- (2.45,0);
  \node[blue!55, font=\footnotesize, anchor=east] at (-1.6,0) {$\mathcal{V}=\mathrm{const}$};

  % axes: X_V (horizontal), grad V (vertical)
  \draw[->, thick, blue!70!black] (0,0) -- (2.1,0)
        node[below right=-2pt, font=\footnotesize] {$X_{\mathcal V}=-J\,\mathrm{grad}\,\mathcal{V}$};
  \draw[->, thick, red] (0,0) -- (0,2.1)
        node[above, font=\footnotesize] {$\mathrm{grad}\,\mathcal{V}$};

  % psidot_A and its projections
  \coordinate (v) at (1.55,1.15);
  \draw[->, very thick] (0,0) -- (v);
  \node[font=\footnotesize, anchor=south] at (1.45,1.22) {$|\dot\psi_A\rangle$};
  \draw[dashed, gray] (v) -- (1.55,0);
  \draw[dashed, gray] (v) -- (0,1.15);

  % grad f_A = J psidot_A (rotated 90 deg)
  \draw[->, thick, gray] (0,0) -- (-1.15,1.55);
  \node[gray, font=\footnotesize, anchor=south] at (-1.5,1.5) {$\mathrm{grad}f_A = J|\dot\psi_A\rangle$};
  \draw[gray, ->] (0.51,0.38) arc (36:122:0.63);
  \node[gray, font=\scriptsize] at (0.12,0.84) {$J$};

  % braces with the Re / Im identification
  \draw[decorate, decoration={brace, mirror, amplitude=4pt}, blue!70!black]
        (0.06,-0.12) -- (1.55,-0.12);
  \node[blue!70!black, font=\scriptsize, align=center] at (0.8,-0.78)
        {$g_{\mathrm{FS}}(\mathrm{grad}\,\mathcal V,\mathrm{grad}\,f_A)\propto\mathrm{Re}\bra{\psi}\bar H^2\ket{v_A}$\\[1pt]
         (covariance)};

  \draw[decorate, decoration={brace, mirror, amplitude=4pt}, red]
        (1.67,0.06) -- (1.67,1.15);
  \node[red, font=\scriptsize, align=center, anchor=west] at (1.85,0.55)
        {$\dfrac{d\mathcal V}{d\lambda}=\{\mathcal V,f_A\}$\\[2pt]
         $\propto\ \mathrm{Im}\bra{\psi}\bar H^2\ket{v_A}$};

  \node[font=\small] at (0.8,-1.95) {(b) tangent plane $T_{[\psi]}\mathcal{P}(\mathcal{H})$};
\end{scope}

\end{tikzpicture}
\caption{Geometric content of the projected double-bracket operator. (a)~The energy variance $\mathcal{V}$ defines a scalar field on the projective Hilbert space $\mathcal{P}(\mathcal{H})$, illustrated here for a two-level system for which $\mathcal{P}(\mathcal{H})\cong S^2$; dashed lines indicate level sets of $\mathcal{V}$, with $\mathcal{V}=0$ at the energy eigenstates $[E_0]$ and $[E_1]$. A Hermitian operator $\hat A$ generates a unitary orbit $e^{-i\lambda\hat A}\ket{\psi}$ (red) through $[\psi]$, with tangent vector $|\dot\psi_A\rangle = -i\bar A\ket{\psi}$. (b)~The tangent plane $T_{[\psi]}\mathcal{P}(\mathcal{H})$ is spanned by two canonical, mutually perpendicular directions associated with $\mathcal{V}$. The first is the Fubini--Study gradient ${\color{red}\mathrm{grad}\,\mathcal{V} = P_\perp\bar{H}^2\ket{\psi}}$ (vertical axis, red), which points in the direction of steepest ascent of the variance field and whose norm measures $\|\mathrm{grad}\,\mathcal{V}\|^2 = \mu_4(\tau) - \mu_2^2(\tau)$. The second is the Hamiltonian vector field ${\color{blue!70!black}X_{\mathcal{V}} = -J\,\mathrm{grad}\,\mathcal{V} = -iP_\perp\bar{H}^2\ket{\psi}}$
(horizontal axis, blue), which is the tangent to the unitary orbit
$e^{-i\lambda\bar{H}^2}\ket{\psi}$ and generates the flow along which
$\mathcal{V}$ is \emph{conserved}; it is perpendicular to
$\mathrm{grad}\,\mathcal{V}$ because $J$ is a Riemannian isometry
($g_{\mathrm{FS}}(u,Ju)=0$ for all $u$), so the level set of $\mathcal{V}$
is tangent to $X_{\mathcal{V}}$ at $[\psi]$.
Projecting $|\dot\psi_A\rangle$ onto this orthogonal pair decomposes the
off-diagonal matrix element $\bra{\psi}\bar H^2\ket{v_A}$ into its two
geometric parts, Eqs.~\eqref{eq:Repart}--\eqref{eq:Impart}: the component
along $\mathrm{grad}\,\mathcal{V}$ gives the Poisson bracket
$\{\mathcal{V},f_A\} = d\mathcal{V}/d\lambda|_{\lambda=0}$ (commutator, imaginary part), i.e.\ the rate of change of the variance
under the unitary flow of $\hat A$; the component along $X_{\mathcal{V}}$
gives the metric overlap
$g_{\mathrm{FS}}(\mathrm{grad}\,\mathcal{V},\mathrm{grad}\,f_A)$
(anticommutator, real part), i.e.\ the static covariance between
energy-variance and $\hat A$ fluctuations, which is invisible to any
unitary flow. The complex structure $J$ rotates $|\dot\psi_A\rangle$ into $\mathrm{grad}\,f_A$ (gray arrow).}
  \label{fig:kahler}
\end{figure}

Finally, consider the second diagonal element. The relevant geometric object here is not the variance gradient but the gradient of the energy
$E([\psi]) = \bra{\psi}\hat{H}\ket{\psi}$ itself,
\begin{align}
\mathrm{grad}\,E = \bar{H}\ket{\psi},
\qquad
X_{E} = -i\bar{H}\ket{\psi},
\end{align}
where $X_E$ generates the (projective) Schr\"odinger flow. Writing $\Pi_{v_A} \equiv \ket{v_A}\bra{v_A}$ for the orthogonal projector onto the complex line spanned by $\ket{v_A}$, we have
\begin{align}
\bra{v_A}\hat{\chi}(\tau)\ket{v_A}
= 2\bigl|\bra{v_A}\bar{H}\ket{\psi}\bigr|^2
= 2\,\bigl\|\Pi_{v_A}\,\mathrm{grad}\,E\bigr\|^2
= \frac{2}{\mu_2^A}\left[
\mathrm{Cov}_\psi(A,H)^2
+ \frac{1}{4}\left(\frac{dE}{d\lambda}\bigg|_{\lambda=0}\right)^{\!2}
\right],
\label{eq:diag2}
\end{align}
with $\mathrm{Cov}_\psi(A,H) = \tfrac{1}{2}\bra{\psi}\{\bar{A},\bar{H}\} \ket{\psi}$ and $dE/d\lambda|_0 = -i\bra{\psi}[\hat{H},\hat{A}]\ket{\psi}$ the rate of energy change along the orbit of $\hat{A}$. This positive contribution therefore measures the squared norm of the projection of the energy gradient onto the direction explored by $\hat{A}$, combining the static $A$--$H$ covariance with the dynamical energy response under the unitary flow.

In summary, the restriction of the double commutator $\hat{\chi}(\tau)$ to the Krylov subspace $\mathcal{S}_A(\tau)$ defines a local quadratic form whose components admit the following geometric interpretation: the diagonal term $\bra{\psi}\hat{\chi}\ket{\psi} = -2\mathcal{V}(\psi)$ measures the magnitude of the energy-variance field; the off-diagonal term encodes both the Riemannian correlation between the gradients of $\mathcal{V}$ and $f_A$ (anticommutator part, Eq.~\eqref{eq:Repart}) and the Poisson bracket $\{\mathcal{V},f_A\}$ giving the rate of change of $\mathcal{V}$ along the unitary orbit of $\hat{A}$ (commutator part, Eq.~\eqref{eq:Impart}); and the diagonal term $\bra{v_A}\hat{\chi}\ket{v_A}$ measures the squared norm of the projection of the energy gradient $\mathrm{grad}\,E$ onto the tangent direction generated by $\hat{A}$, Eq.~\eqref{eq:diag2}. A sketch of this geometrical construction is presented in Fig.~\ref{fig:kahler}.

\subsection{Optimization of the functional and optimal descent}

We now specialize to the case $\hat{M} = \hat{H}$, restricting the discussion to the Krylov subspace
\begin{align}
\mathcal{S}_H(\tau) = \mathrm{span}\left\{\ket{\psi(\tau)}, 
\ket{v_H(\tau)} \equiv \frac{\bar{H}\ket{\psi(\tau)}}{\sqrt{\mu_2(\tau)}}\right\},
\end{align}
where, for brevity, we drop the superscript and write $\mu_n(\tau) \equiv \mu_n^H(\tau)$. Using the matrix representation derived in Eq.~\eqref{xmatrix}, the restriction of the double commutator to $\mathcal{S}_H(\tau)$ takes the form
\begin{align}
\chi(\tau) =
\begin{pmatrix}
-2\mu_2(\tau) & -\gamma(\tau) \\
-\gamma(\tau) & 2\mu_2(\tau)
\end{pmatrix},
\label{eq:matrix}
\end{align}
where
\begin{align}
\gamma(\tau) \equiv \frac{\bra{\psi(\tau)}\bar{H}^3\ket{\psi(\tau)}}{\sqrt{\mu_2(\tau)}}
= \frac{\mu_3(\tau)}{\sqrt{\mu_2(\tau)}}.
\end{align}
The quantity $\gamma(\tau)$ is the third central moment of the energy distribution normalized by the energy standard deviation, encoding the asymmetry of energy fluctuations while retaining the physical dimensions of energy variance.

Diagonalizing Eq.~\eqref{eq:matrix} yields
\begin{align}
\chi(\tau) = \lambda_+(\tau)\ket{\lambda_+(\tau)}\bra{\lambda_+(\tau)}
+ \lambda_-(\tau)\ket{\lambda_-(\tau)}\bra{\lambda_-(\tau)},
\qquad
\lambda_\pm(\tau) = \pm\sqrt{4\mu_2^2(\tau) + \gamma^2(\tau)}.
\end{align}

We aim to minimize the linear functional
\begin{align}
\mathcal{F}_{\chi(\tau)}[\hat{\mathcal{O}}] \equiv \Tr\bigl[\chi(\tau)\,\hat{\mathcal{O}}\bigr]
\end{align}
over the set of Hermitian projectors acting on $\mathcal{S}_H(\tau)$.
Since $\mathcal{S}_H(\tau)$ is two-dimensional, the admissible Hermitian projectors 
are of rank $0$, $1$, or $2$. The rank-$0$ case is trivial and gives 
$\mathcal{F}_{\chi(\tau)}[0] = 0$. The rank-$2$ case gives the full trace,
\begin{align}
\mathcal{F}_{\chi(\tau)}[\hat{I}] = \Tr\bigl[\chi(\tau)\bigr] = 0,
\end{align}
which vanishes identically since $\chi(\tau)$ is traceless, as is manifest from Eq.~\eqref{eq:matrix}. The minimum is therefore attained within the set of rank-$1$ projectors $\hat{\mathcal{O}} = \ket{\phi}\bra{\phi}$, for which the  functional reduces to the Rayleigh quotient
\begin{align}
\mathcal{F}_{\chi(\tau)}[\ket{\phi}\bra{\phi}] = \Tr\bigl[\chi(\tau)\ket{\phi}\bra{\phi}\bigr] = \bra{\phi}\chi(\tau)\ket{\phi}.
\end{align}
By the min-max theorem, the minimum of the Rayleigh quotient over all unit vectors 
$\ket{\phi} \in \mathcal{S}_H(\tau)$ is the smallest eigenvalue of $\chi(\tau)$,  attained at the corresponding eigenvector. The global minimum over all Hermitian projectors on $\mathcal{S}_H(\tau)$ is therefore
\begin{align}
\min_{\hat{\mathcal{O}}}\,\mathcal{F}_{\chi(\tau)}[\hat{\mathcal{O}}]
= \lambda_-(\tau) = -\sqrt{4\mu_2^2(\tau) + \gamma^2(\tau)}
\leq -2\mu_2(\tau) = \mathcal{F}_{\chi(\tau)}[\hat{\rho}(\tau)],
\label{eq:bound}
\end{align}
where equality holds if and only if $\gamma(\tau) = 0$, and the minimum is 
attained uniquely at the rank-one projector
\begin{align}
\hat{\mathcal{O}}_{\mathrm{min}}(\tau) = \ket{\lambda_-(\tau)}\bra{\lambda_-(\tau)},
\end{align}
where 
\begin{align}
\ket{\lambda_-(\tau)} = \cos(\phi(\tau))\ket{\psi(\tau)} 
+ \sin(\phi(\tau))\ket{v_H(\tau)}.
\end{align}
The inequality in Eq.~\eqref{eq:bound} demonstrates that the conventional ITE, which corresponds to $\hat{\mathcal{O}} = \hat{\rho}(\tau)$, does not achieve the steepest possible descent of the energy within $\mathcal{S}_H(\tau)$, unless the energy distribution is symmetric about its mean.

The mixing angle $\phi(\tau)$ is determined by the eigenvalue equation for $\chi(\tau)$. A direct application of the double-angle formula yields
\begin{align}
\tan\bigl(2\phi(\tau)\bigr) = \frac{\gamma(\tau)}{2\mu_2(\tau)} = \frac{\mu_3(\tau)}{2\mu_2^{3/2}(\tau)},
\label{eq:angle}
\end{align}
where $\kappa(\tau)\equiv\mu_3(\tau)/\mu_2^{3/2}(\tau)$ is the Fisher--Pearson skewness, so that $\tan2\phi=\kappa/2$, 
with the explicit trigonometric factors
\begin{align}
\cos\phi(\tau) = \frac{1}{\sqrt{1+r^2(\tau)}}, 
\qquad \sin\phi(\tau) = \frac{r(\tau)}{\sqrt{1+r^2(\tau)}}, \qquad r(\tau) \equiv - \frac{\lambda_-(\tau) + 2\mu_2(\tau)}{\gamma(\tau)}.
\end{align}
The mixing angle is thus determined entirely by the ratio of the third to the (three-halves power of the) second central moment of the energy, i.e., by the skewness of the spectral distribution.

%% Figure for the "Optimization of the functional and optimal descent"
%% subsection of the Supplemental Material.
%% Requires: tikz with libraries arrows.meta, decorations.pathreplacing, calc
%% (your preamble already loads arrows.meta and decorations.pathreplacing;
%%  add calc to \usetikzlibrary{...}).
%% Uses \bra and \ket as in the rest of the SM.

\begin{figure}[t]
  \centering
\begin{tikzpicture}[>=Stealth, line cap=round, line join=round]

%================ Panel (a): Krylov subspace S_H(tau) ================
\begin{scope}
  \def\R{1.6}
  \def\phimix{25} % illustrative mixing angle phi(tau), in degrees

  % unit circle of normalized states |phi> = cos(phi)|psi> + sin(phi)|v_H>
  \draw[gray!45] (0,0) circle (\R);
  \node[gray!55, font=\footnotesize, align=left, anchor=north west] at (1,2.2)
        {$\ket{\lambda_-(\tau)}=\cos\phi\ket{\psi(\tau)}$\\[-1pt]$\;\;\;\quad\quad\quad+\sin\phi\ket{v_H(\tau)}$};
  % basis axes
  \draw[->, thick] (0,0) -- (1.95,0) node[right, font=\footnotesize] {$\ket{\psi(\tau)}$};
  \draw[->, thick] (0,0) -- (0,1.95) node[above, font=\footnotesize] {$\ket{v_H(\tau)}$};

  % ITE direction: phi = 0, i.e. |psi(tau)> itself
  \fill[gray!60!black] (\R,0) circle (1.6pt);
  \node[below, font=\footnotesize, gray!60!black] at (\R,-0.05) {$\hat\rho(\tau)$ (ITE)};

  % AITE direction: |lambda_-(tau)>, at angle phi(tau)
  \coordinate (lmin) at ({\R*cos(\phimix)},{\R*sin(\phimix)});
  \fill[red!70!black] (lmin) circle (1.6pt);
  \draw[->, very thick, red!70!black] (0,0) -- (lmin);
  \node[above right=-2pt, font=\footnotesize, red!70!black] at (lmin)
        {$\ket{\lambda_-(\tau)}$ (AITE)};

  % orthogonal direction |lambda_+(tau)>
  \coordinate (lplus) at ({\R*cos(\phimix+90)},{\R*sin(\phimix+90)});
  \draw[->, thick, blue!55!black] (0,0) -- (lplus);
  \node[above left=-2pt, font=\footnotesize, blue!55!black] at (lplus)
        {$\ket{\lambda_+(\tau)}$};

  % small right-angle mark between lambda_- and lambda_+
  \draw[blue!55!black]
    ($0.15*(lmin)$) -- ($0.15*(lmin)+0.16*(lplus)$) -- ($0.15*(lplus)$);

  % angle arc phi(tau) between |psi> and |lambda_->
  \draw[gray!60!black] (0.55,0) arc (0:\phimix:0.55);
  \node[font=\footnotesize, gray!60!black]
        at ({0.98*cos(\phimix/2)},{0.8*sin(\phimix/2)}) {$\phi(\tau)$};

  \node[font=\small] at (0,-2.15) {(a) optimal direction within $\mathcal{S}_H(\tau)$};
\end{scope}

%================ Panel (b): descent functional vs mixing angle ================
\begin{scope}[xshift=6.7cm, yshift=0.25cm]
  \def\sx{0.0225}
  \def\phimix{25}

  % axes
  \draw[->, thick] (-0.25,-1.35) -- (4.35,-1.35) node[right, font=\footnotesize] {$\phi$};
  \draw[->, thick] (-0.25,-1.35) -- (-0.25,1.35)
        node[above, font=\footnotesize] {$\mathcal{F}_{\chi(\tau)}[\ket{\phi}\!\bra{\phi}]$};

  % curve F(phi) = <phi|chi(tau)|phi> = -2mu2 cos2phi - gamma sin2phi
  \draw[domain=0:180, samples=120, smooth, variable=\x, thick, purple]
    plot ({\x*\sx}, {-cos(2*\x-2*\phimix)});

  % x-axis ticks
  \draw (0,-1.35) -- (0,-1.42) node[below, font=\scriptsize] {$0$};
  \draw ({\phimix*\sx},-1.35) -- ({\phimix*\sx},-1.42)
        node[below, font=\scriptsize] {$\phi(\tau)$};
  \draw ({90*\sx},-1.35) -- ({90*\sx},-1.42) node[below, font=\scriptsize] {$\pi/2$};
  \draw ({180*\sx},-1.35) -- ({180*\sx},-1.42) node[below, font=\scriptsize] {$\pi$};

  % y-axis tick labels
  \draw (-0.25,1) -- (-0.32,1) node[left, font=\scriptsize] {$\lambda_+(\tau)$};
  \draw (-0.25,{-cos(-2*\phimix)}) -- (-0.32,{-cos(-2*\phimix)})
        node[left, font=\scriptsize] {$-2\mu_2(\tau)$};
  \draw (-0.25,-1) -- (-0.32,-1) node[left, font=\scriptsize] {\color{red}$\lambda_-(\tau)$};

  % ITE point at phi=0
  \fill[gray!60!black] (0,{-cos(-2*\phimix)}) circle (1.4pt);
  \draw[dashed, gray] (0,{-cos(-2*\phimix)}) -- (-0.25,{-cos(-2*\phimix)});
  \node[above right=1pt, font=\footnotesize, gray!60!black, anchor=south west]
        at (0,{-cos(-2*\phimix)-0.05}) {ITE};

  % AITE point: global minimum at phi(tau)
  \fill[red!70!black] ({\phimix*\sx},-1) circle (1.4pt);
  \draw[dashed, red!70!black] ({\phimix*\sx},-1) -- (-0.25,-1);
  \node[font=\footnotesize, red!70!black, anchor=south east]
        at ({\phimix*\sx+0.08},{-1.4}) {AITE};

  % vertical guide marking the extra instantaneous descent gained by AITE
  \draw[<->, thin, blue!55!black]
    ({\phimix*\sx+0.55},{-cos(-2*\phimix)}) -- ({\phimix*\sx+0.55},-1);
  \node[blue!55!black, font=\scriptsize, align=left, anchor=west]
        at ({\phimix*\sx+0.62},{-0.5*(1+cos(-2*\phimix))})
        {extra\\descent};

  \node[font=\small] at (2,-2.15) {(b) energy-descent functional vs.\ mixing angle};
\end{scope}

\end{tikzpicture}
  \caption{Optimal descent direction within the Krylov subspace
  $\mathcal{S}_H(\tau)$. (a)~Normalized states
  $\ket{\phi}=\cos\phi\,\ket{\psi(\tau)}+\sin\phi\,\ket{v_H(\tau)}$ trace out
  the unit circle in $\mathcal{S}_H(\tau)$. Standard ITE corresponds to
  $\hat\rho(\tau)=\ket{\psi(\tau)}\!\bra{\psi(\tau)}$, i.e.\ $\phi=0$, while
  the optimal projector $\hat{\mathcal{O}}_{\mathrm{min}}(\tau) =
  \ket{\lambda_-(\tau)}\!\bra{\lambda_-(\tau)}$ is rotated away from
  $\ket{\psi(\tau)}$ by the skewness-dependent mixing angle $\phi(\tau)$ of
  Eq.~\eqref{eq:angle}. The orthogonal state $\ket{\lambda_+(\tau)}$
  corresponds to the steepest-\emph{ascent} direction. (b)~The Rayleigh
  quotient $\mathcal{F}_{\chi(\tau)}[\ket{\phi}\!\bra{\phi}] =
  \bra{\phi}\chi(\tau)\ket{\phi}$, Eq.~\eqref{eq:matrix}, as a function of
  $\phi$. Because $\chi(\tau)$ is traceless, the curve oscillates symmetrically
  between $\lambda_-(\tau)$ and $\lambda_+(\tau)=-\lambda_-(\tau)$. Standard
  ITE evaluates this functional at the fixed point $\phi=0$, yielding
  $-2\mu_2(\tau)$, whereas AITE locates the global minimum
  $\lambda_-(\tau)=-\sqrt{4\mu_2^2(\tau)+\gamma^2(\tau)}$ at $\phi=\phi(\tau)$,
  Eq.~\eqref{eq:bound}. The two values coincide only when $\gamma(\tau)=0$;
  otherwise the gap between them is the additional instantaneous energy
  descent gained by AITE.}
  \label{fig:krylov}
\end{figure}

Several remarks are in order. First, the standard ITE flow is recovered in the limit $\gamma(\tau) \to 0$, corresponding to a symmetric energy distribution with vanishing skewness, for which $\phi(\tau) \to 0$ and $\hat{\mathcal{O}}_{\mathrm{min}}(\tau) \to \hat{\rho}(\tau)$. Second, the optimal descent direction within $\mathcal{S}_H(\tau)$ is tilted away from $\ket{\psi(\tau)}$ by an angle set precisely by the skewness, revealing that standard ITE implicitly neglects higher-order non-Gaussian features of the energy distribution. Third, the bound in Eq.~\eqref{eq:bound} is purely algebraic and holds at every imaginary time $\tau$, independently of the specific Hamiltonian or initial state. An illustrative sketch of this construction can be found in Fig.~\eqref{fig:krylov}.

\subsection{Non unitary version of AITE}

In the previous sections we generalized the double bracket formulation of imaginary time. Let's now seek a generalized Hamiltonian of the form
\begin{equation}
  \hat{H}_{\mathrm{A}}(\tau) = a_1\Hb + a_2\bar{K},
  \qquad a_1,a_2\in\mathbb{R},
  \label{eq:K-ansatz}
\end{equation}
such that normalized imaginary-time evolution under $\hat{H}_{\mathrm{A}}(\tau)$ matches (in leading order of the imaginary time) the target double-bracket flow:
\begin{equation}
  \frac{e^{-\beta \hat{H}_{\mathrm{A}}(\tau)}|\psi\rangle}{\|e^{-\beta \hat{H}_{\mathrm{A}}(\tau)}|\psi\rangle\|}
  \approx
  e^{-s[\hat H,\hat{\mathcal{O}}_{\mathrm{min}}(\tau)]}|\psi\rangle.
  \label{eq:goal}
\end{equation}
Since $\hat{\mathcal{O}}_{\mathrm{min}}(\tau)=\hat{\mathcal{O}}_{\mathrm{min}}^\dagger(\tau)$ implies
$[H,\hat{\mathcal{O}}_{\mathrm{min}}(\tau)]$ is anti-Hermitian, the right-hand side is
exactly normalized.

A direct computation gives the exact decomposition:
\begin{equation}
  [\hat H,\hat{\mathcal{O}}_{\mathrm{min}}(\tau)]|\psi\rangle
  = \Omega_{\rm eff}(\phi)\,|v_H\rangle
    + \frac{\sin(2\phi)}{2\sqrt{\mu_2}}
    \ket{\perp},
  \label{eq:HO-decomp}
\end{equation}
where
\begin{equation}
  \Omega_{\rm eff}(\phi)
  = \sqrt{\mu_2}\cos(2\phi) + \frac{\mu_3}{2\mu_2}\sin(2\phi),
  \label{eq:Omeff}
\end{equation}
and 
\begin{align}
\ket{\perp} = \Hb^2|\psi\rangle - \mu_2|\psi\rangle
- \frac{\mu_3}{\sqrt{\mu_2}}|v_H\rangle
\end{align} 
is the component of $\Hb^2|\psi\rangle$ orthogonal to $\mathrm{span}\{|\psi\rangle,|v_H\rangle\}$. Notice that this latter state is not normalized:
\begin{align}
   ||\perp||^2 =  \langle\perp\ket{\perp} = \mu_4 -\mu_2^2 - \frac{\mu_3^2}{\mu_2}.
\end{align}

When $\phi=0$ (which is the case of standard ITE), $\sin(2\phi)=0$ and $[\hat H,\hat\rho]|\psi\rangle = \sqrt{\mu_2}|v_H\rangle$. When $\phi\neq 0$ (AITE), the $\ket{\perp}$ component is nonzero and proportional to $\Hb^2|\psi\rangle$. This component lies outside $\mathrm{span}\{|\psi\rangle,|v_H\rangle\}$ and can only be matched by the $a_2 \hat K$ term in $\hat{H}_{\mathrm{A}}$.

The target state to first order in $s$ is therefore:
\begin{equation}
  e^{-s[\hat H,\hat{\mathcal{O}}_{\mathrm{min}}(\tau)]}|\psi\rangle
  = |\psi\rangle
    - s\,\Omega_{\rm eff}(\phi)\,|v_H\rangle
    - \frac{s\sin(2\phi)}{2\sqrt{\mu_2}}\ket{\perp} + \mathcal{O}(s^2).
  \label{eq:target}
\end{equation}

Now, expanding $e^{-\beta \hat{H}_{\mathrm{A}}}$ and normalizing to first order in $\beta$:
\begin{equation}
  \frac{e^{-\beta \hat{H}_{\mathrm{A}}}|\psi\rangle}{\|e^{-\beta \hat{H}_{\mathrm{A}}}|\psi\rangle\|}
  \approx |\psi\rangle
    - \beta\,(\hat{H}_{\mathrm{A}} - \langle \hat{H}_{\mathrm{A}}\rangle_\psi)|\psi\rangle.
  \label{eq:ITE-expand}
\end{equation}
Matching \eqref{eq:ITE-expand} to \eqref{eq:target} requires
equating components in the $|v_H\rangle$ and $\ket{\perp}$ directions separately.

Projecting \eqref{eq:target} and \eqref{eq:ITE-expand} onto $|v_H\rangle$:
\begin{equation}
  \beta\bra{v_H}(\hat{H}_{\mathrm{A}}(\tau)-\langle \hat{H}_{\mathrm{A}}(\tau)\rangle_\psi)|\psi\rangle
  = s\,\Omega_{\rm eff}(\phi).
  \label{eq:match-e1}
\end{equation}
Using Eq.~\eqref{eq:K-ansatz} this reads:
\begin{equation}
  \beta\,\bra{v_H}(a_1 \Hb + a_2 \bar{K})|\psi\rangle
=  \frac{\beta}{\sqrt{\mu_2}}\bigl(a_1\mu_2 + a_2\bra{\psi}\bar{H}\bar{K}\ket{\psi}\bigr)
  = s\,\Omega_{\rm eff}(\phi),
  \label{eq:match-e1-explicit}
\end{equation}
where we used the fact $\bra{\psi}\Hb\Hb|\psi\rangle = \mu_2$.

To find the component along $\ket{\perp}$, let's project onto the normalized $\ket{\perp}$ direction:
\begin{equation}
  \beta\bra{\perp}(\hat{H}_{\mathrm{A}}(\tau)-\langle \hat{H}_{\mathrm{A}}(\tau)\rangle_\psi)|\psi\rangle
  = \frac{s\sin(2\phi)}{2\sqrt{\mu_2}}\,||\perp||^2.
  \label{eq:match-perp}
\end{equation}
Using $\bra{\perp}\psi\rangle = \bra{\perp}\bar H\psi\rangle = 0 $ and computing $\bra{\perp}\bar{K}|\psi\rangle$:
\begin{equation}
  \beta a_2 \underbrace{\left(\bra{\psi}\bar H^2 \bar K  \ket{\psi}- \mu_2\bra{\psi}\bar K  \ket{\psi} - \frac{\mu_3}{\mu_2}\bra{\psi}\bar H \bar K  \ket{\psi}\right)}_{=\,L}
  = \frac{s\sin(2\phi)}{2\sqrt{\mu_2}} ||\perp||^2.
  \label{eq:match-perp-explicit}
\end{equation}
As a result,
\begin{align}
    a_2 = \frac{s}{\beta} \frac{\sin(2\phi)}{2\sqrt{\mu_2}}\frac{||\perp||^2}{L}.
\end{align}
Inserting this in Eq.~\eqref{eq:match-e1-explicit} we finally get:
\begin{align}
    a_1 = \frac{s}{\beta} \cos(2\phi) + \frac{s}{\beta} \frac{\sin(2\phi)}{2\mu_2}  \left[\mu_3- \bra{\psi}\bar{H}\bar{K}\ket{\psi} \frac{||\perp||^2}{L}\right].
\end{align}

Remarkably, in the special case $\bar{K} = \bar{H}^2$, the augmented prefactors Hamiltonian reduces to the compact forms:
\begin{align}
    a_2 = \frac{s}{\beta} \frac{\sin(2\phi)}{2\sqrt{\mu_2}}.
\end{align}
Inserting this in Eq.~\eqref{eq:match-e1-explicit} we finally get:
\begin{align}
    a_1 = \frac{s}{\beta} \cos(2\phi).
\end{align}

We can check that the resulting descent rate is the same as in the unitary case. In fact, under the imaginary-time flow
$\partial_\tau|\psi\rangle = -(\hat{H}_{\mathrm{A}} - \langle\hat{H}_{\mathrm{A}}\rangle_\psi)|\psi\rangle$,
the energy descent rate is:
\begin{equation}
  \partial_\tau E
  = \langle\partial_\tau\psi|\hat H|\psi\rangle + \langle\psi|\hat H|\partial_\tau\psi\rangle
  = -2\langle\psi|\bar{H}\hat{H}_{\mathrm{A}}|\psi\rangle,
  \label{eq:descent-general}
\end{equation}
where the centering $\langle\hat{H}_{\mathrm{A}}\rangle_\psi$ drops out since
$\langle\psi|\bar{H}|\psi\rangle=0$.
With $\hat{H}_{\mathrm{A}} = a_1\Hb + a_2\Hb^2$:
\begin{equation}
  \partial_\tau E
  = -2\Bigl(
     a_1\underbrace{\langle\psi|\Hb^2|\psi\rangle}_{=\,\mu_2}
    + a_2\underbrace{\langle\psi|\Hb^3|\psi\rangle}_{=\,\mu_3}
    \Bigr)
  = -2(a_1\mu_2 + a_2\mu_3).
  \label{eq:descent-K}
\end{equation}

Notice that for both standard ITE and the double-bracket flow with $\hat{\mathcal{O}}=\hat\rho$ (i.e., $\phi=0$) give $a_2=0$, so:
\begin{equation}
  \partial_\tau E\big|_{1,2} = -2a_1\mu_2.
  \label{eq:descent-12}
\end{equation}
Choosing $a_1=1$ (coefficient of $\Hb$ in $\hat{H}_{\mathrm{A}} = \hat H$), gives $\partial_\tau E = -2\mu_2$.

Now, substituting $a_1 = \frac{s}{\beta}\cos(2\phi)$ and $a_2 = \frac{s}{\beta}\sin(2\phi)/2\sqrt{\mu_2}$ into
Eq.\eqref{eq:descent-K}:
\begin{align}
  \partial_\tau E\big|_3
  &= -\frac{2s}{\beta}
     \left(\mu_2\cos(2\phi)
           + \mu_3 \frac{\sin(2\phi)}{2\sqrt{\mu_2}}\right)
  \nonumber\\
  &= -\frac{2s}{\beta}
     \left(\mu_2\cos(2\phi) + \frac{\gamma}{2}\sin(2\phi)\right).
  \label{eq:descent-3-mid}
\end{align}
Using $\mu_3 = \gamma\sqrt{\mu_2}$ and recognizing that $|\lambda_-\rangle = \cos\phi\,|e_0\rangle+\sin\phi\,|v_H\rangle$ is the lowest eigenvector of $\hat\chi|_{\mathcal{S}_H}$, the expression
\eqref{eq:descent-3-mid} at $s=\beta=1$ equals:
\begin{equation}
  \partial_\tau E\big|_3
  = -2\mu_2\cos(2\phi) - \gamma\sin(2\phi)
  = \langle\lambda_-|\hat\chi|\lambda_-\rangle
  = \lambda_-
  = -\sqrt{4\mu_2^2+\gamma^2}.
  \label{eq:descent-3}
\end{equation}
The last equality follows directly from the eigenvector equation
$\hat\chi|\lambda_-\rangle = \lambda_-|\lambda_-\rangle$,
which gives $\langle\lambda_-|\hat\chi|\lambda_-\rangle=\lambda_-$
for any sign of $\gamma$ and any quadrant of $\phi$, without
requiring explicit evaluation of $\cos(2\phi)$ and $\sin(2\phi)$.

%% Figure for the "Non-unitary version of AITE" subsection of the
%% Supplemental Material.
%% Requires: tikz with libraries arrows.meta, decorations.pathreplacing, calc
%% (your preamble already loads arrows.meta and decorations.pathreplacing;
%%  add calc to \usetikzlibrary{...}).
%% Uses \bra, \ket, \Hb as in the rest of the SM.

\begin{figure}[t]
  \centering
\begin{tikzpicture}[>=Stealth, line cap=round, line join=round]

%================ Panel (a): ITE vs AITE velocity in (v_H, perp) plane ================
\begin{scope}
  % axes
  \draw[->, thick] (0,0) -- (2.5,0) node[right, font=\footnotesize] {$\ket{v_H(\tau)}$};
  \draw[->, thick] (0,0) -- (0,2.1) node[above, font=\footnotesize] {$\ket{\perp(\tau)}$};

  % ITE direction: along |v_H> only
  \draw[->, thick, gray!60!black] (0,0) -- (1.5,0);
  \node[above, font=\footnotesize, gray!60!black] at (0.95,0) {ITE};

  % AITE direction: tilted into perp
  \coordinate (aite) at (1.3,0.9);
  \draw[->, very thick, red!70!black] (0,0) -- (aite);
  \node[above right=-2pt, font=\footnotesize, red!70!black] at (aite) {AITE};

  % projections
  \draw[dashed, gray] (aite) -- (1.3,0);
  \draw[dashed, gray] (aite) -- (0,0.9);
  \draw[decorate, decoration={brace, mirror, amplitude=3pt}, blue!60!black]
        (0,-0.18) -- (1.3,-0.18);
  \node[blue!60!black, font=\scriptsize] at (0.65,-0.5) {$\Omega_{\rm eff}(\phi)$};
  \draw[decorate, decoration={brace, amplitude=3pt}, blue!60!black]
        (-0.15,0) -- (-0.15,0.9);
  \node[blue!60!black, font=\scriptsize, align=center, anchor=east] at (-0.2,0.45)
        {$\frac{\sin(2\phi)}{2\sqrt{\mu_2}}$};

  % out-of-plane marker for |psi(tau)>, placed clear of other labels
  \draw (-1.45,-0.55) circle (3.2pt);
  \fill (-1.45,-0.55) circle (0.8pt);
  \node[font=\footnotesize, anchor=west] at (-1.35,-0.55) {$\ket{\psi(\tau)}$};
  \node[gray!55, font=\footnotesize, align=left, anchor=north west] at (1,2.4)
        {\large$[\hat H,\hat{\mathcal{O}}_{\rm min}(\tau)]\ket{\psi}$};
  \node[font=\small] at (1.0,-1.25) {(a) AITE leaves $\mathcal{S}_H(\tau)$ when $\phi\neq0$};
\end{scope}

%================ Panel (b): matching via H_A = a_1 Hbar + a_2 K ================
\begin{scope}[xshift=6.6cm]
  % axes
  \draw[->, thick] (0,0) -- (2.5,0) node[right, font=\footnotesize] {$\ket{v_H(\tau)}$};
  \draw[->, thick] (0,0) -- (0,2.1) node[above, font=\footnotesize] {$\ket{\perp(\tau)}$};

  % target (same AITE vector as panel a), dashed
  \coordinate (aite) at (1.3,0.9);
  \draw[->, very thick, dashed, red!70!black] (0,0) -- (aite);

  % a1 * Hbar|psi> : along v_H only
  \coordinate (g1) at (0.94,0);
  \draw[->, thick, blue!60!black] (0,0) -- (g1);
  \node[below, font=\footnotesize, blue!60!black] at (0.47,-0.04)
        {$a_1\bar H\ket{\psi}$};

  % a2 * (Kbar|psi>)_perp : tip-to-tail
  \draw[->, thick, green!45!black] (g1) -- (aite);
  \node[font=\footnotesize, green!45!black, align=left, anchor=west]
        at ($(g1)+(0.19,0.45)$) {$a_2(\bar K\ket{\psi})_{\perp\psi}$};

  \node[font=\small, align=center] at (1.0,-1.25)
        {(b) $\hat H_{\mathrm{A}}(\tau)=a_1\bar H+a_2\bar K$\\ matches the AITE target};
\end{scope}

\end{tikzpicture}
\caption{Why the augmented Hamiltonian $\hat H_{\mathrm A}(\tau)=a_1\bar H+a_2\bar K$ needs two generators. (a)~The instantaneous velocity $[\hat H,\hat{\mathcal{O}}_{\rm min}(\tau)]\ket{\psi}$, decomposed in the plane orthogonal to $\ket{\psi(\tau)}$ spanned by $\ket{v_H(\tau)}$ and the third-order Krylov direction $\ket{\perp(\tau)}$. For standard ITE ($\phi=0$) this velocity is $\sqrt{\mu_2(\tau)}\,\ket{v_H(\tau)}$ and lies entirely along $\ket{v_H(\tau)}$, i.e.\ within $\mathcal{S}_H(\tau)$. For AITE ($\phi\neq0$) it acquires the additional component along $\ket{\perp(\tau)}$ in Eq.~\eqref{eq:HO-decomp}, with the two projections $\Omega_{\rm eff}(\phi)$, Eq.~\eqref{eq:Omeff}, and $\sin(2\phi)/(2\sqrt{\mu_2(\tau)})\,\|\!\perp\!\|$ shown explicitly. (b)~A single generator $\bar H\ket{\psi}\propto\ket{v_H(\tau)}$ can only reproduce velocities along $\ket{v_H(\tau)}$ (standard ITE, $a_2=0$). To reach the tilted AITE target whenever $\phi\neq0$, a second, linearly independent generator $\bar K\ket{\psi}$ is required: its component $(\bar K\ket{\psi})_{\perp\psi}$ orthogonal to $\ket{\psi(\tau)}$ has a nonzero overlap with $\ket{\perp(\tau)}$, and the coefficients $a_1,a_2$ determined by Eqs.~\eqref{eq:match-e1-explicit} and \eqref{eq:match-perp-explicit} are exactly the lengths needed for the sum $a_1\bar H\ket{\psi}+a_2(\bar K\ket{\psi})_{\perp\psi}$ to land on the AITE target (dashed red, reproduced from panel (a)).}
\label{fig:nonunitary}
\end{figure}

\section{Convergence of AITE vs ITE}

General conclusions about scaling are difficult to draw due to the complex interplay of statistical variables in both ITE and AITE. To illustrate the improvement of AITE over ITE, we first consider an analytically solvable two-level model and then provide a rigorous estimate of late-time convergence.

\subsection{Two-level system}

Let's consider a two-level model with eigenpairs $\{(E_0,\ket{E_0}),(E_1,\ket{E_1})\}$ and energy gap $\Delta = E_1 - E_0 > 0$. An arbitrary real state can be written as
\begin{align}
\ket{\psi(\theta)} = \cos\theta\ket{E_0} + \sin\theta\ket{E_1},  \qquad \theta \in [0, \pi),
\end{align}
for which the energy, variance, and third central moment take the analytical forms
\begin{align}
E(\theta) &= E_0 + \Delta\sin^2\theta, \\
\mu_2(\theta) &= \Delta^2\sin^2\theta\cos^2\theta, \\
\mu_3(\theta) &= \Delta^3\sin^2\theta\cos^2\theta\left(\cos^2\theta - \sin^2\theta\right).
\end{align}
We denote the energy error by $\varepsilon(\tau) \equiv E(\theta(\tau)) - E_0$, 
where $\theta = \theta(\tau)$ depends on imaginary time.

\medskip
\noindent\textbf{ITE.} Under ITE, the energy error evolves according to
\begin{align}
\partial_\tau\varepsilon(\tau) = -2\mu_2(\theta(\tau)),
\end{align}
which, using the analytical expressions above, yields the differential equation
\begin{align}
\dot{\theta} = -\frac{\Delta}{2}\sin(2\theta),
\end{align}
with solution
$\theta(\tau) = \arctan\!\left(\tan\theta_0\,e^{-\Delta\tau}\right)$. Substituting into the energy error gives
\begin{align}
\varepsilon_{\mathrm{ITE}}(\tau)
= \frac{\Delta\tan^2\theta_0\,e^{-2\Delta\tau}}{1 + \tan^2\theta_0\,e^{-2\Delta\tau}},
\end{align}
which reproduces the exponential decay characteristic of standard imaginary-time 
evolution.

\medskip
\noindent \textbf{AITE.} For the two-level system, the lowest eigenvalue of
$\chi(\tau)$ simplifies exactly:
\begin{align}
\lambda_-(\theta)
= -\sqrt{4\mu_2^2(\theta) + \frac{\mu_3^2(\theta)}{\mu_2(\theta)}}
= -\Delta^2\sin\theta\cos\theta
= -\frac{\Delta^2}{2}\sin(2\theta),
\end{align}
where we used $4\mu_2^2 + \mu_3^2/\mu_2
= \Delta^4\sin^2\theta\cos^2\theta(\sin^2\theta+\cos^2\theta)^2
= \Delta^4\sin^2\theta\cos^2\theta$.
Since $\partial_\tau\varepsilon = \lambda_-$ and
$\partial_\tau\varepsilon = \Delta\sin(2\theta)\,\dot\theta$, this gives
\begin{align}
\dot\theta = -\frac{\Delta}{2},
\end{align}
a constant independent of $\theta$. The exact solution is therefore
\begin{align}
\theta(\tau) = \theta_0 - \frac{\Delta}{2}\tau,
\qquad
\varepsilon_{\mathrm{AITE}}(\tau) = \Delta\sin^2\!\left(\theta_0 - \frac{\Delta}{2}\tau\right),
\label{eq:energyerror2level}
\end{align}
which vanishes exactly at the finite extinction time
\begin{align}
\tau^* = \frac{2\theta_0}{\Delta}.
\end{align}
This is the remarkable feature of the AITE  flow: convergence to the ground state occurs in \emph{finite} imaginary  time. Near $\tau^*$, using $\sin(x)\approx x$,
\begin{align}
\varepsilon_{\mathrm{AITE}}(\tau) \approx \frac{\Delta^3}{4}(\tau^*-\tau)^2,
\end{align}
which is a quadratic power-law extinction with coefficient $\Delta^3/4$.

\medskip
\noindent\textbf{Self-accelerating flow.}
To understand the origin of the augmented convergence, we compare the
instantaneous energy decay rates. From the exact solutions derived above,
\begin{align}
\dot{\varepsilon}_{\mathrm{ITE}}
&= -\frac{\Delta^2}{2}\sin^2(2\theta)
 = -2\Delta^2\sin^2\theta\cos^2\theta, \\
\dot{\varepsilon}_{\mathrm{AITE}}
&= \lambda_-(\theta)
 = -\frac{\Delta^2}{2}\sin(2\theta)
 = -\Delta^2\sin\theta\cos\theta.
\end{align}
Their ratio is
\begin{align}
\frac{|\dot{\varepsilon}_{\mathrm{AITE}}|}{|\dot{\varepsilon}_{\mathrm{ITE}}|}
= \frac{1}{\sin(2\theta)}
= \frac{1}{2\sin\theta\cos\theta}.
\end{align}
As $\theta \to 0$, using $\sin\theta \approx \sqrt{\varepsilon/\Delta}$
and $\cos\theta \to 1$, this ratio diverges as
\begin{align}
\frac{|\dot{\varepsilon}_{\mathrm{AITE}}|}{|\dot{\varepsilon}_{\mathrm{ITE}}|}
= \frac{1}{\sin(2\theta)}
\approx \frac{1}{2\sqrt{\varepsilon/\Delta}}
\sim \sqrt{\frac{\Delta}{\varepsilon}} \to \infty,
\end{align}
confirming that the AITE flow becomes increasingly faster than ITE as the
ground state is approached. Near $\theta = 0$, using
$\sin(2\theta) \approx 2\theta$ and $\varepsilon = \Delta\sin^2\theta
\approx \Delta\theta^2$, the energy equation of motion reduces to the
power-law ODE
\begin{align}
\dot{\varepsilon}_{\mathrm{AITE}}
= -\frac{\Delta^2}{2}\sin(2\theta)
\approx -\Delta^2\theta
\approx -\Delta^{3/2}\,\varepsilon^{1/2},
\label{eq:powerlaw_ode}
\end{align}
of the form $\dot{\varepsilon} = -c\,\varepsilon^{\alpha}$ with $\alpha =
1/2$ and $c = \Delta^{3/2}$. Since $\alpha < 1$, this ODE exhibits
\emph{finite-time extinction}: integrating directly,
\begin{align}
\int_{\varepsilon(\tau)}^{0}\frac{d\varepsilon'}{\varepsilon'^{1/2}}
= -c(\tau^* - \tau)
\implies
\varepsilon_{\mathrm{AITE}}(\tau) \sim \frac{\Delta^3}{4}(\tau^* - \tau)^2,
\end{align}
consistent with the exact result of Eq.~\eqref{eq:energyerror2level}. The
instantaneous relative decay rate
\begin{align}
\frac{|\dot{\varepsilon}_{\mathrm{AITE}}|}{\varepsilon_{\mathrm{AITE}}}
= c\,\varepsilon^{\alpha - 1}
= \Delta^{3/2}\,\varepsilon^{-1/2}
\to \infty
\quad\text{as}\quad \varepsilon \to 0,
\end{align}
diverges as the ground state is approached, in contrast to the constant ratio $|\dot{\varepsilon}|/\varepsilon = 2\Delta$ of standard ITE. 

\begin{figure}[t]
\centering
{\includegraphics[width=12cm]{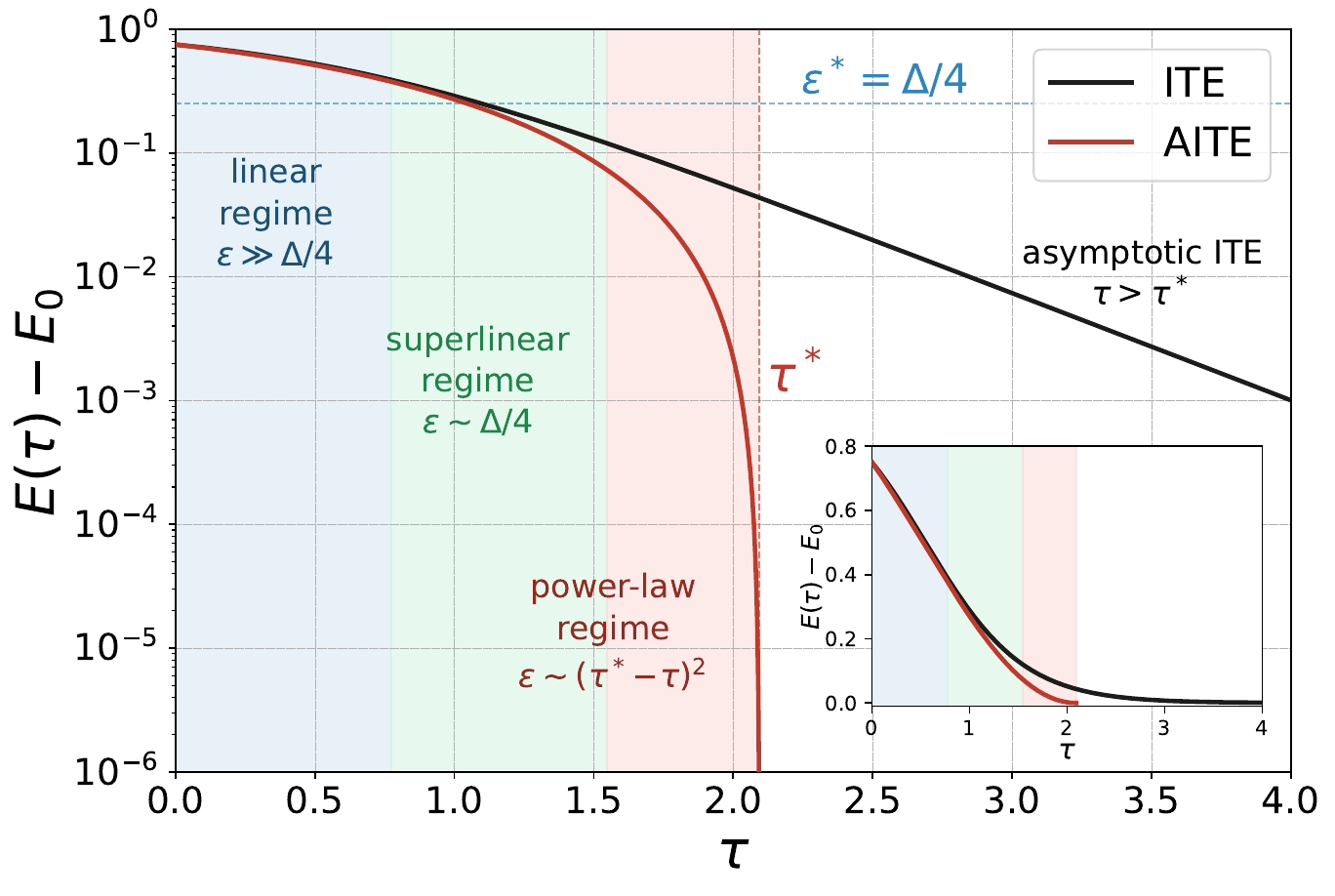}}
\caption{Energy error $E(\tau)-E_{\rm GS}$ as a function of imaginary time $\tau$ for ITE (black) and AITE (red) in the two-level model, Eq.~\eqref{eq:exact_solution}, with $\Delta=1$ and $\varepsilon_0=0.75$. Initially the energy distribution is left-skewed ($\kappa<0$) and the two flows nearly coincide; they are closest at the symmetric point $\varepsilon=\Delta/2$ ($\kappa=0$). Below the crossover $\varepsilon_\times=\tfrac{\Delta}{2}(1-1/\sqrt2)\approx0.15\,\Delta$ the skewness term dominates, and AITE reaches the ground state at the finite time $\tau^*=(2/\Delta)\arcsin\sqrt{\varepsilon_0/\Delta}$ with the power law $\varepsilon\propto(\tau^*-\tau)^2$, whereas ITE decays exponentially.}
\label{figstwolevel}
\end{figure}

Fig.~\ref{figstwolevel} sketches the behavior of the convergence for this toy model. 

\subsection{Finite-time convergence with power-law extinction}

We derive now rigorous estimates for the behavior of the algorithm near convergence. Let $\{(E_j, |E_j\rangle)\}_{j=0}^{N-1}$ be the eigenpairs of the Hamiltonian $\hat{H}$ with $E_0 < E_1 \leq \cdots \leq E_{N-1}$, and write the spectral gap and width as $\Delta_1 = E_1 - E_0$ and $\Delta_{\max} = E_{N-1} - E_0$. At imaginary time $\tau$, the quantum state $\ket{\psi(\tau)}$ is expanded in
the energy eigenbasis as
\begin{equation}
  \ket{\psi(\tau)} = \sum_{j=0}^{N-1} c_j(\tau)\,\ket{E_j},
  \qquad \sum_{j=0}^{N-1} |c_j(\tau)|^2 = 1,
\end{equation}
with populations $p_j(\tau) = |c_j(\tau)|^2 \geq 0$.
The energy error is defined as
\begin{equation}
  \eps(\tau) \;\equiv\; \bra{\psi(\tau)}\hat{H}\ket{\psi(\tau)} - E_0.
\end{equation}
We also define the $k$-th \emph{central moment} of the energy as
\begin{equation}
  \mu_k(\tau) \;=\; \bigl\langle\bigl(\hat{H} - \bra{\psi(\tau)}\hat{H}\ket{\psi(\tau)}\bigr)^k\bigr\rangle_\tau
  \;=\; \sum_{j=0}^{N-1}(E_j - \bra{\psi(\tau)}\hat{H}\ket{\psi(\tau)})^k\,p_j(\tau).
\end{equation}
Note $\mu_1 = 0$, $\mu_2$ is the variance, and $\mu_3$ is the third central moment (related to skewness).

\begin{proposition}[Energy error]\label{prop:eps}
  \begin{equation}
    \eps(\tau) = \sum_{j\geq 1} \Delta_j\,p_j(\tau).
    \label{eq:eps_exact}
  \end{equation}
\end{proposition}

\begin{proof}
Using normalization $\sum_j p_j = 1$,
\begin{align}
  \eps(\tau) 
  = \sum_{j=0}^{N-1} E_j\,p_j - E_0 = \sum_{j=0}^{N-1}(E_j - E_0)\,p_j
  = \sum_{j\geq 1}\Delta_j\,p_j,
\end{align}
since $\Delta_0 = 0$.
\end{proof}

\begin{proposition}[Variance]\label{prop:mu2}
  \begin{equation}
    \mu_2(\tau) = \sum_{j\geq 1}\Delta_j^2\,p_j(\tau) - \eps^2(\tau).
    \label{eq:mu2_exact}
  \end{equation}
\end{proposition}

\begin{proof}
Write $E_j - \expect{\hat{H}} = (E_j - E_0) - \eps = \Delta_j - \eps$.
Then
\begin{align}
  \mu_2 &= \sum_j (\Delta_j - \eps)^2\,p_j
  = \sum_j\Delta_j^2\,p_j
    - 2\eps\underbrace{\sum_j\Delta_j\,p_j}_{=\,\eps}
    + \eps^2\underbrace{\sum_j p_j}_{=\,1}
  = \sum_{j\geq 1}\Delta_j^2\,p_j - \eps^2,
\end{align}
where $\Delta_0 = 0$ drops the $j=0$ term in the first sum.
\end{proof}

\begin{proposition}[Third central moment]\label{prop:mu3}
  \begin{equation}
    \mu_3(\tau) = \sum_{j\geq 1}\Delta_j^3\,p_j(\tau)
      - 3\,\eps(\tau)\,\mu_2(\tau) - \eps^3(\tau).
    \label{eq:mu3_exact}
  \end{equation}
\end{proposition}

\begin{proof}
Expand $(\Delta_j - \eps)^3 = \Delta_j^3 - 3\Delta_j^2\eps + 3\Delta_j\eps^2 - \eps^3$,
then sum against $p_j$:
\begin{align}
  \mu_3 &= \sum_j\Delta_j^3\,p_j
    - 3\eps\sum_j\Delta_j^2\,p_j
    + 3\eps^2\sum_j\Delta_j\,p_j
    - \eps^3 \notag\\
  &= \sum_{j\geq 1}\Delta_j^3\,p_j
    - 3\eps\bigl(\mu_2 + \eps^2\bigr)
    + 3\eps^3
    - \eps^3 \notag\\
  &= \sum_{j\geq 1}\Delta_j^3\,p_j - 3\eps\,\mu_2 - \eps^3,
\end{align}
where in the second line we used $\sum_j \Delta_j^2 p_j = \mu_2 + \eps^2$
(from Eq.~\ref{eq:mu2_exact}) and $\sum_j \Delta_j p_j = \eps$.
\end{proof}

From now on we work in the asymptotic regime in which the quantum state is dominated by its ground-state component. We impose the following assumption:

\medskip
\noindent\textbf{Assumption A (energy-weighted two-level dominance).} Let $\hat{H}$ be a Hamiltonian acting on a finite-dimensional Hilbert space of dimension $N$.
Following the standard spectral-filtering argument underlying imaginary-time
methods~~\cite{10.1063/1.431514,Foulkes2001,hartung2025}, under which excited-state
populations are exponentially suppressed relative to the ground state as
$\tau\to\infty$, we assume that for all $j\geq 2$ and $k=1,2,3$,
\begin{equation}
  \Delta_j^k\,p_j(\tau) \ll \Delta_1^k\,p_1(\tau).
  \label{eq:A1}
\end{equation}
This is the energy-weighted analogue of the standard condition that higher excited states carry negligible population near convergence, and controls each term in the moment sums $\sum_{j\geq 2}\Delta_j^k\,p_j$ individually.

\medskip
\noindent\textbf{Near-convergence.}
We further restrict to $p_1(\tau)\ll 1$, equivalently $\eps(\tau)\ll\Delta_1$. 

\medskip
\noindent\textbf{Two-level normalization.}
Under exact ITE dynamics, the total population outside the two lowest
levels decays as
\begin{equation}
  \sum_{j\geq 2}p_j(\tau)
  \leq (N-2)\,R_{\max}\,e^{-2(E_2-E_1)\tau}
  \xrightarrow{\tau\to\infty} 0,
  \qquad
  R_{\max} \equiv \max_{j\geq 2}\frac{|c_j(0)|^2}{|c_1(0)|^2},
\end{equation}
so for sufficiently large $\tau$ we have $\sum_{j\geq 2}p_j\ll 1$
and normalization gives
\begin{equation}
  p_0 + p_1 \approx 1, \qquad p_0 \approx 1-p_1.
  \label{eq:A2}
\end{equation}
The energy distribution is therefore well approximated by a two-point
Bernoulli on $\{0,\Delta_1\}$ with weight $p_1$ on the excited level.

\medskip
\noindent\textbf{Validity of Assumption~A under ITE.} Under exact ITE dynamics, Assumption~A holds for all $\tau \geq t_k$, where
\begin{equation}
  t_k = \frac{k\ln(\Delta_{\max}/\Delta_1)
    + \ln(R_{\max}/\eta)}{2(E_2-E_1)},
  \label{eq:onset}
\end{equation}
which follows by solving $(\Delta_j/\Delta_1)^k R_j\, e^{-2(E_j - E_1)\tau} \leq \eta$ for $\tau$, for a tolerance $0<\eta\ll1$ and with $R_j\equiv|c_j(0)|^2/|c_1(0)|^2$,

\begin{proposition}[Asymptotic energy error]\label{prop:eps_asymp}
Under Assumption~A,
  \begin{equation}
    \eps(\tau) \approx \Delta_1\,p_1(\tau).
    \label{eq:eps_asymp}
  \end{equation}
\end{proposition}

\begin{proof}
Starting from Eq.~\eqref{eq:eps_exact} and separating the $j=1$ term,
\begin{equation}
  \eps = \Delta_1\,p_1 + \sum_{j\geq 2}\Delta_j\,p_j.
  \label{eq:eps_split}
\end{equation}
By Assumption~A with $k=1$, each term satisfies $\Delta_j\,p_j \ll \Delta_1\,p_1$, i.e., $\Delta_j\,p_j/(\Delta_1\,p_1) = o(1)$ as $\tau\to\infty$ (that is, this  ratio vanishes in the limit, not merely remains bounded), for each fixed $j\geq 2$. 
%Since $N$ is finite and fixed, the sum
%\begin{equation}
%  \frac{\sum_{j\geq 2}\Delta_j\,p_j}{\Delta_1\,p_1}
%  = \sum_{j\geq 2}\frac{\Delta_jp_j}
%  {\Delta_1p_1} % {\color{blue} = \mathrm{O}(N)},
%\end{equation}
%is a sum of a fixed number ($N-2$) of $o(1)$ terms as $\tau\to\infty$, 
%and is therefore itself $o(1)$. Substituting into Eq.~\eqref{eq:eps_split},
%$\eps = \Delta_1\,p_1\,\bigl(1 + o(1)\bigr)
%  \approx \Delta_1\,p_1$, which is Eq.~\eqref{eq:eps_asymp}.
Since $\Delta_j\leq \Delta_{max}$, in the near convergence regime, the probability distribution is such that 
\begin{align}
    \sum_{j\geq 2} p_j \ll p_1 \large({\Delta_1}/{\Delta_{max}}\large)^3,
    \label{prob_assumption}
\end{align}
which is a sufficient condition for claiming the following propositions. The assumption is intuitively valid  near convergence, since the leakage is predominantly to the first excited state and is negligible to all higher excited states. Therefore, 
within the assumption in Eq. \ref{prob_assumption}, one can show that
$\eps = \Delta_1p_1 + \sum_{j\geq 2}\Delta_j\,p_j = \Delta_1 p_1 [1+R] \approx \Delta_1\,p_1$, where $R =  \sum_{j\geq2} \Delta_j p_j/(\Delta_1 p_1) \leq ({\Delta_{max}}/{\Delta_1})\sum_{j\geq2} {p_j}/{p_1} \ll 1$.
\end{proof}

Similarly, one can show that 
$\sum_{j\geq 1}\Delta_j^2\,p_j \approx \Delta_1^2\,p_1$ and $\sum_{j\geq 1}\Delta_j^3\,p_j \approx \Delta_1^3\,p_1$.

\begin{proposition}[Asymptotic variance]\label{prop:mu2_asymp} Under Assumption~A,
  \begin{equation}
    \mu_2(\tau) \approx \Delta_1^2\,p_1(\tau)\bigl(1 - p_1(\tau)\bigr),
    \label{eq:mu2_first}
  \end{equation}
and, additionally using $p_1(\tau)\ll 1$,
  \begin{equation}
    \mu_2(\tau) \approx \Delta_1\eps(\tau).
    \label{eq:mu2_second}
  \end{equation}
\end{proposition}

\begin{proof}
Starting from Eq.~\eqref{eq:mu2_exact}, we split the first sum:
\begin{equation}
  \mu_2 = \sum_{j\geq 1}\Delta_j^2\,p_j - \eps^2
  = \Delta_1^2\,p_1 + \underbrace{\sum_{j\geq 2}\Delta_j^2\,p_j}_{\ll\,\Delta_1^2 p_1} - \eps^2,
\end{equation}
where the remainder is negligible by the strengthened assumption
$\Delta_j^2\,p_j\ll\Delta_1^2\,p_1$ for all $j\geq 2$, giving
\begin{equation}
  \sum_{j\geq 1}\Delta_j^2\,p_j \approx \Delta_1^2\,p_1.
\end{equation}
For the $\eps^2$ term, we use Eq.~\eqref{eq:eps_asymp} to write
$\eps\approx\Delta_1 p_1$, so $\eps^2\approx\Delta_1^2 p_1^2$.
Substituting,
\begin{equation}
  \mu_2 \approx \Delta_1^2\,p_1 - \Delta_1^2\,p_1^2
  = \Delta_1^2\,p_1(1-p_1),
\end{equation}
which is Eq.~\eqref{eq:mu2_first}. For Eq.~\eqref{eq:mu2_second},
we additionally use $p_1\ll 1$ so that $(1-p_1)\approx 1$, giving
\begin{equation}
  \mu_2 \approx \Delta_1^2\,p_1 = \Delta_1(\Delta_1\,p_1)
  \approx \Delta_1\eps,
\end{equation}
which is Eq.~\eqref{eq:mu2_second}.
\end{proof}

\begin{proposition}[Asymptotic third central moment]\label{prop:mu3_asymp} Under Assumption~A,
  \begin{equation}
    \mu_3(\tau) \approx \Delta_1^3\,p_1(\tau)
      \bigl(1-p_1(\tau)\bigr)\bigl(1-2p_1(\tau)\bigr),
    \label{eq:mu3_first}
  \end{equation}
and, additionally using $p_1(\tau)\ll 1$,
  \begin{equation}
    \mu_3(\tau) \approx \Delta_1^2\,\eps(\tau).
    \label{eq:mu3_second}
  \end{equation}
\end{proposition}

\begin{proof}
Starting from Eq.~\eqref{eq:mu3_exact}, we split the cubic sum:
\begin{equation}
  \sum_{j\geq 1}\Delta_j^3\,p_j
  = \Delta_1^3\,p_1 + \underbrace{\sum_{j\geq 2}\Delta_j^3\,p_j}
    _{\ll\,\Delta_1^3\,p_1},
\end{equation}
where the remainder is negligible by the strengthened assumption
$\Delta_j^3\,p_j\ll\Delta_1^3\,p_1$ for all $j\geq 2$, giving
\begin{equation}
  \sum_{j\geq 1}\Delta_j^3\,p_j \approx \Delta_1^3\,p_1.
\end{equation}
Substituting into Eq.~\eqref{eq:mu3_exact} and using
$\eps\approx\Delta_1\,p_1$ (Eq.~\eqref{eq:eps_asymp}) and
$\mu_2\approx\Delta_1^2\,p_1(1-p_1)$ (Eq.~\eqref{eq:mu2_first}):
\begin{align}
  \mu_3
  &\approx \Delta_1^3\,p_1
    - 3(\Delta_1 p_1)\,\Delta_1^2\,p_1(1-p_1)
    - (\Delta_1 p_1)^3 \notag\\
  &= \Delta_1^3\,p_1\bigl[1 - 3p_1(1-p_1) - p_1^2\bigr] \notag\\
  &= \Delta_1^3\,p_1\bigl[1 - 3p_1 + 2p_1^2\bigr] \notag\\
  &= \Delta_1^3\,p_1(1-p_1)(1-2p_1),
\end{align}
where the last step uses $1-3p_1+2p_1^2=(1-p_1)(1-2p_1)$.
This is Eq.~\eqref{eq:mu3_first}. Using additionally $p_1\ll 1$,
so $(1-p_1)\approx 1$ and $(1-2p_1)\approx 1$,
\begin{equation}
  \mu_3 \approx \Delta_1^3\,p_1
  = \Delta_1^2(\Delta_1\,p_1) \approx \Delta_1^2\,\eps,
\end{equation}
which is Eq.~\eqref{eq:mu3_second}.
\end{proof}

Collecting these results, the central moments satisfy the following
hierarchy as $\eps(\tau)\to 0$, under the uniform strengthened
assumption $\Delta_j^k\,p_j(\tau)\ll\Delta_1^k\,p_1(\tau)$ for
$k=1,2,3$ and all $j\geq 2$:
\begin{align}
  \eps(\tau)   &\approx \Delta_1\,p_1(\tau),
  \label{eq:mom1_asymp}\\[4pt]
  \mu_2(\tau)  &\approx \Delta_1^2\,p_1(\tau)\bigl(1-p_1(\tau)\bigr)
                \approx \Delta_1\,\eps(\tau),
  \label{eq:mom2_asymp}\\[4pt]
  \mu_3(\tau)  &\approx \Delta_1^3\,p_1(\tau)\bigl(1-p_1(\tau)\bigr)
                \bigl(1-2p_1(\tau)\bigr)
                \approx \Delta_1^2\,\eps(\tau),
  \label{eq:mom3_asymp}
\end{align}
where in each case the second approximation additionally uses
$p_1(\tau)\ll 1$. These results reveal a clean geometric progression
$\mu_k(\tau)\approx\Delta_1^{k-1}\,\eps(\tau)$ for $k=1,2,3$
(identifying $\mu_1\equiv\eps$), which is the signature of the
two-level Bernoulli structure: as the state approaches the ground
state, the energy distribution converges to a two-point distribution
on $\{0,\Delta_1\}$ with weight $p_1$ on the excited state, for which
the exact $k$-th central moment is
$\Delta_1^k\,p_1(1-p_1)(1-2p_1)\cdots$, reducing to
$\Delta_1^k\,p_1\approx\Delta_1^{k-1}\,\eps$ for small $p_1$.

\begin{remark}[Connection to the ITE ODE]
  The relation $\mu_2\approx\Delta_1\,\eps$ is particularly significant.
  Under imaginary-time evolution, the energy error obeys
  $\dot{\eps}=-2\mu_2$, which in the two-level regime gives $\dot{\eps}\approx-2\Delta_1\eps$. This provides the microscopic derivation of the ITE convergence rate directly from the spectral structure of $\hat{H}$.
\end{remark}

\subsection{AITE and comparison with ITE}

In the asymptotic regime the energy distribution is a two-point (Bernoulli)
distribution on $\{0,\Delta\}$ with weight $p=\eps/\Delta$ on the excited level
($\Delta\equiv\Delta_1$), for which, exactly,
\begin{equation}
\mu_2=\Delta^2p(1-p)=\eps(\Delta-\eps),\qquad
\mu_3=\Delta^3p(1-p)(1-2p)=\eps(\Delta-\eps)(\Delta-2\eps).
\end{equation}
The ITE and AITE energy equations then become
\begin{align}
\dot\eps_{\rm ITE}&=-2\mu_2=-2\eps(\Delta-\eps),
\label{eq:ITE_reduced}\\
\dot\eps_{\rm AITE}&=-\sqrt{4\mu_2^2+\mu_3^2/\mu_2}
=-\Delta\sqrt{\eps(\Delta-\eps)},
\label{eq:new_reduced}
\end{align}
where we used $4\eps(\Delta-\eps)+(\Delta-2\eps)^2=\Delta^2$. Note that the
leading-order moments $\mu_2\simeq\Delta\eps$, $\mu_3\simeq\Delta^2\eps$ fix only the
leading term $\dot\eps_{\rm AITE}\simeq-\Delta^{3/2}\sqrt\eps$; corrections of order
$\eps^2$ must be kept consistently in all moments.

\subsection{The enhancement factor and ITE as a limiting case}

\begin{proposition}[Enhancement factor]\label{prop:enhancement}
The ratio of the two rates is
\begin{equation}
f(\eps)\equiv\frac{\dot\eps_{\rm AITE}}{\dot\eps_{\rm ITE}}
=\frac{\Delta}{2\sqrt{\eps(\Delta-\eps)}}=\sqrt{1+\frac{\kappa^2}{4}}\;\ge\;1,
\label{eq:enhancement}
\end{equation}
with equality only at $\eps=\Delta/2$, where the distribution is symmetric
($\kappa=0$). ITE is therefore the zero-skewness limit of AITE, and
$f\simeq\tfrac12\sqrt{\Delta/\eps}\to\infty$ as $\eps\to0$.
\end{proposition}

\begin{proof}
Divide Eq.~\eqref{eq:new_reduced} by Eq.~\eqref{eq:ITE_reduced}. For the Bernoulli
distribution $\kappa^2=\mu_3^2/\mu_2^3=(\Delta-2\eps)^2/[\eps(\Delta-\eps)]$, so
$1+\kappa^2/4=\Delta^2/[4\eps(\Delta-\eps)]=f^2$. Since
$\eps(\Delta-\eps)\le\Delta^2/4$, $f\ge1$, with equality iff $\eps=\Delta/2$.
\end{proof}

\subsection{Exact solution and finite-time convergence}

\begin{proposition}[Exact solution]\label{prop:exact_solution}
With $\eps(0)=\eps_0\in(0,\Delta)$ and $\theta_0=\arcsin\sqrt{\eps_0/\Delta}$,
\begin{equation}
\eps_{\rm AITE}(\tau)=\Delta\sin^2\!\Big(\theta_0-\frac{\Delta\tau}{2}\Big),\qquad
\eps_{\rm ITE}(\tau)=\frac{\Delta\tan^2\theta_0\,e^{-2\Delta\tau}}
{1+\tan^2\theta_0\,e^{-2\Delta\tau}},
\label{eq:exact_solution}
\end{equation}
in agreement with Eq.~\eqref{eq:energyerror2level}.
\end{proposition}

\begin{proof}
With $\eps=\Delta\sin^2\theta$, $\dot\eps=\Delta\sin2\theta\,\dot\theta$ and
$\sqrt{\eps(\Delta-\eps)}=\tfrac{\Delta}{2}\sin2\theta$, Eq.~\eqref{eq:new_reduced}
gives $\dot\theta=-\Delta/2$, and Eq.~\eqref{eq:ITE_reduced} gives
$\dot\theta=-\tfrac{\Delta}{2}\sin2\theta$, whose solution is
$\tan\theta=\tan\theta_0\,e^{-\Delta\tau}$.
\end{proof}

\begin{corollary}[Finite extinction time]\label{cor:tstar}
The AITE error vanishes exactly at
\begin{equation}
\tau^*=\frac{2}{\Delta}\arcsin\sqrt{\frac{\eps_0}{\Delta}}
\;\le\;\frac{\pi}{\Delta},
\qquad
\tau^*\simeq\frac{2\sqrt{\eps_0}}{\Delta^{3/2}}\quad(\eps_0\ll\Delta),
\label{eq:tstar}
\end{equation}
whereas the ITE error is positive at every finite $\tau$.
\end{corollary}

\begin{corollary}[Power-law extinction]\label{cor:power_law}
Near $\tau^*$,
\begin{equation}
\eps_{\rm AITE}(\tau)\simeq\frac{\Delta^3}{4}(\tau^*-\tau)^2 .
\label{eq:power_law}
\end{equation}
\end{corollary}

\begin{proof}
$\eps=\Delta\sin^2[\Delta(\tau^*-\tau)/2]\simeq\Delta\,[\Delta(\tau^*-\tau)/2]^2$.
\end{proof}

\subsection{Dynamical regimes}

\begin{proposition}[Regimes]\label{prop:crossover}
The enhancement factor~\eqref{eq:enhancement} organizes the flow into three regimes:
\begin{enumerate}
\item \emph{Near-symmetric regime} ($\eps\approx\Delta/2$, $\kappa\approx0$):
$f\approx1$ and AITE follows ITE.
\item \emph{Crossover}: $f=\sqrt2$ at
$\eps_\times=\tfrac{\Delta}{2}\bigl(1-1/\sqrt2\bigr)\approx0.15\,\Delta$, where the
skewness term equals the variance term in Eq.~\eqref{eq:new_reduced}.
\item \emph{Power-law regime} ($\eps\ll\Delta$):
$\dot\eps\simeq-\Delta^{3/2}\sqrt\eps$, with extinction
law~\eqref{eq:power_law}.
\end{enumerate}
In multilevel systems far from convergence, the distribution is typically weakly or
negatively skewed and AITE is again close to ITE; the two-level analysis describes the
approach to convergence.
\end{proposition}

\begin{proposition}[Local convergence exponent]\label{prop:local_exponent}
With $\Lambda(\tau)\equiv|d\ln\eps/d\tau|$, the local exponent
$\alpha(\tau)\equiv\Lambda(\tau)(\tau^*-\tau)\to2$ as $\tau\to\tau^*$, whereas for ITE
$\Lambda\to2\Delta$ and no finite power-law exponent exists.
\end{proposition}

\begin{proof}
From Eq.~\eqref{eq:power_law}, $\ln\eps=2\ln(\tau^*-\tau)+\mathrm{const}$, so
$\Lambda=2/(\tau^*-\tau)+O(1)$.
\end{proof}

\begin{remark}[Convergence classes]
The ITE error is positive at every finite $\tau$, whereas the AITE error vanishes at
$\tau^*<\infty$. In the two-level case both flows traverse the same path, and they are
related by the time change $d\tau_{\rm ITE}/d\tau_{\rm AITE}=f=\sqrt{1+\kappa^2/4}$,
which diverges as $\tau_{\rm AITE}\to\tau^*$. No change of time variable with bounded
derivative can map one flow onto the other: finite-time extinction requires this
singular clock, generated by the skewness. Beyond two levels, the trajectories
themselves differ.
\end{remark}

\begin{table}[t]
\centering
\renewcommand{\arraystretch}{1.5}
\begin{tabular}{lll}
\hline
Property & ITE & AITE \\
\hline
ODE (moments) & $\dot\eps=-2\mu_2$ & $\dot\eps=-\sqrt{4\mu_2^2+\mu_3^2/\mu_2}$ \\
ODE (two-level, exact) & $\dot\eps=-2\eps(\Delta-\eps)$ & $\dot\eps=-\Delta\sqrt{\eps(\Delta-\eps)}$ \\
Solution & $\dfrac{\Delta\tan^2\theta_0\,e^{-2\Delta\tau}}{1+\tan^2\theta_0\,e^{-2\Delta\tau}}$ & $\Delta\sin^2(\theta_0-\Delta\tau/2)$ \\
Convergence time & $\tau^*=\infty$ & $\tau^*=\tfrac{2}{\Delta}\arcsin\sqrt{\eps_0/\Delta}$ \\
Near-zero behavior & $\eps\sim e^{-2\Delta\tau}$ & $\eps\simeq\tfrac{\Delta^3}{4}(\tau^*-\tau)^2$ \\
Enhancement factor & $1$ & $\sqrt{1+\kappa^2/4}=\Delta/[2\sqrt{\eps(\Delta-\eps)}]$ \\
Local exponent $\alpha(\tau)$ & unbounded & $\to2$ as $\tau\to\tau^*$ \\
Convergence class & exponential, asymptotic & power-law, finite-time \\
\hline
\end{tabular}
\caption{ITE and AITE in the two-level regime; $\theta_0=\arcsin\sqrt{\eps_0/\Delta}$.
The two flows coincide at zero skewness ($\eps=\Delta/2$).}
\label{tab:comparison}
\end{table}

\subsection{Iterations}

We estimate the number of iterations $K$ needed to reduce the energy error
$\varepsilon(\tau)\equiv\langle\hat H\rangle_\tau-E_0$ below a target precision
$\varepsilon_*$, for ITE and AITE, with a fixed imaginary-time step $\delta\tau$, so that
$\tau=k\,\delta\tau$ after $k$ steps. Let $\{(E_j,\ket{E_j})\}_{j=0}^{N-1}$ be the
eigenpairs of $\hat H$, with $E_0<E_1\le\dots\le E_{N-1}$, gap $\Delta_1=E_1-E_0$ and
spectral width $\Delta_{\max}=E_{N-1}-E_0$. With $\Delta_j=E_j-E_0$ and
$p_j(\tau)=|\braket{E_j|\psi(\tau)}|^2$,
\begin{equation}
\varepsilon=\sum_{j\ge1}\Delta_j\,p_j,\qquad
\mu_2=\sum_{j\ge1}\Delta_j^2\,p_j-\varepsilon^2 .
\end{equation}
Since $\Delta_1\le\Delta_j\le\Delta_{\max}$ for $j\ge1$, the variance obeys, for any state,
\begin{equation}
\varepsilon\,(\Delta_1-\varepsilon)\;\le\;\mu_2\;\le\;\varepsilon\,(\Delta_{\max}-\varepsilon).
\label{eq:beta_bounds_general}
\end{equation}
The lower bound is the key ingredient of the convergence theorem: it ties the
variance to the energy error through the gap.

\medskip
\noindent\textbf{ITE: step count.}
The ITE energy rate is $\partial_\tau\varepsilon=-2\mu_2$. Using the lower bound in
Eq.~\eqref{eq:beta_bounds_general}, $\dot\varepsilon\le-2\varepsilon(\Delta_1-\varepsilon)$,
so that for $\varepsilon\ll\Delta_1$
\begin{equation}
\varepsilon(\tau)\simeq\varepsilon_0\,e^{-2\Delta_1\tau},
\label{eq:qite_asymptotic}
\end{equation}
the familiar exponential convergence. Setting $\varepsilon(\tau)=\varepsilon_*$ gives
\begin{equation}
\boxed{K_{\rm ITE}(\varepsilon_*)\simeq\frac{1}{2\Delta_1\,\delta\tau}\,
\ln\frac{\varepsilon_0}{\varepsilon_*}
=\mathcal O\!\left(\frac{1}{\Delta_1\delta\tau}\log\frac{1}{\varepsilon_*}\right).}
\label{eq:k_qite}
\end{equation}
Without a gap no such bound exists, since an excited eigenstate has $\mu_2=0$. A
gap-free bound instead uses the ground-state weight: by the Cauchy--Schwarz
inequality, $\mu_2\ge\varepsilon^2p_0/(1-p_0)$, and since ITE does not decrease $p_0$,
$\varepsilon(\tau)\le\varepsilon_0/(1+2c\,\varepsilon_0\tau)$ with $c=p_0(0)/[1-p_0(0)]$,
i.e.\ only algebraic convergence.

\medskip
\noindent\textbf{AITE: energy rate.}
Restricted to the Krylov subspace $\mathcal S_H(\tau)$, the AITE energy rate is
\begin{equation}
\partial_\tau\varepsilon=-\sqrt{4\mu_2^2+\gamma^2},\qquad \gamma=\mu_3/\sqrt{\mu_2},
\label{eq:aite_rate}
\end{equation}
with $\gamma$ the normalized third central moment of Eq.~\eqref{eq:matrix}. For a
two-level state with ground-state weight $1-p$, one has $\varepsilon=\Delta_1p$,
$\mu_2=\Delta_1^2p(1-p)$ and $\mu_3=\Delta_1^3p(1-p)(1-2p)$, and
Eq.~\eqref{eq:aite_rate} reduces exactly to
\begin{equation}
\partial_\tau\varepsilon=-\Delta_1\sqrt{\varepsilon(\Delta_1-\varepsilon)}
\;\simeq\;-\Delta_1^{3/2}\sqrt{\varepsilon}\qquad(\varepsilon\ll\Delta_1).
\label{eq:power_law_ode}
\end{equation}
Integrating from $\varepsilon_0<\Delta_1$ to $\varepsilon=0$ gives the exact two-level
extinction time
\begin{equation}
\tau^*=\frac{2}{\Delta_1}\arcsin\sqrt{\frac{\varepsilon_0}{\Delta_1}}
\;\simeq\;\frac{2\sqrt{\varepsilon_0}}{\Delta_1^{3/2}}\qquad(\varepsilon_0\ll\Delta_1).
\label{eq:tau_star_aite}
\end{equation}
For general spectra, the convergence theorem provides a rigorous version of
Eq.~\eqref{eq:power_law_ode}: $\dot\varepsilon\le-\tfrac{5}{16}\Delta_1^{3/2}\sqrt{\varepsilon}$
for all $\varepsilon\le\Delta_1/4$, so that the ground state is reached within flow time
$\tfrac{32}{5}\sqrt{\varepsilon}/\Delta_1^{3/2}\le16/(5\Delta_1)$ of entering this basin.

\medskip
\noindent\textbf{AITE: step count.}
Combining the two phases, entry into the basin $\varepsilon\le\Delta_1/4$ (in flow
time at most $\tfrac{1}{2\Delta_1}\ln\tfrac{3\varepsilon_0}{\Delta_1-\varepsilon_0}$ for
$\varepsilon_0<\Delta_1$) and extinction within it, gives
\begin{equation}
\boxed{K_{\rm AITE}(\varepsilon_*)
=\mathcal O\!\left(\frac{1}{\Delta_1\delta\tau}\Big[\log\frac{\varepsilon_0}{\Delta_1}+1\Big]\right),}
\label{eq:k_aite}
\end{equation}
independent of the target precision $\varepsilon_*$. Compared with
Eq.~\eqref{eq:k_qite}, AITE removes the $\log(\Delta_1/\varepsilon_*)$ contribution of
the final approach: the reduction in step count is not a fixed factor, but grows
logarithmically with the target precision. Both methods retain the same
$1/\Delta_1$ dependence.

\medskip
\noindent\textbf{Remarks.}
Three points qualify these estimates.
\begin{enumerate}
\item \emph{Steps are not cost.} Equation~\eqref{eq:k_aite} counts steps at fixed
$\delta\tau$. Each AITE step becomes more expensive as $\mu_2\to0$, because the
coefficient of $\bar H^2$ in the augmented Hamiltonian grows as $1/\sqrt{\mu_2}$, and
the moments $\mu_2,\mu_3$ must be estimated at every step. The relevant comparison of
resources, in Hamiltonian applications, is given in the next section.
\item \emph{Initial overlap.} Entering the basin from a state with small ground-state
weight $p_0$ requires imaginary time of order $\Delta_1^{-1}\log(1/p_0)$, which grows
only linearly with system size $n$ even if $p_0\sim e^{-n}$. The exponential cost
appears in the implementation: realizing a non-unitary evolution on a quantum
computer requires post-selection or amplitude amplification, with a cost of order
$1/\sqrt{p_0}$.
\item \emph{Complexity.} Estimating ground-state energies of local Hamiltonians is
QMA-hard in general~\cite{Kempe2006}, while with an initial state of
inverse-polynomial overlap the problem becomes BQP-complete~\cite{GharibianLeGall2022}.
The finite extinction time of AITE is consistent with this: the complexity remains
in the gap dependence, in the cost of entering the basin, and in the preparation of
a state with sufficient overlap. AITE does not change the $1/\Delta_1$ scaling; it
removes the precision-dependent part of the step count.
\end{enumerate}

\section{Cost model}
\label{sec:S_costmodel}
The imaginary time $\tau$ is a flow parameter, not a computational resource. We therefore quantify the cost of ITE and AITE with three measures, each appropriate to a different implementation:
\begin{enumerate}
\item \emph{Discretization steps}, relevant when each step is a separate circuit or update.
\item \emph{Hamiltonian applications} (matrix--vector products) needed to apply the propagators with an adaptive Krylov integrator at fixed accuracy, the standard cost measure for propagation methods and the classical analogue of query complexity (Sec.~\ref{sec:S_stepwise}).
\item \emph{Queries to $\hat H$} in a fault-tolerant implementation, in which a spectral filter is realized by quantum singular value transformation (QSVT), quantum eigenvalue transformation of unitary matrices (QETU), or a linear combination of unitaries (LCU), followed by amplitude amplification (Secs.~\ref{sec:S_filter}--\ref{sec:S_benchmarks}).
\end{enumerate}
In all cases the cost of estimating the moments $\mu_2$ and $\mu_3$ is a measurement cost common to all AITE variants and is discussed separately.

\subsection{The integrated AITE flow is a Gaussian filter}
\label{sec:S_gauss}
The augmented Hamiltonian $\hat H_{\rm A}=\cos2\phi\,\Hb+\sin2\phi\,\Hb^2/(2\sqrt{\mu_2})$ is a function of $\Hh$ alone. Every AITE step therefore multiplies each energy eigencomponent by a scalar function of its energy, and the whole trajectory collapses to a single spectral filter,
\begin{equation}
\ket{\psi(\tau)}\propto e^{-B(\tau)[\Hh-c(\tau)]^2}\ket{\psi_0},
\label{eq:S_gauss}
\end{equation}
with
\begin{equation}
B(\tau)=\int_0^\tau\frac{\sin2\phi}{2\sqrt{\mu_2}}\,dt,\qquad
c(\tau)=\frac{1}{B(\tau)}\left[\int_0^\tau\frac{\sin2\phi}{2\sqrt{\mu_2}}\,E\,dt-\frac12\int_0^\tau\cos2\phi\,dt\right].
\label{eq:S_Bc}
\end{equation}
We verified numerically that the one-shot filter~\eqref{eq:S_gauss} reproduces the stepwise AITE state at every step to numerical precision. Two consequences follow.
\begin{itemize}
\item \emph{Finite extinction} is the divergence of the sharpness: $B(\tau)\to\infty$ as $\tau\to\tau^*$, and the filter collapses onto the ground state. In the two-level regime AITE follows the same trajectory as ITE on a singular time scale: $d\tau_{\rm ITE}/d\tau_{\rm AITE}=\sqrt{1+\kappa^2/4}\to\infty$. No regular reparametrization of the exponential filter $e^{-\tau\Hh}$ can reach the ground state at finite $\tau$, since $e^{-\tau x}$ never vanishes on the excited energies.
\item \emph{Implementation.} The Gaussian can be applied in one shot by QSVT or QETU, or as a positive-weight average of real-time evolutions,
\begin{equation}
e^{-B(\Hh-c)^2}=\frac{1}{\sqrt{4\pi B}}\int_{-\infty}^{\infty}e^{-t^2/4B}\,e^{-i(\Hh-c)t}\,dt ,
\label{eq:S_lcu}
\end{equation}
which requires only controlled Hamiltonian evolution and is suited to early fault-tolerant hardware.
\end{itemize}

\subsection{Query cost of bounded spectral filters}
\label{sec:S_filter}
We use the standard query model for filter-based ground-state preparation~\cite{GeTuraCirac2019,LinTong2020,Dong2022}. A filter $F$ implemented by QSVT, QETU or an LCU must satisfy $|F(x)|\le1$ on the whole block-encoding interval $I=[L,E_{\max}]$, where $L$ is a certified lower bound on the spectrum. The filtered state is obtained with success amplitude $\|F(\Hh)\ket{\psi_0}\|$ and amplitude amplification requires
\begin{equation}
R\simeq\frac{1}{\|F(\Hh)\ket{\psi_0}\|}\simeq\frac{1}{\gamma\,F(E_0)}
\end{equation}
rounds, where $\zeta=|\langle E_0|\psi_0\rangle|$ is the initial overlap (we reserve $\gamma$ for the normalized third moment of Eq.~\eqref{eq:matrix}). With $d$ the degree of the polynomial approximating $F$ on $I$, the total number of queries to $\Hh$ is $Q=d\times R$.

We make the dependence on prior knowledge explicit through the error $\delta\ge0$ of the lower bound, $L=E_0-\delta$, and write $r=\delta/\Delta$ and $L_\eta=\ln[\sqrt{1-\gamma^2}/(\gamma\sqrt\eta)]$ for a target infidelity $\eta$.

\paragraph{ITE filter.} $F(x)=e^{-\tau(x-L)}$ peaks at $L$, so the ground state lies on its slope and $F(E_0)=e^{-\tau\delta}$. With $\tau=L_\eta/\Delta$,
\begin{equation}
R_{\rm ITE}\simeq\gamma^{-1}e^{\,r L_\eta},\qquad
Q_{\rm ITE}\simeq\frac{L_\eta}{\gamma}\sqrt{\frac{1+r}{\Delta}}\;e^{\,rL_\eta}.
\label{eq:S_Qite}
\end{equation}
The loss of ground-state amplitude is first order in $\delta$. The degree $d\propto\Delta^{-1/2}$ reflects Chebyshev edge acceleration: the filter varies fastest at the edge of $I$, where Chebyshev polynomials resolve best.

\paragraph{Gaussian (AITE) filter.} $F(x)=e^{-B(x-c)^2}$ peaks at its centre, so a centre displaced by $\delta$ loses amplitude only at second order. For a Gaussian centred at the same bound, $c=L$, with $B=L_\eta/[\Delta^2(1+2r)]$,
\begin{equation}
R_{\rm Gauss}\simeq\gamma^{-1}e^{\,[r^2/(1+2r)]L_\eta},
\label{eq:S_Rgauss}
\end{equation}
quadratic in $r$ for small errors and never more than half the ITE exponent. For a centred Gaussian, $Q_{\rm AITE}\simeq L_\eta/(\gamma\Delta)$, matching near-optimal filtering~\cite{LinTong2020}.

\paragraph{Which is cheaper.} Including the degrees, $d_{\rm Gauss}/d_{\rm ITE}\simeq\tfrac12\sqrt{W/\Delta}$ with $W=E_{\max}-L$, and
\begin{equation}
\frac{Q_{\rm Gauss}}{Q_{\rm ITE}}\simeq\tfrac12\sqrt{\frac{W}{\Delta}}\,(1+2r)^{-1/2}\exp\!\left[-L_\eta\,\frac{r(1+r)}{1+2r}\right].
\label{eq:S_ratio}
\end{equation}
The prefactor is a polynomial degree penalty for the Gaussian; the exponential is its gain in amplitude amplification. Hence:
\begin{itemize}
\item if $E_0$ is known exactly ($\delta=0$), the ITE filter is cheaper by a constant factor;
\item the Gaussian becomes cheaper beyond
\begin{equation}
\delta^*\simeq\frac{\Delta\,\ln(W/4\Delta)}{2L_\eta},
\label{eq:S_crossover}
\end{equation}
which is $0.1$--$0.2\,\Delta$ in our systems, and its advantage then grows exponentially with $L_\eta\,\delta/\Delta$.
\end{itemize}
Both filters belong to the same asymptotic class, $O(\log(1/\eta)/(\gamma\Delta))$ up to constant factors, when each is given ideal prior knowledge: the Gaussian does not change the scaling, but it removes the exponential penalty that the ITE filter pays for imperfect knowledge of $E_0$.

\paragraph{The filter generated by the AITE flow.} The Gaussian obtained by integrating the AITE flow, Eq.~\eqref{eq:S_Bc}, requires no bound on $E_0$: its centre is generated by the dynamics and lands within $0.003$--$0.02\,\Delta$ of $E_0$ in our systems. Its amplitude-amplification factor therefore stays at the ideal value $R\simeq1/\gamma$ for any $\delta$; only its degree grows with $\delta$, because the centre moves away from the edge of the interval $I$, where Chebyshev polynomials resolve best (the degree grows by at most a factor of about three and saturates once the centre lies well inside $I$).

\subsection{Numerical benchmarks of the query cost}
\label{sec:S_benchmarks}
We computed $Q=d\times R$ exactly for hydrogen chains from the Hartree--Fock state. The energy distribution of $\ket{\psi_0}$ is obtained from a Lanczos quadrature (300 steps, converged to 12 digits); each filter is rescaled to $\max_IF=1$; the filter parameter ($\tau$, $B$ or the sharpness $k$ of the step filter) is the smallest value reaching an energy error of $10^{-8}$~Ha; $d$ is the smallest Chebyshev degree approximating $F$ on $I$ within $0.1\sqrt{10^{-8}/W}\,\|F(\Hh)\ket{\psi_0}\|$. All filters receive the same prior knowledge ($L=E_0-\delta$ and the gap), and the same accuracy rule, so constant factors are comparable across filters (none of the polynomials is minimax-optimal). We compare four filters: ITE anchored at $L$; a Gaussian centred at $L$; the step filter of Lin and Tong~\cite{LinTong2020}, $\tfrac12\mathrm{erfc}[k(x-L-\Delta/2)]$; and the Gaussian generated by the AITE flow (``full AITE'').

\begin{table}[H]
\centering\small
\begin{tabular}{@{}llcccc@{}}
\toprule
system & $\delta/\Delta$ & ITE & Gaussian at $L$ & step (Lin--Tong) & full AITE\\
\midrule
H$_6$, $R=0.9$~\AA{} ($\gamma^2=0.93$) & 0 & \textbf{32} & 51 & 52 & 61\\
 & 0.5 & 158 & \textbf{61} & 73 & 70\\
 & 1 & 795 & 95 & 128 & \textbf{79}\\
 & 2 & 18\,800 & 320 & 488 & \textbf{94}\\
\midrule
H$_8$, $R=0.9$~\AA{} ($\gamma^2=0.90$) & 0 & \textbf{44} & 73 & 73 & 84\\
 & 0.5 & 220 & \textbf{86} & 102 & 95\\
 & 1 & 1\,090 & 133 & 177 & \textbf{107}\\
 & 2 & 25\,900 & 463 & 710 & \textbf{127}\\
\midrule
H$_8$, $R=1.4$~\AA{} ($\gamma^2=0.62$) & 0 & \textbf{70} & 117 & 116 & 131\\
 & 0.5 & 294 & \textbf{133} & 159 & 146\\
 & 1 & 1\,210 & 193 & 261 & \textbf{160}\\
 & 2 & 20\,500 & 561 & 872 & \textbf{186}\\
\midrule
H$_8$, $R=2.0$~\AA{} ($\gamma^2=0.21$) & 0 & \textbf{238} & 411 & 400 & 424\\
 & 0.5 & 814 & \textbf{450} & 539 & 466\\
 & 1 & 2\,750 & 607 & 820 & \textbf{496}\\
 & 2 & 31\,200 & 1\,480 & 2\,340 & \textbf{573}\\
\bottomrule
\end{tabular}
\caption{Queries $Q=d\times R$ to $\Hh$ needed to reach an energy error of $10^{-8}$~Ha, with the same prior lower bound $L=E_0-\delta$ for all filters. Bold marks the cheapest filter in each row.}
\label{tab:S_queries}
\end{table}

\begin{figure}[t!]
\centering
\includegraphics[width=0.92\linewidth]{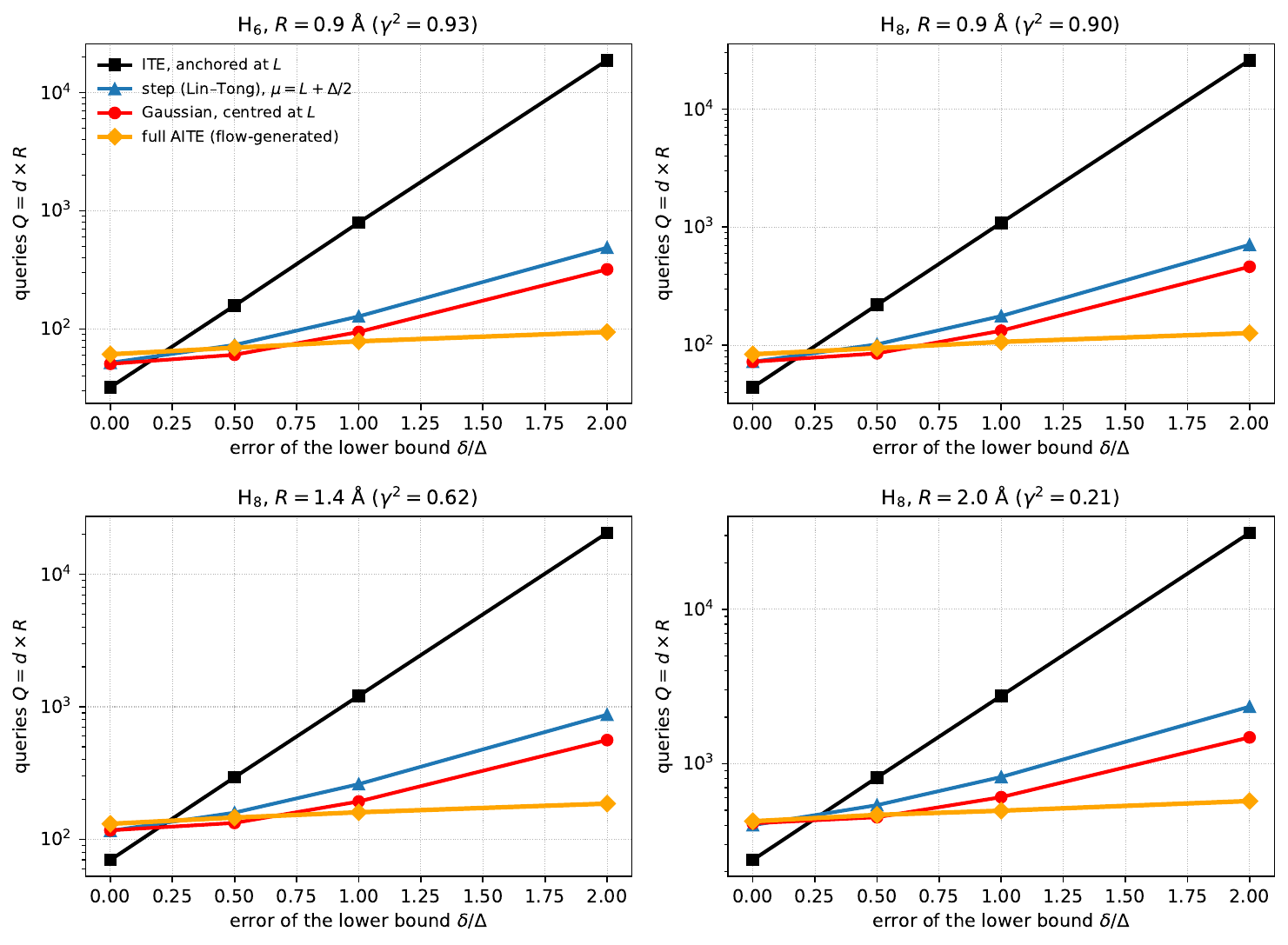}
\caption{Data of Table~\ref{tab:S_queries}: queries to $\Hh$ against the error $\delta$ of the lower bound on $E_0$, in units of the gap.}
\label{fig:S_queries}
\end{figure}

\noindent The results confirm the analysis of Sec.~\ref{sec:S_filter}:
\begin{itemize}
\item with $E_0$ known exactly, the ITE filter is the cheapest, by $1.6$--$1.9\times$;
\item already at $\delta=0.5\Delta$ the Gaussian filters are cheaper, and at $\delta=\Delta$ the full-AITE filter is the cheapest of the four in all systems: $5$--$10\times$ cheaper than ITE and $1.6$--$1.7\times$ cheaper than the step filter;
\item at $\delta=2\Delta$ the full-AITE filter is $54$--$200\times$ cheaper than ITE and $4$--$5.6\times$ cheaper than the step filter;
\item the cost of the full-AITE filter grows by only 35--55\% from $\delta=0$ to $2\Delta$, against factors of $6$--$10$ for the step filter and $130$--$590$ for ITE, including the strongly correlated regime ($R=2.0$~\AA, $\gamma^2=0.21$).
\end{itemize}
The comparison with the step filter excludes the energy-estimation stage that Lin--Tong-type methods use to place their threshold, and the full-AITE filter excludes the estimation of $\mu_2,\mu_3$ along the trajectory; constant factors of order $1.1$--$1.5$ could shift with optimized polynomial constructions for all filters.

\subsection{Stepwise resources: Hamiltonian applications}
\label{sec:S_stepwise}
For stepwise implementations we count the Hamiltonian applications needed to apply the propagators, using the step sizes of the main text and the same adaptive Krylov (Lanczos) integrator on $\Hh$ for both methods, converged to $10^{-10}$. Because $\hat H_{\rm A}$ is a function of $\Hh$, each AITE propagator is evaluated as a function of $\Hh$ on the Krylov space of $\Hh$ itself. This matters: a Krylov expansion in $\hat H_{\rm A}$ inherits the slow convergence of folded-spectrum methods, because the $\Hb^2$ term folds the spectrum around the instantaneous energy; for the same trajectory it costs about $3$--$4\times$ more Hamiltonian applications.

\begin{table}[H]
\centering\small
\begin{tabular}{@{}lcccccc@{}}
\toprule
 & \multicolumn{3}{c}{steps to $10^{-8}$~Ha} & \multicolumn{3}{c}{Hamiltonian applications to $10^{-10}$~Ha}\\
\cmidrule(lr){2-4}\cmidrule(lr){5-7}
method & $R=0.9$ & $1.4$ & $1.9$ & $R=0.9$ & $1.4$ & $1.9$\\
\midrule
ITE & 89 & 201 & 369 & 628 & 1\,265 & 2\,217\\
AITE & \textbf{9} & \textbf{24} & \textbf{45} & \textbf{155} & \textbf{384} & 918\\
AITE ($\kappa_{\rm MF}$) & 9 & 30 & 66 & 165 & 554 & 2\,004\\
ITE $\to$ AITE, switch at $\kappa\ge30$ & 22 & 61 & 132 & 207 & 455 & \textbf{864}\\
\bottomrule
\end{tabular}
\caption{Stepwise resources for H$_8$ (step sizes of the main text; estimation of $\mu_2,\mu_3$ not included).}
\label{tab:S_stepwise}
\end{table}

We define the gain $G(\varepsilon)=N_H^{\rm ITE}(\varepsilon)/N_H(\varepsilon)$, the ratio of Hamiltonian applications needed to first reach an energy error $\varepsilon$ (Fig.~\ref{fig:S_gain}). Each AITE step becomes more expensive as $\mu_2\to0$, because the coefficient of $\Hb^2$ grows as $1/\sqrt{\mu_2}$; far from convergence ITE is therefore the cheaper method. Close to the ground state the balance reverses: ITE pays a fixed, gap-limited cost per decade of precision, whereas AITE covers many decades in a few steps. The gain exceeds one below $\varepsilon\sim10^{-3}$--$10^{-4}$~Ha and grows with the target precision, reaching $2.4$--$4.1$ at $10^{-10}$~Ha.

A practical switching criterion follows from the skewness itself. In the two-level regime $\kappa\simeq\sqrt{\Delta/\varepsilon}$, so a large measured skewness signals proximity to the ground state relative to the gap, without any knowledge of $E_{\rm GS}$. Switching from ITE to AITE once $\kappa\ge30$ occurs at $\varepsilon\simeq0.003\,\Delta$ for all three geometries and yields gains of $2.6$--$3.0$ at $10^{-10}$~Ha; it is the most efficient practical choice at strong correlation. (Small thresholds such as $\kappa\ge3$ are unreliable: far from convergence, high excited states inflate $\kappa$.)

\begin{figure}[t!]
\centering
\includegraphics[width=\linewidth]{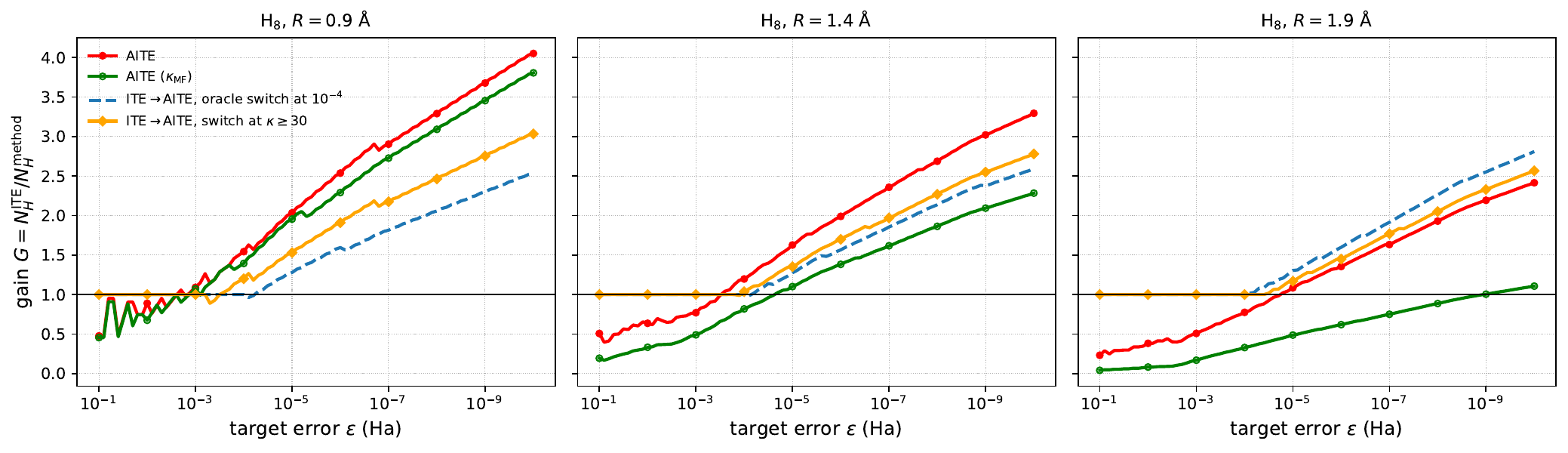}
\caption{Gain $G(\varepsilon)=N_H^{\rm ITE}(\varepsilon)/N_H(\varepsilon)$ in Hamiltonian applications for H$_8$ at $R=0.9$, $1.4$ and $1.9$~\AA, for AITE, AITE with $\kappa_{\rm MF}$, and ITE followed by AITE (switch at $10^{-4}$, oracle; switch at $\kappa\ge30$). $G>1$ indicates fewer applications than ITE.}
\label{fig:S_gain}
\end{figure}

\begin{figure}[t!]
\centering
\includegraphics[width=0.8\linewidth]{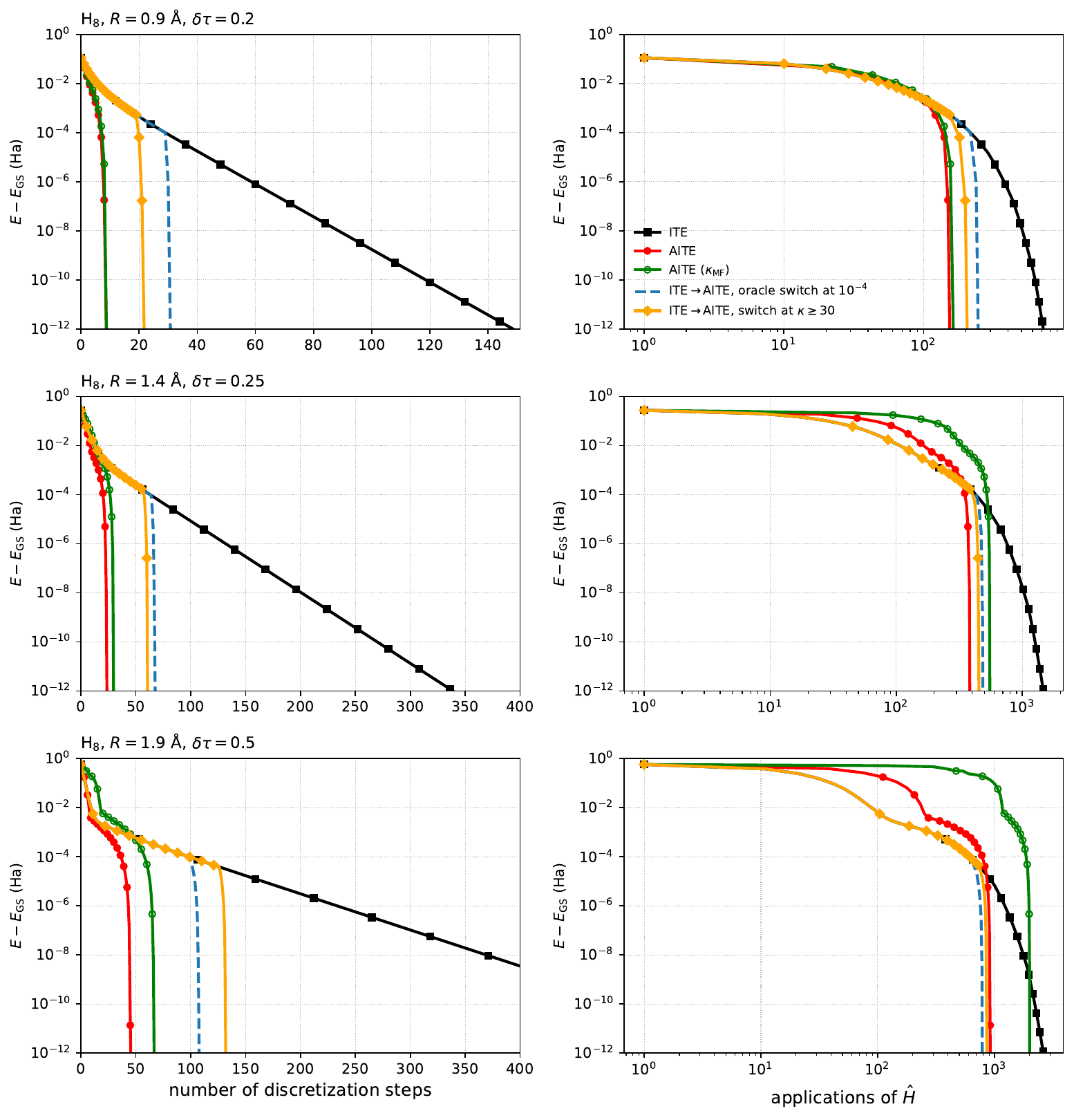}
\caption{Energy error against the number of discretization steps (left) and of Hamiltonian applications (right) for H$_8$ at $R=0.9$, $1.4$ and $1.9$~\AA.}
\label{fig:S_steps}
\end{figure}

\noindent These comparisons refer to the two flows under the same discretization. A one-shot application of $e^{-\tau\Hh}$ or a Lanczos calculation is cheaper classically (about 40 and 18 Hamiltonian applications for H$_8$ at $R=0.9$~\AA{}); the relevant quantum cost is the filter cost of Secs.~\ref{sec:S_filter}--\ref{sec:S_benchmarks}.

\end{document}